\documentclass[sigconf, nonacm]{acmart}

\newcommand\vldbdoi{XX.XX/XXX.XX}
\newcommand\vldbpages{XXX-XXX}
\newcommand\vldbvolume{14}
\newcommand\vldbissue{1}
\newcommand\vldbyear{2020}
\newcommand\vldbauthors{\authors}
\newcommand\vldbtitle{\shorttitle} 
\newcommand\vldbavailabilityurl{https://github.com/VIDA-NYU/SamplingMethodsForInnerProductSketching}
\newcommand\vldbpagestyle{plain} 

\usepackage{colortbl}
\usepackage{graphicx}
\usepackage{subcaption}
\usepackage{algorithmicx}
\usepackage{algorithm}
\usepackage[noend]{algpseudocode}
\usepackage{amsthm}
\usepackage{amsmath}
\usepackage{bbm}
\usepackage{xcolor}
\usepackage{hyperref}
\usepackage{enumitem}
\usepackage{xspace}
\usepackage[capitalize, nameinlink]{cleveref}

\newcommand{\algrule}[1][.2pt]{\par\vskip.2\baselineskip\hrule height #1\par\vskip.2\baselineskip}
\algrenewcommand\algorithmicrequire{\textbf{Input:}}
\algrenewcommand\algorithmicensure{\textbf{Output:}}

\usepackage{pifont}% http://ctan.org/pkg/pifont
\newcommand{\cmark}{\ding{51}}%
\newcommand{\xmark}{\ding{55}}%

\def\withnotes{0} % change to (0) to hide all comments, (1) to show all comments

\ifnum\withnotes=1
    \newcommand{\chris}[1]{\textcolor{violet}{Chris: #1}}
    \newcommand{\majid}[1]{\textcolor{brown}{Majid: #1}}
    \newcommand{\juliana}[1]{\textcolor{red}{Juliana: #1}}
    \newcommand{\jf}[1]{\textcolor{red}{Juliana: #1}}
    \newcommand{\as}[1]{\textcolor{pink}{Aecio: #1}}
    \newcommand{\haoxiang}[1]{\textcolor{olive}{Haoxiang: #1}}
\else
  \newcommand{\chris}[1]{\xspace}
  \newcommand{\majid}[1]{\xspace}
  \newcommand{\juliana}[1]{\xspace}
  \newcommand{\jf}[1]{\xspace}
  \newcommand{\as}[1]{\xspace}
  \newcommand{\haoxiang}[1]{\xspace}
\fi

\definecolor{HighlightColor}{rgb}{0.05,0.05,0.70}
\def\submission{0}

\ifnum\submission=1
    \newcommand{\subm}[1]{#1}
    \newcommand{\arxiv}[1]{\xspace}
    \newcommand{\revised}[1]{{\color{HighlightColor}#1\xspace}}
\else
    \newcommand{\arxiv}[1]{#1}
    \newcommand{\subm}[1]{\xspace}
    \newcommand{\revised}[1]{{#1\xspace}}
\fi

\usepackage{verbatim}
\usepackage{balance}

  \usepackage{nth}
  \usepackage{intcalc}

  \newcommand{\cAAAI}[1]{AAAI\ Conference\ on\ Artificial (AAAI)}

\newcommand{\bs}[1]{\boldsymbol{#1}}
\newcommand{\bv}[1]{\mathbf{#1}}
\newcommand{\Var}{\operatorname{Var}}

\DeclareMathOperator*{\E}{\mathbb{E}}
\DeclareMathOperator*{\R}{\mathbb{R}}

\newtheorem{theorem}{Theorem}

\newtheorem{corollary}[theorem]{Corollary}

\newtheorem{lemma}[theorem]{Lemma}

\newcommand{\thresholdsampling}{Threshold Sampling\xspace}

\newcommand{\thresholdandpriority}{Threshold and Priority Sampling\xspace}
\newcommand{\TSuniform}{TS-uniform\xspace}
\newcommand{\TSweighted}{TS-weighted\xspace}

\newcommand{\prioritysampling}{Priority Sampling\xspace}

\newcommand{\PSuniform}{PS-uniform\xspace}
\newcommand{\PSweighted}{PS-weighted\xspace}

\newcommand{\dartmh}{DartMinHash\xspace}

\newcommand{\wmh}{WMH\xspace}
\newcommand{\wmhsexperiments}{MH-weighted\xspace}
\newcommand{\runtime}{run-time\xspace}

\newcommand{\myparagraph}[1]{\smallskip\noindent{\textbf{#1.}}}
\newcommand{\hide}[1]{}

\begin{document}
	\title{Sampling Methods for Inner Product Sketching}

% Chris: I don't think we should make this change to a single line unless it's needed.
	 \author{Majid Daliri, Juliana Freire, Christopher Musco, A\'ecio Santos, Haoxiang Zhang}
	 \affiliation{%
		 	\institution{New York University}
		 }
	 \email{{daliri.majid,juliana.freire, cmusco,aecio.santos, haoxiang.zhang}@nyu.edu}
%   \author{Majid Daliri}
% 	 \affiliation{%
% 		 	\institution{New York University}
% 		 }
% 	 \email{daliri.majid@nyu.edu}
  
% 	 \author{Juliana Freire}
% 	 \affiliation{%
% 		 	\institution{New York University}
% 		 }
% 	 \email{juliana.freire@nyu.edu}
% 	 \author{Christopher Musco}
% 	 \affiliation{%
% 		 	\institution{New York University}
% 		 }
% 	 \email{cmusco@nyu.edu}
% 	 \author{A\'{e}cio Santos}
% 	 \affiliation{%
% 		 	\institution{New York University}
% 		 }
% 	 \email{aecio.santos@nyu.edu}
% 	 \author{Haoxiang Zhang}
% 	 \affiliation{%
% 		 	\institution{New York University}
% 		 }
% 	 \email{haoxiang.zhang@nyu.edu}
	
	\begin{abstract}
    Recently, Bessa et al. (PODS 2023) showed that sketches based on coordinated weighted sampling theoretically and empirically outperform popular linear sketching methods like Johnson-Lindentrauss projection and CountSketch for the ubiquitous problem of \revised{inner product estimation. 
	% Despite decades of literature on such sampling methods, this observation seems to have been overlooked. 
	We further develop this finding by introducing} and analyzing two alternative sampling-based \revised{methods}. In contrast to the computationally expensive algorithm in Bessa et al., our methods run in linear time (to compute the sketch) and perform better in practice, significantly beating linear sketching on a variety of tasks. 
    For example, they provide state-of-the-art results for estimating the correlation between columns in unjoined tables, a problem that we show how to reduce to inner product estimation in a black-box way. While based on known sampling techniques (threshold and priority sampling) we introduce significant new theoretical 
    analysis to prove approximation guarantees for our methods. 
	\end{abstract}
	
	\maketitle
	
	\pagestyle{\vldbpagestyle}
	\begingroup\small\noindent\raggedright\textbf{PVLDB Reference Format:}\\
	\vldbauthors. \vldbtitle. PVLDB, \vldbvolume(\vldbissue): \vldbpages, \vldbyear.\\
	\href{https://doi.org/\vldbdoi}{doi:\vldbdoi}
	\endgroup
	\begingroup
	\renewcommand\thefootnote{}\footnote{\noindent
		This work is licensed under the Creative Commons BY-NC-ND 4.0 International License. Visit \url{https://creativecommons.org/licenses/by-nc-nd/4.0/} to view a copy of this license. For any use beyond those covered by this license, obtain permission by emailing \href{mailto:info@vldb.org}{info@vldb.org}. Copyright is held by the owner/author(s). Publication rights licensed to the VLDB Endowment. \\
		\raggedright Proceedings of the VLDB Endowment, Vol. \vldbvolume, No. \vldbissue\ %
		ISSN 2150-8097. \\
		\href{https://doi.org/\vldbdoi}{doi:\vldbdoi} \\
	}\addtocounter{footnote}{-1}\endgroup
	
	\ifdefempty{\vldbavailabilityurl}{}{
		\vspace{.3cm}
		\begingroup\small\noindent\raggedright\textbf{PVLDB Artifact Availability:}\\
		The source code, data, and/or other artifacts have been made available at \url{\vldbavailabilityurl}.
		\endgroup
	}
	\vspace{-1.5em}
\section{Introduction}
\label{sec:intro}
We study methods for approximating the inner product $\langle \bv{a}, \bv{b}\rangle = \sum_{i=1}^n \bv{a}_i\bv{b}_i$ between two length $n$ vectors $\bv{a}$ and $\bv{b}$. \revised{We} are interested in algorithms that {independently} compute compact \emph{sketches} $\mathcal{S}(\bv{a})$ and $\mathcal{S}(\bv{b})$ of $\bv{a}$ and $\bv{b}$,
% (possibly using a shared random seed) 
and approximate $\langle \bv{a}, \bv{b}\rangle$ using only the information in these sketches. $\mathcal{S}(\bv{a})$ and $\mathcal{S}(\bv{b})$ should take much less than $n$ space to store, allowing them to be quickly retrieved from disk or transferred over a network.
%(e.g., if $\bv{a}$ and $\bv{b}$ reside on different machines in a network). 
Additionally, both the sketching procedure $\bv{a} \rightarrow \mathcal{S}(\bv{a})$ and the estimation procedure that returns an approximation to $\langle \bv{a},\bv{b}\rangle$ should be computationally efficient\revised{, ideally running in linear time.}
%\footnote{
\revised{We note that computing an inner product between two length $n$ vectors naively takes just $O(n)$ time. As such, the goal of sketching methods is not to speed up a single inner product, but rather to speed up many. For example, the methods we study can compute sketches of size $m$ for a collection of $D$, length $n$ vectors in $O(nD)$ time. We can then estimate all pairwise inner products between those vectors in $O(D^2 m)$ time, which is significantly faster than the baseline $O(D^2 n)$ time when $m \ll n$.}
%}
%
Sketching methods for the inner product have been studied for decades and find applications throughout data science and database applications,
They can be used to quickly compute document similarity, to speed up the evaluation of machine learning models, and to estimate quantities like join size~\cite{SaltonWongYang:1975, AlonMatiasSzegedy:1999,Achlioptas:2003,RusuDobra:2008,CormodeGarofalakisHaas:2011}. Recently, inner product sketching has found applications in scalable dataset search and augmentation, where sketches can be used to estimate 
correlations between columns in unjoined tables~\cite{SantosBessaChirigati:2021}. \revised{In such applications, we have a large repository of $D$ vectors that we wish to compare against a query vector using inner products. By preprocessing the database with sketching, 
we can efficiently evaluate new queries in much less than the naive $O(Dn)$ time.
}

% } 
% In this setting, a collection of datasets (e.g., in a repository or data lake) is pre-processed and vectors are created for their attributes; suppose there is a total of $D$ vectors. When dataset discovery query $Q$ is issued, the dataset in the query is sketched to size $m$ in $O(N)$ time. To evaluate $Q$ and identify all attributes that are \emph{related} to $Q$, all inner products between $Q$ and the $D$ vectors can be computed in $O(D)$ time. The na\"{\i}ve method would take $O(ND)$ time. 

%\vspace{-.15cm}
\subsection{Prior Work}
\noindent\textbf{Inner Product Estimation via Linear Sketching.}
Until recently, all sketching algorithms with 
%yielded 
strong worst-case accuracy guarantees for approximating the inner product between arbitrary inputs were based on \emph{linear sketching}. Such methods include Johnson-Lindenstrauss random projection (JL) \cite{Achlioptas:2003, DasguptaGupta:2003}, the closely related AMS sketch \cite{AlonMatiasSzegedy:1999,AlonGibbonsMatias:1999}, and the CountSketch algorithm \revised{\cite{CharikarChenFarach-Colton:2002,CormodeGarofalakis:2005}.} These methods are considered ``linear'' because the sketching operation $\bv{a} \rightarrow \mathcal{S}(\bv{a})$ is a linear map, meaning that $\mathcal{S}(\bv{a}) = \bs{\Pi}\bv{a}$ for a matrix $\bs{\Pi}\in \R^{m\times n}$. $\bs{\Pi}$ is typically chosen at random and its row count $m$ is equal to the size of the sketch $\mathcal{S}(\bv{a})$. To estimate the inner product between  $\bv{a}$ and $\bv{b}$, the standard approach is to 
simply return $\langle \mathcal{S}(\bv{a}), \mathcal{S}(\bv{b})\rangle = \langle \bs{\Pi}\bv{a}, \bs{\Pi}\bv{b}\rangle$.
% \footnote{Variants of linear sketching methods that further compress $\bs{\Pi}\bv{a}$ by thresholding or rounding its entries can also be used to estimate inner products. Such methods include the seminal SimHash method \cite{Charikar:2002} and quantized Johnson-Lindenstrauss methods \cite{Jacques:2015}. 
% } 
For all common linear sketching methods (including those listed above), it can be shown (see e.g., \cite{ArriagaVempala:2006}) that, if we choose the sketch size $m = O\left(1/\epsilon^2\right)$, then with high probability:
\begin{align}
\label{eq:jl_guar}
%\vspace{-.2cm}
\left|\langle \mathcal{S}(\bv{a}), \mathcal{S}(\bv{b})\rangle - \langle\bv{a},\bv{b}\rangle\right| \leq \epsilon \|\bv{a}\|_2 \|\bv{b}\|_2.
%\vspace{-.15cm}
\end{align}
Here $\|\bv{x}\|_2 = \sqrt{\sum_{i=1}^n \bv{x}_i^2}$ denotes the Euclidean norm of a vector $\bv{x}$.

\begin{table*}[t]
\begin{center}
\begin{small}
    \def\arraystretch{1.4}%
	\begin{tabular}{|p{0.268\linewidth}|p{0.255\linewidth}|p{0.267\linewidth}|p{0.12\linewidth}|}
		\hline
		\textbf{Method}  & \textbf{High probability error guarantee for sketch of size $m = O(1/\epsilon^2)$} & \textbf{Time to compute sketch for length $n$ vector with $N$ non-zero entries} & \textbf{Strict bound on sketch size?}\\
		\hline
		JL Projection/AMS Sketch \cite{AlonMatiasSzegedy:1999,ArriagaVempala:2006} & $\epsilon \cdot\|\bv{a}\|_2 \|\bv{b}\|_2$ & $O(Nm)$ & \cmark\\ 
	       \hline
    CountSketch\revised{/Fast-AGMS \cite{CharikarChenFarach-Colton:2002,CormodeGarofalakis:2005}} & $\epsilon \cdot\|\bv{a}\|_2 \|\bv{b}\|_2$ & $O(N)$ & \cmark \\
          \hline
		Weighted MinHash (\wmh) \cite{BessaDFMMSZ:2023} & $\epsilon \cdot \max\left(\|\bv{a}_{\mathcal{I}}\|_2\|\bv{b}\|_2, \|\bv{a}\|_2\|\bv{b}_{\mathcal{I}}\|_2 \right)$ & $O(Nm\log n)$ & \cmark  \\
		  \hline
        
        \textbf{\thresholdsampling} & $\epsilon \cdot \max\left(\|\bv{a}_{\mathcal{I}}\|_2\|\bv{b}\|_2, \|\bv{a}\|_2\|\bv{b}_{\mathcal{I}}\|_2 \right)$ & $O(N)$ & \xmark\\
		\hline
        \textbf{\prioritysampling} & $\epsilon \cdot \max\left(\|\bv{a}_{\mathcal{I}}\|_2\|\bv{b}\|_2, \|\bv{a}\|_2\|\bv{b}_{\mathcal{I}}\|_2 \right)$ & $O(N\log m)$ & \cmark\\
		\hline
	\end{tabular}
         \end{small}
          \end{center}
	\caption{Comparison of error guarantees and computational cost for sketching methods when used to estimate the inner product between vectors $\bv{a}$ and $\bv{b}$. Note that $\epsilon \cdot \max\left(\|\bv{a}_{\mathcal{I}}\|_2\|\bv{b}\|_2, \|\bv{a}\|_2\|\bv{b}_{\mathcal{I}}\|_2 \right)$ is always a better  guarantee than $\epsilon \cdot\|\bv{a}\|_2 \|\bv{b}\|_2$, and often significantly so when $\bv{a}$ and $\bv{b}$ are sparse with limited overlap between non-zero entries. Our Threshold and \prioritysampling methods obtain this better bound while matching or nearly matching the fast runtime of the less accurate CountSketch method.}
	\label{tab:error_guarantees}
	
\vspace{-.3cm}
\end{table*}

\revised{\myparagraph{Better Accuracy via Weighted MinHash}
While \eqref{eq:jl_guar} is a strong guarantee, it was recently improved by Bessa et al. \cite{BessaDFMMSZ:2023}, who introduce a method based on the popular Weighted MinHash (WMH) algorithm
\cite{ChiZhu:2017,Shrivastava:2016,ManasseMcSherryTalwar:2010,GollapudiPanigrahy:2006}. Like unweighted MinHash and techniques such as conditional random sampling \cite{Broder:1997,LiChurchHastie:2006}, %Bessa et al.'s 
the \wmh sketch contains a subsample of entries from $\bv{a}$ and $\bv{b}$ that can be used to approximate the inner product.} Importantly, entries with higher absolute value are sampled with higher probability, since they can contribute more to the inner product sum $\langle \bv{a}, \bv{b}\rangle = \sum_{i=1}^n \bv{a}_i\bv{b}_i$.
%
%Bessa et al. prove that, with 
Using sketches of size $O(1/\epsilon^2)$, \wmh achieves accuracy: 
\begin{align}
\label{eq:wmh_guar}
%\vspace{-.2cm}
\left|\langle \mathcal{S}(\bv{a}), \mathcal{S}(\bv{b})\rangle - \langle\bv{a},\bv{b}\rangle\right| \leq \epsilon \max\left(\|\bv{a}_\mathcal{I}\|_2\|\bv{b}\|_2, \|\bv{a}\|_2\|\bv{b}_\mathcal{I}\|_2\right).
%\vspace{-.15cm}
\end{align}
Here $\mathcal{I} = \{i: \bv{a}[i] \neq 0 \text{ and } \bv{b}[i] \neq 0 \}$ is the set of all indices in the \emph{intersection} of the supports of $\bv{a}$ and $\bv{b}$, and $\bv{a}_{\mathcal{I}}$ and $\bv{b}_{\mathcal{I}}$ denote the vectors restricted to the indices in $\mathcal{I}$.\footnote{Prior to the work of \cite{BessaDFMMSZ:2023}, the stronger guarantee of \eqref{eq:wmh_guar} was known to be obtainable for the special case of 
% binary vectors, \revised{for which inner product estimation corresponds to the set intersection problem \cite{PaghStockelWoodruff:2014}.
\revised{inner product of binary vectors, which corresponds to the set intersection problem \cite{PaghStockelWoodruff:2014}.
}} Since we always 
have $\|\bv{a}_{\mathcal{I}}\|_2 \leq \|\bv{a}\|_2$ and $\|\bv{b}_{\mathcal{I}}\|_2 \leq \|\bv{b}\|_2$, the error in \eqref{eq:wmh_guar} is always less or equal to the error in \eqref{eq:jl_guar} for the linear sketching methods. 

As confirmed by experiments in \cite{BessaDFMMSZ:2023}, the improvement over \revised{linear sketching} can be significant in applications where $\bv{a}$ and $\bv{b}$ are sparse and their non-zero entries only overlap at a small fraction of indices.  I.e., when $|\mathcal{I}|$ is much smaller than the number of non-zeros in $\bv{a}$ and $\bv{b}$. This is common when inner product sketches are used for data discovery, either to estimate join-sizes or correlations between unjoined tables~\cite{ZhuNargesianPu:2016, Castro-FernandezMinNava:2019, YangZhangZhang:2019}. In these applications, overlap between non-zeros in $\bv{a}$ and $\bv{b}$ corresponds to overlap between the keys of the 
tables being joined, which is often small. \revised{For example, consider a setting where we want to find additional data for use in taxi demand prediction. Given a table of 2022-2023 taxi trip data, we would like to augment it using weather information available in a table of historical weather data from the last 50 years; this leads to just a $4\%$ overlap in keys. More examples are discussed in \Cref{sec:applications}.}

%\begin{comment}
%As a result, we typically have $\|\bv{a}_{\mathcal{I}}\|_2 \ll \|\bv{a}\|_2$ and $\|\bv{b}_{\mathcal{I}}\|_2 \ll \|\bv{b}\|_2$, so \eqref{eq:wmh_guar} is much smaller than \eqref{eq:jl_guar}.
%\end{comment}

% JF: was not able to get rid of this space; to do so we need to remove a few lines in the previous page
%\vspace{-.8cm}
\myparagraph{Limitations of \wmh sketches}
While \wmh provides better accuracy than linear sketching, it has important limitations.
Notably, the method has high computational complexity, requiring $O(Nm\log n)$ time to produce a sketch of size $m$ from a length $n$ vector $\bv{a}$ 
with $N\leq n$ non-zero entries.
While this nearly matches the $O(Nm)$ complexity of a JL projection or AMS sketch (which require multiplying $\bv{a}$ by a dense matrix), it is far slower than methods like CountSketch or the $k$-minimum values (KMV) sketch \cite{BeyerHaasReinwald:2007}, which can be applied in $O(N)$ or $O(N\log m)$ time, respectively.
It is possible to reduce the complexity of \wmh to $O(N + m\log m)$ using recent work \cite{Christiani:2020,Ertl:2018}. However, as shown in \Cref{sec:experiments}, even these theoretically faster methods are orders of magnitude slower
in practice than the simpler sketches introduced in our work. 

Beyond computational cost, another disadvantage of \wmh is that it is complex, both to implement and analyze. For example, while a high probability bound is obtained 
%by Bessa et al. 
in \cite{BessaDFMMSZ:2023}, they are unable to \revised{analyze} the variance of the method. This makes it difficult, for example, \revised{to compute confidence intervals for estimated inner products.} 
Moreover, the \wmh requires careful discretization of the vectors being sketched\revised{, which leads} to large constant factors in the results of \cite{BessaDFMMSZ:2023}. Such factors do not impact the  Big O claim that a sketch of size $O(1/\epsilon^2)$ achieves error guarantee \eqref{eq:wmh_guar}, but matter a lot in practice. Practical accuracy of the method is also negatively impacted by the fact that it samples entries from $\bv{a}$ and $\bv{b}$ \emph{with replacement}, which can lead \revised{to redundancy} in 
the sketches $\mathcal{S}(\bv{a})$ and $\mathcal{S}(\bv{b})$.

% \begin{align*}
% \E\left[\langle \mathcal{S}(\bv{a}), \mathcal{S}(\bv{b})\rangle\right]&= \langle \bv{a},\bv{b}\rangle &\text{and} \Var\left[\langle \mathcal{S}(\bv{a}), \mathcal{S}(\bv{b})\rangle\right]&= \frac{c}{k}\|\bv{a}\|_2^2\|\bv{b}\|_2^2,
% \end{align*}
% where $c$ is a fixed constant.
% So, $\langle \mathcal{S}(\bv{a}), \mathcal{S}(\bv{b})\rangle$ provides and unbiased estimate for the true inner product between $\bv{a}$ and $\bv{b}$. Moreover, if we choose $k = O\left(1/\epsilon^2\right)$, the variance bound combined with a simple application of Chebyshev's inequality ensures that with high probability,

% \begin{fact}[Linear Sketching for Inner Products \cite{ArriagaVempala:2006}]
% 	\label{fact:jl_result}
% 	Let $\epsilon,\delta \in (0,1)$ be accuracy and failure probability parameters respectively and let $m = O(\log(1/\delta)/\epsilon^2)$.
% 	Let $\bv{\Pi} \in \R^{m \times n}$ be a random matrix with each entry set independently to $+\sqrt{1/m}$ or $-\sqrt{1/m}$ with equal probability. For length $n$ vectors $\bv{a},\bv{b} \in \R^n$, let  $\mathcal{S}(\bv{a}) = \bv{\Pi}\bv{a}$ and $\mathcal{S}(\bv{b}) = \bv{\Pi}\bv{b}$. With probability at least $1-\delta$,
% 	\begin{align*}
% 		\left |  \langle \mathcal{S}(\bv{a}), \mathcal{S}(\bv{b})\rangle - \langle \bv{a}, \bv{b} \rangle \right | \leq \epsilon \|\bv{a}\| \|\bv{b}\|
% 	\end{align*} where $\norm{\bv x} = \left (\sum_{k=1}^n x[k]^2\right)^{1/2}$ denotes the standard Euclidean norm.
% \end{fact}

\subsection{Our Contributions}
\noindent\textbf{Methods and Theory.} In this paper, we present and analyze two algorithms for inner product sketching that eliminate the limitations of \wmh sketches, while maintaining the same strong theoretical guarantees. Both are based on existing methods for weighted sampling of vectors \emph{without replacement}, but our choice of sampling probabilities, estimation procedure, and theoretical analysis are new, and tailored to the problem of inner product estimation.

The first method we study is based on \emph{Threshold Sampling}
\cite{Flajolet:1990,DuffieldLundThorup:2005}. 
We show that, when used to sample vector entries with probability proportional to their {squared value}, this method produces inner product sketches that yield the same accuracy guarantee as \wmh sketches.
%, beating linear sketching. 
At the same time, the method is extremely simple to implement and can be applied to a vector with $N$ non-zero entries in linear $O(N)$ time. Moreover, unlike \wmh, the analysis of the method is 
% completely % JF: too many adjectives...
straightforward.
Its only disadvantage is that Threshold Sampling produces sketches that \emph{randomly vary} in size. The user can specify a parameter $m$ and is guaranteed that the sketch has size $m$ in \emph{expectation}, and will not exceed $m + O(\sqrt{m})$ with high-probability. 
However, there is no hard bound. 

We address this drawback with an alternative method based on Priority Sampling, which has been widely studied in the \revised{sketching and statistics} literature \cite{Ohlsson:1998,DuffieldLundThorup:2004,SzegedyThorup:2007}. Priority Sampling offers a hard sketch size bound and can \revised{construct a size $m$ sketch in near-linear $O(N\log m)$ time.} \revised{While significantly more challenging to analyze than Threshold Sampling, by introducing a new estimation procedure and building on a recent analysis of Priority Sampling for a different problem (subset sum estimation)~\cite{DaliriFreireMusco:2023}, we are able to show that it enjoys the same guarantees as \wmh.} Our analysis of Priority Sampling is the main theoretical contribution of this paper.

% \revised{While sampling-based approaches utilize the benefit of sampling without replacement, leading to more effective sampling in practice, linear sketching approaches leverage their inherent linearity to construct sketches online from the sum of a stream of vectors. Recent work, such as that by Cohen, Pagh, and Woodruff \cite{CohenPaghWoodruff:2020}, has focused on linearizing the sampling method. This advancement could enable sampling approaches to match the flexibility offered by linear sketching techniques.}

%\vspace{-.15cm}
\myparagraph{Experimental Results} In addition to theoretical analysis, we experimentally compare Threshold and \prioritysampling with linear sketching algorithms like JL random projections and CountSketch, as well as  sampling-based sketches like $k$-minimum values (KMV)\footnote{The KMV sketch is not typically thought of as a sketch for estimating inner products between arbitrary vectors, but can be modified to do so. See \cite{BessaDFMMSZ:2023} for details.}, MinHash, and \wmh.
We evaluate these on a variety of applications, including join size estimation and correlation estimation between unjoined tables. 
%For the second problem, 
We introduce an approach to perform 
%so-called 
\emph{join-correlation estimation} \cite{SantosBessaChirigati:2021}  using \emph{any} inner product sketching method (\Cref{sec:applications}) that we believe may be of independent interest.

\revised{Our} Threshold and \prioritysampling methods offer significantly better accuracy than the baselines, beating both linear sketches and \wmh sketches. This indicates that, despite having identical worst-case \revised{guarantees}, the hidden \revised{constants} are smaller for our methods than for \wmh. An optimized version of our sketches tailored to the application of join-correlation estimation outperforms the recently introduced Correlation Sketches method from \cite{SantosBessaChirigati:2021}, which is based on KMV.
 We also test the \runtime efficiency of Threshold and Priority Sampling for sketch construction. Even when \wmh is implemented using the efficient \dartmh algorithm~\cite{Christiani:2020},
our methods are faster by more than an order of magnitude.

%\vspace{-.3cm}
\myparagraph{Our Approach} 
\revised{As in \cite{BessaDFMMSZ:2023}, sketches consist of samples from $\bv{a}$ and $\bv{b}$. We estimate the inner product $\sum_{i=1}^n \bv{a}_i\bv{b}_i$ using only a subset of terms in the sum.} Specifically, our estimators are of the form $\sum_{j\in \mathcal{T}} w_j \cdot \bv{a}_j\bv{b}_j$, where $\mathcal{T}$ is a subset of $\{1, \ldots, n\}$ and $\{w_j, j\in \mathcal{T}\}$ are appropriately chosen positive weights.
% To compute this estimate, we need that \emph{both} $\bv{a}_j$ is stored in $\mathcal{S}(\bv{a})$, and $\bv{b}_j$ is stored in $\mathcal{S}(\bv{b})$.
To compute this estimate, we need to store \emph{both} $\bv{a}_j$ in $\mathcal{S}(\bv{a})$ and $\bv{b}_j$  and $\mathcal{S}(\bv{b})$.
If $\bv{a}$ and $\bv{b}$ are sampled independently at random, the probability of obtaining matching indices in both sketches would be small, thus leading to a small number of usable samples, and a poor inner product estimate. Our Threshold and Priority Sampling methods avoid this issue by using shared random seeds to sample from the vectors in a \emph{coordinated way}, which ensures that if entry $\mathbf{a}_j$ is sampled from $\mathbf{a}$, it is more  likely that the corresponding $\mathbf{b}_j$ is sampled from $\mathbf{b}$. This idea is {not new}: coordinated variants of Threshold and Priority Sampling have been studied in prior work on different problems, as have coordinated variants of related methods like PPSWOR sampling \cite{CohenKaplan:2013, Cohen:2023}. What \emph{is new} is how we apply and analyze such methods 
for the problem of inner product estimation.

Besides \wmh \cite{BessaDFMMSZ:2023}, we are only aware of one prior paper that addresses the inner product estimation problem using coordinated sampling: the ``End-Biased Sampling'' algorithm of \cite{EstanNaughton:2006} can be viewed as a variant of Threshold Sampling where the $i^\text{th}$ entry of $\bv{a}$ is sampled with probability proportional to the magnitude $|\bv{a}_i|$. We instead use the squared magnitude $|\bv{a}_i|^2$. 
While variance bounds are shown in \cite{EstanNaughton:2006}, due to this choice of sampling probability, they fall short of improving on results for linear sketches, i.e., on Eq. \eqref{eq:jl_guar}. Additionally, unlike our work, \cite{EstanNaughton:2006} does not address the issue of how to obtain a fixed-size sketch.
We discuss End-Biased Sampling 
further in \cref{sec:experiments} and fully review related work in \cref{sec:related_work}.

%\vspace{-.3cm}
\myparagraph{Summary and Paper Roadmap} 
Our contributions can be summarized as follows:
\begin{itemize}[leftmargin=10pt]
    \item We show how to apply two coordinated sampling methods, Threshold and Priority Sampling, to the inner product sketching problem, invoking these methods with a specific choice of sampling probabilities and estimation procedures. 
    \item We prove that these methods enjoy better theoretical accuracy guarantees than linear sketches, and match the best-known guarantees provided by \wmh~\cite{BessaDFMMSZ:2023} (\Cref{sec:thresholdsampling} and \Cref{sec:prioritysampling}).
    % To the best of our knowledge, this is the first sketch to provide such strong guarantees while being both simple to implement and extremely efficient.
    \item We perform an empirical evaluation, showing that Threshold and \prioritysampling outperform state-of-the-art sketches in both accuracy and \runtime on a variety of applications (\Cref{sec:experiments}). 
    \item We show a black-box reduction from one such application, join-correlation estimation, to inner product estimation (\Cref{sec:applications}).
\end{itemize}
% While all key results and proofs are in the main text, some secondary results (like our optimized sketches for join-correlation) are deferred to an appendix included with our supplemental material. 

%--------------------------------
% \section{Our Method and Analysis}
\section{\thresholdsampling} \label{sec:thresholdsampling}
We begin by introducing an inner product sketch based on Threshold Sampling, which is a method popularized in computer science by \cite{DuffieldLundThorup:2005}, but long studied in statistics under the name ``Poisson Sampling''.\footnote{A variant of Threshold Sampling with \emph{uniform} probabilities was also studied under the name ``adaptive sampling'' by Wegman in 1984 and later by Flajolet \cite{Flajolet:1990}.} 
%As mentioned, 
Our algorithm based on Threshold Sampling is straightforward to implement and analyze, but still matches the strong theoretical guarantees of \wmh sketches \cite{BessaDFMMSZ:2023}, while improving on runtime and 
%practical 
performance. Its presentation serves as a warm-up for our Priority Sampling method (\Cref{sec:prioritysampling}), which is more difficult to analyze, but 
has the advantage of a deterministic sketch size.

\smallskip
\noindent\textbf{Sketching.} %As discussed in the introduction,
\revised{As discussed, the goal of our sketching methods (and of \wmh) is to randomly sample entries from $\bv{a}$ and $\bv{b}$, and to use those samples to estimate the inner product sum $\langle \bv{a}, \bv{b}\rangle = \sum_{i=1}^n \bv{a}_i\bv{b}_i$. To obtain strong guarantees, we need the samples to be both \emph{coordinated} and \emph{weighted}. In particular, since they contribute more to the inner product, entries with larger magnitude should be sampled with higher probability. Moreover, coordination requires that $\bv{b}_j$ is more likely to be sampled if $\bv{a}_j$ is.}
% Additionally, we need the samples to be \emph{coordinated}, meaning that if $\bv{a}_j$ is sampled from $\bv{a}$, then it should be more likely that $\bv{b}_j$ is sampled from $\bv{b}$: the reason for this is that we need \emph{both} $\bv{a}_i$ and $\bv{b}_i$ to compute a single term in the sum $\sum_{i=1}^n \bv{a}_i\bv{b}_i$. 
Ensuring coordination is not obvious because, in the sketching setting we consider, $\mathcal{S}(\bv{a})$ and $\mathcal{S}(\bv{b})$ need to be computed completely independently from each other: when we sample entries from $\bv{b}$ to form $\mathcal{S}(\bv{b})$, we have no knowledge about what entries were sampled from $\bv{a}$ to form $\mathcal{S}(\bv{a})$.

\begin{figure}[ht]
\vspace{-1em}
\begin{algorithm}[H]
    \caption{\thresholdsampling}
    \label{alg:threshold_sampling}
    \begin{algorithmic}[1]
        \Require Length $n$ vector $\bv{a}$, random seed $s$, target sketch size $m$.
        \Ensure Sketch $\mathcal{S}(\bv{a}) = \{K_{\bv{a}}, V_{\bv{a}}, \tau_{\bv{a}}\}$, where $K_{\bv{a}}$ is a subset of indices from $\{1, \ldots, n\}$ and $V_{\bv{a}}$ contains $\bv{a}_i$ for all $i\in K_{\bv{a}}$.
        \algrule
        \State Use random seed $s$ to select a uniformly random hash function $h: \{1,..., n\}\rightarrow [0,1]$. Initialize $K_{\bv{a}}$ and $V_{\bv{a}}$ to be empty lists.
        \For{$i$ such that $\bv{a}[i]\neq 0$}
            \State Set threshold $\tau_i =  m\cdot \frac{\bv{a}_i^2}{\|\bv{a}\|_2^2}$.
                \If{$h(i) \leq \tau_i$}
              \State Append $i$ to $K_{\bv{a}}$, append $\bv{a}_i$ to $V_{\bv{a}}$.
              \EndIf
        \EndFor
        % \State Compute $s = \left\{i | h(i) \leq \tau \cdot \frac{\bv{a}_i^2}{\|\bv{a}\|_2^2}, i\in \{1, \ldots, n\},\, \bv{a}[i] \neq 0\right\}$.  \label{alg:sampling_limbo_sktech_weighted}
        % \State Set $\mathcal{S}(\bv{a})^{key}[i] = s$ and $\mathcal{S}(\bv{a})^{val}[i] = \left\{\bv{a}[i] | i\in s\right\}$
        \State \Return $\mathcal{S}(\bv{a}) = \{K_{\bv{a}}, V_{\bv{a}}, \tau_{\bv{a}}\}$ where $\tau_{\bv{a}} = m/\|\bv{a}\|_2^2$.
    \end{algorithmic}
    \vspace{-.2em}
\end{algorithm}
\vspace{-2em}
\end{figure}

Threshold Sampling achieves sampling that is both weighted and coordinated using a simple  technique. We first assume access to a hash function $h:\{1, \ldots, n\}\rightarrow [0,1]$ that maps indices to uniformly random real numbers in the interval $[0,1]$. Assuming access to such a function is standard in the literature, and we note that, in practice, $h$ can be replaced with a pseudorandom function that maps to a sufficiently large discrete set,  e.g., to $\{1/U, 2/U \ldots, 1\}$ for $U = 2^{32}$ or some other large integer \cite{BeyerHaasReinwald:2007, CormodeGarofalakisHaas:2011}. 
As shown in \Cref{alg:threshold_sampling} and illustrated in \Cref{fig:threshold_sampling_on_example}, we sketch the vector $\mathbf{a}$ by selecting a threshold, $\tau_i$ for each index (Line 3). We then hash all indices $i$ for which $\bv{a}[i]\neq 0$ to the interval $[0,1]$, and keep as a sample all entries of $\bv{a}$ for which the hash value 
$h(i)$ is below the threshold (Line 4,5).

Concretely, we choose the threshold $\tau_i = m\cdot {\bv{a}_i^2}/{\|\bv{a}\|_2^2}$. Here $m$ is a fixed parameter that controls the size of the final sketch, $\mathcal{S}(\bv{a})$, returned by \Cref{alg:threshold_sampling}. So, we see that the threshold $\tau_i$ is higher for indices $i$ where $\bv{a}_i^2$ is larger. Thus, larger entries in the vector are sampled with higher probability. Note that this is in contrast to ``End-Biased Sampling'' \cite{EstanNaughton:2006}, which sets $\tau_i = m\cdot \frac{|\bv{a}_i|}{\|\bv{a}\|_1}$, where $\|\bv{a}\|_1 = \sum_{i=1}^n |\bv{a}_i|$ is the $\ell_1$ norm. While this choice also aligns with the goal that larger entries should be sampled with higher probability, it does not lead to the same strong theoretical guarantees.
% \footnote{\revised{We note that \thresholdsampling can be implemented in a single pass over the vector $\bv{a}$ using a standard rejection sampling approach. Details are provided in Appendix}}
% that we will prove for \Cref{alg:threshold_sampling} and empirically verify (Section~\ref{sec:experiments}).

In addition to collecting a weighted sample, since the \emph{same hash function} $h$ is used when sampling from both $\bv{a}$ and $\bv{b}$, the samples are coordinated. If $h(i)$ is small, we are more likely sample \emph{both} $\bv{a}_i$ and $\bv{b}_i$. The same idea is present in common methods for unweighted coordinated sampling like MinHash or the KMV sketch \cite{Broder:1997,BeyerHaasReinwald:2007}.

Finally, we note that the sketch procedure in \Cref{alg:threshold_sampling} runs in $O(N)$ time when $\bv{a}$ has $N$ non-zero entries, at least when the vector is stored in a standard sparse-vector format (e.g., a key/value store) which allows iteration over the non-zero entries in $O(N)$ time.\footnote{\revised{One computational {disadvantage} of sampling methods like Threshold Sampling in comparison to linear sketching is that they cannot be immediately implemented in a streaming setting where entries in $\bv{a}$ and $\bv{b}$ are updated  incrementally; we need to know the magnitude of each entry in advance to perform sampling. We believe it is possible to resolve this issue using streaming $\ell_2$ sampling algorithms (see e.g., \cite{JayaramWoodruff:2018} or \cite{CohenPaghWoodruff:2020}). We leave the details of how to do so most effectively to future work.}}

\smallskip
\noindent\textbf{Estimation.} Once our sketches $\mathcal{S}(\bv{a})$ and $\mathcal{S}(\bv{b})$ are computed, to estimate the inner product between $\bv{a}$ and $\bv{b}$, we simply compute a weighted sum between entries that are sampled in both $\mathcal{S}(\bv{a})$ and $\mathcal{S}(\bv{b})$ (see \Cref{alg:inner_product_estimator}). To ensure the sum equals the true inner product $\langle \bv{a}, \bv{b}\rangle$ in expectation, the weight for index $i$ in the sum is the inverse of the probability that \emph{both} $\bv{a}_i$ and $\bv{b}_i$ were included in the sketches $\mathcal{S}(\bv{a})$ and $\mathcal{S}(\bv{b})$. We can check that this probability is equal to $\min\left(1, m\cdot \bv{a}_i^2/\|\bv{a}\|_2^2,m\cdot \bv{b}_i^2/\|\bv{b}\|_2^2\right)$. This can be computed in $O(1)$ time, so overall the estimator can be computed in time linear in the sketch size. Note that the estimator requires knowledge of the scaling parameters $m/\|\bv{a}\|_2^2$ and $m/\|\bv{b}\|_2^2$, so we include these numbers in our sketches $\mathcal{S}(\bv{a})$ and $\mathcal{S}(\bv{b})$ as $\tau_{\bv{a}}$ and $\tau_{\bv{b}}$.

\begin{figure}[t]
\vspace{-.6em}
\centering
\footnotesize
\begin{subfigure}[b]{\linewidth}
\centering
\parbox{1.0\linewidth}{
\centering
\setlength{\tabcolsep}{0.35em}
\begin{tabular}{|c|c|c|c|c|c|c|c|c|c|c|c|c|c|c|c|c|}
\multicolumn{10}{c}{}        \\
    \hline
    index & \bf 1 & \bf 2 & \bf 3 & \bf 4 & \bf 5 & \bf 6 & \bf 7 & \bf 8 & \bf 9 & \bf 10 & \bf 11 & \bf 12 & \bf 13 & \bf 14 & \bf 15 & \bf 16 \\
    \hline
    $\bv a$ & 0 & 0 & 2.5 & 0 & 0 & 2.3 & 0 & 4 & 0 & 0 & 0.5 & 0 & 3 & 0 & 0 & -3.7\\
    \hline
    $\bv b$ & 0 & 0 & -3.1 & 0 & 0 & 0 & 0.4 & -4.2 & 0 & 1.5 & 1 & 0 & -2.6 & -5.9 & 0 & 0\\
    \hline
\end{tabular}
}
\caption{
Vectors $\bv{a},\bv{b}$ to be sketched. Their
% with Euclidean norms $\|\bv{a}\|_2 = 7.10$ and $\|\bv{b}\|_2 = 8.50$. 
inner product is $\langle \bv{a}, \bv{b}\rangle = -31.85$. 
}
\label{fig:example_vector}
\end{subfigure}

\vspace{.5em}
\begin{subfigure}[b]{\linewidth}
\centering
\parbox[t]{.4\linewidth}{
    \centering
    \begin{tabular}[t]{|c|c|c|c|}
    % \multicolumn{4}{c}{\hspace{1em}}        \\
    \hline
    \textbf{$i$} & \textbf{$h(i)$} & \textbf{$\tau_{i}(\bv{a})$} & \textbf{$\tau_{i}(\bv{b})$} \\
    \hline
    % 1  & 0.71 & 0.201 & \xmark  \\
    % 2  & 0.66 & \xmark & 0.229 \\
    3  & 0.11 & \cellcolor{purple!10}0.495 & \cellcolor{purple!10} 0.532 \\
    6 & 0.39 & 0.419 & \xmark \\
    7 & 0.92  & \xmark & 0.009\\
    8 & 0.14 & \cellcolor{purple!10}1.268 & \cellcolor{purple!10}0.977   \\
    10 & 0.42 & \xmark & 0.125    \\
    11 & 0.8 & 0.020 & 0.055 \\
    13 & 0.43 & \cellcolor{purple!10}0.713 & 0.374 \\
    14 & 0.07 & \xmark & \cellcolor{purple!10} 1.928 \\
    16 & 0.23 & \cellcolor{purple!10} 1.085 & \xmark \\
    \hline
    \end{tabular}
}
\hspace{1em}
\parbox[t]{.2\linewidth}{
    \centering
    \begin{tabular}[t]{|c|c|}
    % \multicolumn{2}{c}{$\mathcal{S}(\bv{a})$}        \\
    \hline
    \textbf{$K_{\bv{a}}$} & \textbf{$V_{\bv{a}}$} \\
    \hline
    3	& 2.5	 \\
    8	& 4    \\
    13	& 3     \\
    16	& -3.7	     \\
    \hline
    \end{tabular}
    
    \vspace{.5em}
    $\tau_{\bv{a}} = .079$
    \vspace{-.5em}

    $\underbrace{\hspace{5em}}_{%
    \textstyle
     \begin{gathered}
      \mathcal{S}(\bv{a})
    \end{gathered}
    }$
} 
\parbox[t]{.2\linewidth}{
    \centering
    \begin{tabular}[t]{|c|c|}
    % \multicolumn{2}{c}{$\mathcal{S}(\bv{b})$}        \\
    \hline
    \textbf{$K_{\bv{b}}$} & \textbf{$V_{\bv{b}}$} \\
    \hline
    3	& -3.1	 \\
    8	& -4.2	 \\
    14	& -5.9	 \\
    \hline 
    \end{tabular}
    
    \vspace{.5em}
    $\tau_{\bv{b}} = .055$
    \vspace{-.5em}

    $\underbrace{\hspace{5em}}_{%
    \textstyle
     \begin{gathered}
      \mathcal{S}(\bv{b})
    \end{gathered}
    }$
} 
% \hspace{0.5em}
% \parbox{.3\linewidth}{
%     \centering
%     \texttt{SIZE}($\mathcal{S}_{\bv{a}}$) = 3 \\
%     \texttt{SIZE}($\mathcal{S}_{\bv{b}}$) = 3 \\
%     \texttt{DOT\_PRODUCT}($\mathcal{S}_{\bv{a}}, \mathcal{S}_{\bv{b}}$) = 21 \\
% } 
\vspace{-0.1em}
\caption{
Example sketches $\mathcal{S}(\bv{a})$ and $\mathcal{S}(\bv{b})$ obtained using \Cref{alg:threshold_sampling} with target sketch size $m=4$. Since the size of the sketch returned by the method is random, $\mathcal{S}(\bv{a})$ has size $4$, but $\mathcal{S}(\bv{b})$ is slightly smaller. The columns $\tau_i(\bv{a})=m\cdot {\bv{a}_i^2}/{\|\bv{a}\|_2^2}$ and $\tau_i(\bv{b})=m\cdot {\bv{b}_i^2}/{\|\bv{b}\|_2^2}$ contain the thresholds computed in Line 3 of \Cref{alg:threshold_sampling}.
Thresholds are only computed for non-zero entries since we never sample entries with value $0$. The highlighted thresholds correspond to items that are included in the sketch, i.e., the threshold is larger than the hash value $h(i)$. If the sketches $\mathcal{S}(\bv{a})$ and  $\mathcal{S}(\bv{b})$ above are used used in our estimator from \Cref{alg:inner_product_estimator}, we obtain an approximate inner product of -32.85, which is close to the true inner product of -31.85.
}
\label{fig:example_of_sketching_procedure}
\end{subfigure}
\vspace{-2em}
\caption{
Sketching with \thresholdsampling (\cref{alg:threshold_sampling}).
}
\label{fig:threshold_sampling_on_example}
\vspace{-.5em}
\end{figure}

\begin{figure}[t]
\vspace{-1em}
\begin{algorithm}[H]
    \caption{Inner Product Estimator}
    \label{alg:inner_product_estimator}
	\begin{algorithmic}[1]
		\Require Sketches $\mathcal{S}(\bv{a}) = \{K_{\bv{a}}, V_{\bv{a}}, \tau_{\bv{a}}\}$, $\mathcal{S}(\bv{b}) = \{K_{\bv{b}}, V_{\bv{b}}, \tau_{\bv{b}}\}$ constructed  by \Cref{alg:threshold_sampling} or \Cref{alg:priority_sampling} with the same seed $s$.
		\Ensure Estimate $w$ of $\langle \bv{a}, \bv{b}\rangle$.
		\algrule
        \State Compute $\mathcal{T} = K_{\bv{a}} \cap K_{\bv{b}}$. Note that for all $i\in \mathcal{T}$, $V_{\bv{a}}$ and $V_{\bv{b}}$ contain $\bv{a}_i$ and $\bv{b}_i$.
		\State \Return 
  \vspace{-1em}
  \begin{align*}
  W = \sum_{i \in \mathcal{T}} \frac{\bv{a}_i\bv{b}_i}{\min(1, \bv{a}_i^2 \cdot \tau_{\bv{a}}, \bv{b}_i^2 \cdot \tau_{\bv{b}})}.
  \end{align*}
         \vspace{-.5em}
\label{alg:inner_product_estimator_summation}
	\end{algorithmic}
     \vspace{-.2em}
\end{algorithm}
\vspace{-2em}
\end{figure}

\myparagraph{Comparison to \wmh} 
While both \wmh and \thresholdsampling use coordinated weighted sampling, 
%Note that our approach differs substantially from the \wmhs  \cite{BessaDFMMSZ:2023}. To obtain coordinated and weighted samples, 
\wmh does so in a less efficient way. It creates a variable number of copies of every entry in $\bv{a}$ to ensure that \revised{larger entries} are selected with higher probability. \revised{Only} an integer number of copies is possible, so this step requires careful discretization of $\bv{a}$'s entries. Our method, in contrast, encodes weight information more efficiently through the threshold $\tau_i$. Furthermore, to compute a sketch with $m$ samples, \wmh requires applying $m$ independent hash functions to every index $i$ where $\bv{a}$ is non-zero. This accounts for its \runtime dependence on $O(Nm)$. \thresholdsampling uses one \revised{hash function, so} runs in $O(N)$ time. 

Another difference between \thresholdsampling and \wmh is that, when run with parameter $m$, \thresholdsampling returns a sketch whose size is at most $m$ \emph{in expectation} (see \Cref{thm:main}). However, since entries of $\bv{a}$ are sampled independently, the actual size of the sketch will vary randomly around its expectation. In contrast, \wmh allows the user to set an exact sketch size. This issue motivates our Priority Sampling method (\Cref{sec:prioritysampling}), which is similar to \thresholdsampling but has a fixed sketch size.

\myparagraph{Theoretical Guarantees} 
Our main theoretical result on \thresholdsampling is as follows:
\begin{theorem}\label{thm:main}
For vectors $\bv a,\bv b \in \R^n$ and target sketch size $m$, let $\mathcal{S}(\bv a)=\{K_{\bv{a}}, V_{\bv{a}}, \tau_{\bv{a}}\}$ and $\mathcal{S}(\bv b)=\{K_{\bv{b}}, V_{\bv{b}}, \tau_{\bv{b}}\}$ be sketches returned by \Cref{alg:threshold_sampling}. Let $W$ be the inner product estimate returned by \Cref{alg:inner_product_estimator} applied to these sketches.  We have $\E\left[W\right] = \langle\bv{a},\bv{b}\rangle$ and
\begin{align*}
\Var\left[W\right] &\leq \frac{2}{m} \max\left(\|\bv{a}_{\mathcal{I}}\|_2^2\|\bv{b}\|_2^2, \|\bv{a}\|_2^2\|\bv{b}_{\mathcal{I}}\|_2^2\right).
\end{align*}
Moreover, let $|K_\bv{a}|$ and $|K_\bv{b}|$ be the number of index/values pairs stored in $\mathcal{S}(\bv a)$ and $\mathcal{S}(\bv b)$. We have $\E\left[|K_\bv{a}|\right] \leq m$ and $\E\left[|K_\bv{b}|\right] \leq m$.
% and, with probability $99/100$, $|K_\bv{a}|, |K_\bv{b}| \leq cm$ for a fixed constant $c$. So the sketches have size $O(m)$ with high probability.
\end{theorem}
Above, $\E\left[\cdot\right]$ denotes expected value and $\Var\left[\cdot\right]$ denotes variance. Recall that $\mathcal{I} = \{i: \bv{a}[i] \neq 0 \text{ and } \bv{b}[i] \neq 0 \}$ and $\bv{a}_{\mathcal{I}}$ and $\bv{b}_{\mathcal{I}}$ denote the vectors restricted to the indices in $\mathcal{I}$.
\Cref{thm:main} shows that the inner product estimate obtained using Threshold Sampling is correct in expectation and has bounded variance. Moreover, if the sketches are constructed with parameter $m$, the expected number of samples collected is always $\leq m$. Since the sketch needs to store 2 numbers \revised{for each sample (an index and a value),} as well as the scalar value $\tau_{\bv{a}}$, the expected storage size is thus $O(m)$. 

Given the expectation and variance bound in \Cref{thm:main}, we can apply Chebyshev's Inequality to obtain the following corollary:
\begin{corollary}\label{crl:main}
For any given values of $\epsilon, \delta \in (0,1)$ and vectors $\bv a,\bv b \in \R^n$, when run with target sketch $m$, \thresholdsampling returns an inner product estimate $W$ satisfying, with probability $1-\delta$,
\begin{align*}
% \label{eq:limbo_guar}
\left|W - \langle\bv{a},\bv{b}\rangle\right| \leq \sqrt{\frac{2/\delta}{m}} \max\left(\|\bv{a}_\mathcal{I}\|_2\|\bv{b}\|_2, \|\bv{a}\|_2\|\bv{b}_\mathcal{I}\|_2\right).
\end{align*}
Setting $m = \frac{2/\delta}{\epsilon^2}$, the error is $\epsilon \cdot \max\left(\|\bv{a}_\mathcal{I}\|_2\|\bv{b}\|_2, \|\bv{a}\|_2\|\bv{b}_\mathcal{I}\|_2\right)$.
\end{corollary}

This corollary matches the asymptotic guarantee of \wmh \cite{BessaDFMMSZ:2023}, improving on the bounds known for linear sketches like JL and CountSketch \cite{ArriagaVempala:2006}. At the same time, as we show in \cref{sec:experiments}, \thresholdsampling tends to perform better than \wmh in practice. We believe there are a number of reasons for this, including the fact that \thresholdsampling selects vector entries without replacement, and the fact that the variance bound in \Cref{thm:main} has an small constant factor of $2$. 
\revised{We prove \Cref{thm:main} below:}

\begin{proof}[Proof of \cref{thm:main}] Let $\mathcal{I}$ denote the set of all indices $i$ for which $\bv{a}_i \neq 0$ and $\bv{b}i \neq 0$. For any $i\in \mathcal{I}$, let $\mathbbm{1}_i$ denote the indicator random variable for the event that $i$ is included in \emph{both} $K_{\bv{a}}$ and $K_{\bv{b}}$. $\mathbbm{1}_i = 1$ if this event occurs and $0$ if it does not. Note that, for $i\neq j$, $\mathbbm{1}_i$ is independent from $\mathbbm{1}_j$, since the hash values $h(i)$ and $h(j)$ are drawn uniformly and independently from $[0,1]$.
Moreover, we claim that $\mathbbm{1}_i$ is equal to $1$ with probability:
\begin{align}
\label{eq:prob_claim}
    p_i = \min\left(1, \frac{m\cdot \bv{a}_i^2}{\|\bv{a}\|_2^2}, \frac{m\cdot \bv{b}_i^2}{\|\bv{b}\|_2^2}\right) = \min(1, \tau_{\bv{a}}\cdot \bv{a}_i^2, \tau_{\bv{b}}\cdot \bv{b}_i^2).
\end{align}
% \vspace{-0.5cm}
To see why this is the case, assume without loss of generality that $\bv{a}_i^2 \leq \bv{b}_i^2$. Then, by examining Line 3 of \Cref{alg:threshold_sampling}, we can see that $i$ is included in $K_{\bv{a}}$ with probability $\min\left(1, {m\cdot \bv{a}_i^2}/{\|\bv{a}\|_2^2}\right).$ Moreover, if $i$ is included in $K_{\bv{a}}$, it is \emph{guaranteed} to be included in $K_{\bv{b}}$ since the threshold ${m\cdot \bv{b}_i^2}/{\|\bv{b}\|_2^2}$ is at least as large as ${m\cdot \bv{a}_i^2}/{\|\bv{a}\|_2^2}$. 
% In other words, if $h(i)$ is less than $\frac{m\cdot \bv{a}_i^2}{|\bv{a}|_2^2}$, it is defintely less than $\frac{m\cdot \bv{b}_i^2}{|\bv{b}|_2^2}$.
It follows that, when $\bv{a}_i^2 \leq \bv{b}_i^2$, we have that $p_i = \min\left(1, {m\cdot \bv{a}_i^2}/{\|\bv{a}\|_2^2}\right)$. The analysis is identical for the case $\bv{b}_i^2 < \bv{a}_i^2$, in which case $p_i = \min\left(1, {m\cdot \bv{b}_i^2}/{\|\bv{b}\|_2^2}\right)$. Combining the two cases establishes \eqref{eq:prob_claim}.

% Note that the term $\frac{m\cdot \bv{a}_i^2}{|\bv{a}|_2^2}$ represents the probability that the $i$-th entry of $\bv{a}$ is sampled by the sketch. Similarly, $\frac{m\cdot \bv{b}_i^2}{|\bv{b}|_2^2}$ represents the probability that the $i$-th entry of $\bv{b}$ is sampled by the sketch. The $\min$ function is used to ensure that the probability of sampling the $i$-th entry is at most $1$, since the sketch can only contain entries from $1$ to $m$. Therefore, $\mathbbm{1}_i$ is equal to $1$ if and only if the $i$-th entry is sampled in both sketches, which happens with probability $p_i = \min\left(1, \frac{m\cdot \bv{a}_i^2}{|\bv{a}|_2^2}, \frac{m\cdot \bv{b}_i^2}{|\bv{b}|_2^2}\right)$.

Let $W$ be the estimate returned by \Cref{alg:inner_product_estimator}. We can write 
    $W = \sum_{i\in \mathcal{I}} \mathbbm{1}_i \cdot \frac{\bv{a}_i\bv{b}_i}{p_i}$, and 
applying linearity of expectation, we have:  
\begin{align}
\label{eq:expectation_bound}
    \E[W] = \sum_{i\in \mathcal{I}} p_i \cdot \frac{\bv{a}_i\bv{b}_i}{p_i} = \sum_{i\in \mathcal{I}}\bv{a}_i\bv{b}_i = \langle \bv{a}, \bv{b}\rangle.
\end{align}
Next, since each term in the sum $W = \sum_{i\in \mathcal{I}} \mathbbm{1}_i \cdot \frac{\bv{a}_i\bv{b}_i}{p_i}$ is independent,
\vspace{-0.3cm}
\begin{align*}
    \Var[W] = \sum_{i\in \mathcal{I}}\Var\left[\mathbbm{1}_i \cdot \frac{\bv{a}_i\bv{b}_i}{p_i}\right] = \sum_{i\in \mathcal{I}} \frac{(\bv{a}_i\bv{b}_i)^2}{p_i^2} \Var[\mathbbm{1}_i].
\end{align*}
$\Var[\mathbbm{1}_i] = p_i - p_i^2$, which is $0$ when $p_i$ equals~$1$. If $p_i \neq 1$, then $\Var[\mathbbm{1}_i]\leq p_i = m\cdot \min\left({\bv{a}_i^2}/{\|\bv{a}\|_2^2}, {\bv{b}_i^2}/{\|\bv{b}\|_2^2}\right)$. 
% \subm{\revised{By plugging in this, we can proved the variance bound $\max(\|\bv{a}_{\mathcal{I}}\|_2^2\|\bv{b}\|_2^2, \|\bv{a}\|_2^2\|\bv{b}_{\mathcal{I}}\|_2^2)$. The complete proof is detailed in the extended version of the paper \cite{DaliriFreireMusco:2023b}.}}
%
\subm{\revised{So we have:
\begin{align*}
    \Var[W] &\leq \sum_{i\in \mathcal{I}, p_i\neq 1} \|\bv{a}\|_2^2\|\bv{b}\|_2^2\frac{(\bv{a}_i^2/\|\bv{a}\|_2^2)(\bv{b}_i^2/\|\bv{b}\|_2^2)}{m\cdot \min({\bv{a}_i^2}/{\|\bv{a}\|_2^2}, {\bv{b}_i^2}/{\|\bv{b}\|_2^2})}\\
    &= \sum_{i\in \mathcal{I}, p_i\neq 1} \|\bv{a}\|_2^2\|\bv{b}\|_2^2\frac{\max(\bv{a}_i^2/\|\bv{a}\|_2^2,\bv{b}_i^2/\|\bv{b}\|_2^2)}{m}\\
    &\leq \frac{\|\bv{a}\|_2^2\|\bv{b}\|_2^2}{m}\sum_{i\in \mathcal{I}} \frac{\bv{a}_i^2}{\|\bv{a}\|_2^2} + \frac{\bv{b}_i^2}{\|\bv{b}\|_2^2} \\
    & = \frac{1}{m}\left(\|\bv{b}\|_2^2\|\bv{a}_{\mathcal{I}}\|_2^2 + \|\bv{a}\|_2^2\|\bv{b}_{\mathcal{I}}\|_2^2 \right).
    % &= \frac{1}{m}\left(\|\bv{a}_{\mathcal{I}}\|_2^2\|\bv{b}\|_2^2 + \|\bv{a}\|_2^2\|\bv{b}_{\mathcal{I}}\|_2^2\right) \\ 
	% &\leq \frac{2}{m} \max(\|\bv{a}_{\mathcal{I}}\|_2^2\|\bv{b}\|_2^2, \|\bv{a}\|_2^2\|\bv{b}_{\mathcal{I}}\|_2^2). 
\end{align*}
We obtain our final variance bound by upper bounding the sum by $2$x the maximum.
}}
\arxiv{Plugging in, we are able to prove our desired variance bound:
\vspace{-0.2cm}
\begin{flalign*}
    \Var[W] &\leq \sum_{i\in \mathcal{I}, p_i\neq 1} \frac{(\bv{a}_i\bv{b}_i)^2}{p_i}\\
    &= \sum_{i\in \mathcal{I}, p_i\neq 1} \|\bv{a}\|_2^2\|\bv{b}\|_2^2\frac{(\bv{a}_i^2/\|\bv{a}\|_2^2)(\bv{b}_i^2/\|\bv{b}\|_2^2)}{m\cdot \min({\bv{a}_i^2}/{\|\bv{a}\|_2^2}, {\bv{b}_i^2}/{\|\bv{b}\|_2^2})}\\
    &= \sum_{i\in \mathcal{I}, p_i\neq 1} \|\bv{a}\|_2^2\|\bv{b}\|_2^2\frac{\max(\bv{a}_i^2/\|\bv{a}\|_2^2,\bv{b}_i^2/\|\bv{b}\|_2^2)}{m}\\
    &\leq \frac{\|\bv{a}\|_2^2\|\bv{b}\|_2^2}{m}\sum_{i\in \mathcal{I}} \frac{\bv{a}_i^2}{\|\bv{a}\|_2^2} + \frac{\bv{b}_i^2}{\|\bv{b}\|_2^2} \\
    & = \frac{\|\bv{a}\|_2^2\|\bv{b}\|_2^2}{m}\left(\frac{\|\bv{a}_{\mathcal{I}}\|_2^2}{\|\bv{a}\|_2^2} + \frac{\|\bv{b}_{\mathcal{I}}\|_2^2}{\|\bv{b}\|_2^2} \right) \\
    &= \frac{1}{m}\left(\|\bv{a}_{\mathcal{I}}\|_2^2\|\bv{b}\|_2^2 + \|\bv{a}\|_2^2\|\bv{b}_{\mathcal{I}}\|_2^2\right) 
    %\leq \frac{2}{m} \max(\|\bv{a}_{\mathcal{I}}\|_2^2\|\bv{b}\|_2^2, \|\bv{a}\|_2^2\|\bv{b}_{\mathcal{I}}\|_2^2). 
\end{flalign*}
Given the non-negativity of the norms, we have:
\begin{align*}
    \Var[W] &\leq \frac{1}{m}\left(\|\bv{a}_{\mathcal{I}}\|_2^2\|\bv{b}\|_2^2 + \|\bv{a}\|_2^2\|\bv{b}_{\mathcal{I}}\|_2^2\right) \\
    &\leq \frac{2}{m} \max(\|\bv{a}_{\mathcal{I}}\|_2^2\|\bv{b}\|_2^2, \|\bv{a}\|_2^2\|\bv{b}_{\mathcal{I}}\|_2^2).
\end{align*}

This holds because the sum of two non-negative numbers is always less than or equal to twice the larger of the two.
}

Finally, we prove the claimed bound on the expected sketch size. We have that
$
    \left|K_{\bv{a}}\right| = \sum_{i=1}^n \mathbbm{1}\left[i\in K_{\bv{a}}\right],
$
where $\mathbbm{1}\left[i\in K_{\bv{a}}\right]$ is an indicator random variable that is $1$ if $i$ is included in $K_{\bv{a}}$ and zero otherwise. 
% Note that this random variable {is not} the same as $\mathbbm{1}_i$. 
By linearity of expectation, we have that:
\begin{align}
\label{eq:expected_size}
    \E\left[\left|K_{\bv{a}}\right|\right] = \sum_{i=1}^n \E\left[\mathbbm{1}\left[i\in K_{\bv{a}}\right] \right] &= \sum_{i=1}^n \min(1, {m\cdot \bv{a}_i^2}/{\|\bv{a}\|_2^2}) \leq m.
\end{align}
An identical analysis shows that $\E\left[\left|K_{\bv{b}}\right|\right] \leq m$, which completes the proof. \subm{In the extended version of this paper \cite{DaliriFreireMusco:2023b}}\arxiv{In \Cref{{app:adaptive_and_high_prob}}}, we further prove that $\left|K_{\bv{a}}\right|$ and $\left|K_{\bv{b}}\right|$ are less than $m + O(\sqrt{m})$ with high probability.
\end{proof}

% \vspace{-.5em}
\noindent\textbf{Practical Implementation.} In \Cref{thm:main}, we show that the expected sketch size is \emph{upper bounded} by $m$. As apparent from \eqref{eq:expected_size}, it will be less than $m$ whenever there are entries in the vector for which $\bv{a}_i^2/\|\bv{a}\|_2^2 > 1/m$. This is not ideal: we would like a sketch whose size is as close to our budget $m$ as possible. Fortunately, Threshold Sampling \revised{can be} modified so that the {expected} sketch size is \emph{exactly} $m$. We simply use binary search to compute $m'$ such that $\sum_{i=1}^n \min\left(1, {m'\cdot \bv{a}_i^2}/{\|\bv{a}\|_2^2}\right) = m$. Then, we replace $m$ in Lines 3 and 6 of \cref{alg:threshold_sampling} with $m'$. \revised{Doing so does not increase our estimator's variance.}
\revised{Further details are provided in \subm{the extended version of this paper \cite{DaliriFreireMusco:2023b}}\arxiv{\Cref{app:adaptive_and_high_prob}}, and} we implement this variant of Threshold Sampling for our experiments in \Cref{sec:experiments}.

\section{\prioritysampling}
\label{sec:prioritysampling}
While a simple and effective method for inner product sketching, one limitation of Threshold Sampling is that the user cannot exactly control the size of the sketch $\mathcal{S}(\bv{a})$. We address this issue by analyzing an alternative algorithm based on Priority Sampling, a technique introduced in computer science by \cite{DuffieldLundThorup:2004}, and studied in  statistics under the name ``Sequential Poisson Sampling'' \cite{Ohlsson:1998}. 

\smallskip
\noindent\textbf{Sketching.} To motivate the method, observe from rearranging Lines 3 and 4 in \Cref{alg:threshold_sampling}, that Threshold Sampling selects all entries from $\bv{a}$ for which ${h(i)}/{a_i^2}$ falls below a fixed ``global threshold'', $\tau_{\bv{a}} = {m}/{\|\bv{a}\|_2^2}$. There will be at most $m$ such values in expectation, but there could be more or less depending on the randomness in $h$. Priority Sampling (\Cref{alg:priority_sampling}) removes this variability by simply selecting the $m$ \emph{smallest} values of ${h(i)}/{a_i^2}$. It then treats the $(m+1)^\text{st}$ smallest value as the global threshold $\tau_{\bv{a}}$. 

\smallskip\noindent\textbf{Estimation.} \revised{Given sketches} %NOTE(aecio): fixed typo
$\mathcal{S}(\bv{a})$ and $\mathcal{S}(\bv{b})$ computed using Priority Sampling, we can actually use the exact same estimator for $\langle \bv{a}, \bv{b}\rangle$ as Threshold Sampling (\Cref{alg:inner_product_estimator}). In particular,
\begin{align}
\label{eq:est_sum}
W = \sum_{i \in K_{\bv{a}} \cap K_{\bv{b}}} \frac{\bv{a}_i\bv{b}_i}{\min(1, \bv{a}_i^2 \cdot \tau_{\bv{a}}, \bv{b}_i^2 \cdot \tau_{\bv{b}})}
\end{align}
% (computed on Line 2 of \Cref{alg:inner_product_estimator}) remains an unbiased estimate for the inner product. However, analyzing the variance of the estimator is a lot trickier. Notably, we no longer have that the summation terms in \eqref{eq:est_sum} are independent; they all depend on the \emph{same} {random} numbers $\tau_\bv{a}$ and $\tau_\bv{b}$, which were previously fixed quantities for Threshold Sampling. Moreover, bounding the variance of each term in the sum is complicated by the presence of random variables in the denominator. These issues arise in other applications of Priority Sampling. For example, for the problem of subset-sum estimation (which Priority Sampling was originally developed for \cite{DuffieldLundThorup:2004}), an optimal variance analysis proved elusive, until finally 
% being given in a tour de force result by Szegedy \cite{AlonDuffieldLund:2005,Szegedy:2006}.
(computed on Line 2 of \Cref{alg:inner_product_estimator}) remains an unbiased estimate for the inner product. However, analyzing the variance of the estimator is a lot trickier. Notably, we no longer have that the summation terms in \eqref{eq:est_sum} are independent; they all depend on the \emph{same} {random} numbers $\tau_\bv{a}$ and $\tau_\bv{b}$, which were previously fixed quantities for Threshold Sampling. Moreover, bounding the variance of each term in the sum is complicated by the presence of random variables in the denominator. These issues arise in earlier applications of Priority Sampling, like subset-sum estimation \cite{DuffieldLundThorup:2004}. For this problem, an optimal variance analysis proved elusive, until finally 
being given in a tour de force result by Szegedy \cite{AlonDuffieldLund:2005,Szegedy:2006}.

\begin{figure}[t]
\vspace{-1em}
\begin{algorithm}[H]
    \caption{Priority Sampling}
    \label{alg:priority_sampling}
    \begin{algorithmic}[1]
                \Require Length $n$ vector $\bv{a}$, random seed $s$, target sketch size $m$.
        \Ensure Sketch $\mathcal{S}(\bv{a}) = \{K_{\bv{a}}, V_{\bv{a}}, \tau_{\bv{a}}\}$, where $K_{\bv{a}}$ is a subset of indices from $\{1, \ldots, n\}$ and $V_{\bv{a}}$ contains $\bv{a}_i$ for all $i\in K_{\bv{a}}$.
        \algrule
        \State Use random seed $s$ to select a uniformly random hash function $h: \{1,..., n\}\rightarrow [0,1]$. Initialize $K_{\bv{a}}$ and $V_{\bv{a}}$ to be empty lists.
        \State Compute rank $R_i = {h(i)}/{\bv{a}_i^2}$ for all $i$ such that $\bv{a}_i\neq 0$.
        \State Set $\tau_{\bv{a}}$ equal to the $(m+1)^{\text{st}}$ smallest value $R_{i}$, or set $\tau_{\bv{a}} = \infty$ if $\bv{a}$ has less than $m+1$ non-zero values.
        \For{$i$ such that $\bv{a}_i\neq 0$}
                \If{$R_i < \tau_{\bv{a}}$}
              \State Append $i$ to $K_{\bv{a}}$, append $\bv{a}_i$ to $V_{\bv{a}}$.
              \EndIf
        \EndFor
        \State \Return $\mathcal{S}(\bv{a}) = \{K_{\bv{a}}, V_{\bv{a}},\tau_{\bv{a}}\}$
    \end{algorithmic}
    \vspace{-.2em}
\end{algorithm}
\vspace{-2em}
\end{figure}

\smallskip
\noindent\textbf{Theoretical Analysis.} Drawing inspiration from a new analysis of Priority Sampling for the subset sum problem \cite{DaliriFreireMusco:2023}, we are able to overcome these obstacles for inner product estimation as well.
Our main theoretical result on Priority Sampling is as follows:
\begin{theorem}\label{thm:main_priority}
For vectors $\bv a,\bv b \in \R^n$ and sketch size $m$, let $\mathcal{S}(\bv a)=\{K_{\bv{a}}, V_{\bv{a}}, \tau_{\bv{a}}\}$ and $\mathcal{S}(\bv b)=\{K_{\bv{b}}, V_{\bv{b}}, \tau_{\bv{b}}\}$ be sketches returned by \Cref{alg:priority_sampling}. Let $W$ be the inner product estimate returned by \Cref{alg:inner_product_estimator} applied to these sketches.  We have that $\E\left[W\right] = \langle\bv{a},\bv{b}\rangle$ and
\begin{align*}
\Var\left[W\right] &\leq \frac{2}{m-1} \max\left(\|\bv{a}_{\mathcal{I}}\|_2^2\|\bv{b}\|_2^2, \|\bv{a}\|_2^2\|\bv{b}_{\mathcal{I}}\|_2^2\right)
\end{align*}
Moreover, let $|K_\bv{a}|$ and $|K_\bv{b}|$ be the number of index/values pairs stored in $\mathcal{S}(\bv a)$ and $\mathcal{S}(\bv b)$. We have $|K_\bv{a}|\leq m$ and $|K_\bv{b}|\leq m$, with equality in the typical case when $\bv{a}$ and $\bv{b}$ have at least $m$ non-zero entries. 
\end{theorem}
\Cref{thm:main_priority} almost exactly matches our \Cref{thm:main} for Threshold Sampling, except that the leading constant on the variance is $\frac{2}{m-1}$ instead of $\frac{2}{m}$. Again, we can apply Chebyshev's inequality to conclude that if we set $m = \frac{2/\delta}{\epsilon^2} + 1$, then $|W - \langle \bv{a}, \bv{b}\rangle|$ is bounded by $\epsilon \cdot \max\left(\|\bv{a}_\mathcal{I}\|_2\|\bv{b}\|_2, \|\bv{a}\|_2\|\bv{b}_\mathcal{I}\|_2\right)$ with probability $\geq 1-\delta$. The matching theoretical results align with experiments: as seen in \cref{sec:experiments},  Priority Sampling performs almost identically to \thresholdsampling, albeit with the added benefit of a fixed sketch size bound.

\begin{proof}[Proof of \cref{thm:main_priority}]
We start by introducing additional notation. Let $\mathcal{A} = \{i : \bv{a}_i \neq 0\}$ denote the set of indices where $\bv{a}$ is non-zero and let $\mathcal{B} = \{i : \bv{b}_i \neq 0\}$ denote the set of indices where $\bv{b}$ is non-zero.  Recall that $\tau_{\bv{a}}$ as computed in \Cref{alg:priority_sampling} is the $(m+1)^\text{st}$ smallest value of $h(i)/a_i^2$ over all $i\in \mathcal{A}$. For any $i\in \mathcal{A}$, let $\tau_{\bv{a}}^i$ denote the $m^\text{th}$ smallest of $h(j)/a_j^2$ over all $j\in \mathcal{A}\setminus\{i\}$. If $\mathcal{A}\setminus\{i\}$ has fewer than $m$ values, define $\tau_{\bv{a}}^i = \infty$. Define $\tau_{\bv{b}}^i$ analogously for all $i \in \mathcal{B}$. Let $\mathcal{T} = K_{\bv{a}} \cap K_{\bv{b}}$ be as in \Cref{alg:inner_product_estimator}. Later on we will use the easily checked fact that, for all $i \in \mathcal{T}$, $\tau_{\bv{a}}^i = \tau_{\bv{a}}$ and $\tau_{\bv{b}}^i = \tau_{\bv{b}}$.

\revised{
	The estimate $W$ returned by \Cref{alg:inner_product_estimator} can be rewritten as:
\begin{align}
	\label{eq:w_sum_of_var}
	W &=\hspace{-.7em} \sum_{i \in \mathcal{A}\cap\mathcal{B}} \hat{w}_{i} &
	\text{where \quad} \hat{w}_{i} = \left\{\begin{matrix}\frac{\bv{a}_i\bv{b}_i}{\min(1,\bv{a}_i^2 \tau_{\bv{a}},\bv{b}_i^2\tau_{\bv{b}})}  &i\in \mathcal{T}
 \\0  &i\notin \mathcal{T}.
\end{matrix}\right.
\end{align}
}
From \eqref{eq:w_sum_of_var}, we can see that, to prove $\E[W] = \langle \bv{a}, \bv{b}\rangle = \sum_{i \in \mathcal{A}\cap\mathcal{B}} \bv{a}_i\bv{b}_i$, it suffices to prove that, for all $i \in \mathcal{A}\cap\mathcal{B}$, $\E[\hat{w}_{i}] = \bv{a}_i \bv{b}_i$.
To establish this equality, first observe that for $i$ to be in $\mathcal{T}$, it must be that both ${h(i)}/{\bv{a}_i^2}$ and ${h(i)}/{\bv{b}_i^2}$ are among the $m^\text{th}$ smallest ranks computed when sketching $\bv{a}$ and $\bv{b}$, respectively. %In other words, it must be that ${h(i)}/{\bv{a}_i^2}< \tau_{\bv{a}}^i$ and ${h(i)}/{\bv{b}_i^2}< \tau_{\bv{b}}^i$. So, $\Pr\left[i \in \mathcal{T}\right]$ equals:
In other words, it must be that ${h(i)}/{\bv{a}_i^2}< \tau_{\bv{a}}^i$ and ${h(i)}/{\bv{b}_i^2}< \tau_{\bv{b}}^i$. So, conditioning on $\tau_{\bv{a}}^i$ and $\tau_{\bv{b}}^i$,
  %   \revised{\begin{align*}
		% \Pr\left[i \in \mathcal{T}\right] = \Pr\left[{h(i)}/{\bv{a}_i^2}< \tau_{\bv{a}}^i \cap {h(i)}/{\bv{b}_i^2}< \tau_{\bv{b}}^i\right] = \min(1,\bv{a}_i^2\tau_{\bv{a}}^i,\bv{a}_i^2\tau_{\bv{b}}^i).
  %   \end{align*}}
    {\begin{align*}
		\Pr\left[i \in \mathcal{T}\mid \tau_{\bv{a}}^i,\tau_{\bv{b}}^i\right] &= \Pr\left[{h(i)}/{\bv{a}_i^2}< \tau_{\bv{a}}^i \cap {h(i)}/{\bv{b}_i^2}< \tau_{\bv{b}}^i\right]\\ &= \min(1,\bv{a}_i^2\tau_{\bv{a}}^i,\bv{a}_i^2\tau_{\bv{b}}^i).
    \end{align*}}
	\revised{
    Combined with the fact discussed earlier that, conditioned on $i\in \mathcal{T}$, $\tau_{\bv{a}} = \tau_{\bv{a}}^i$ and $\tau_{\bv{b}} = \tau_{\bv{b}}^i$, we have:
    % \begin{align*}
    %     \E[\hat{w}_{i}]
    %     &= \frac{\bv{a}_i \bv{b}_i}{\min(1,\bv{a}_i^2\tau_{\bv{a}},\bv{a}_i^2\tau_{\bv{b}})}  \min(1,\bv{a}_i^2\tau_{\bv{a}}^i,\bv{a}_i^2\tau_{\bv{b}}^i) = \bv{a}_i \bv{b}_i.
    % \end{align*}
    \begin{align*}
        \E[\hat{w}_{i}]
        &= \E_{\tau_{\bv{a}}^i,\tau_{\bv{b}}^i}\big[\frac{\bv{a}_i \bv{b}_i}{\min(1,\bv{a}_i^2\tau_{\bv{a}},\bv{a}_i^2\tau_{\bv{b}})}  \min(1,\bv{a}_i^2\tau_{\bv{a}}^i,\bv{a}_i^2\tau_{\bv{b}}^i)\big] = \bv{a}_i \bv{b}_i.
    \end{align*}
	}
    As desired, $\E[W] = \langle \bv{a}, \bv{b}\rangle$  follows by linearity of expectation.

    Next, we turn our attention to bounding the variance of $W$. As discussed, this is complicated by the fact that $\hat{w}_i$ and $\hat{w}_j$ are non-independent. However, it is possible to show that the random variables are \revised{\emph{pairwise uncorrelated}, which will allow us to apply linearity of variance to the sum in \eqref{eq:w_sum_of_var}. I.e., we want to show that, for all $i,j$, $\E[\hat{w}_{i}\hat{w}_{j}] = \E[\hat{w}_{i}]\E[\hat{w}_{j}]$.}
For any $i,j \in \mathcal{A}$ define $\tau_{\bv{a}}^{i,j}$ to equal the $(m-1)^\text{st}$ smallest of $h(k)/a_k^2$ over all $k\in \mathcal{A}\setminus\{i,j\}$, or $\infty$ if there are not $m-1$ values in $\mathcal{A}\setminus\{i,j\}$. Define $\tau_{\bv{b}}^{i,j}$ analogously for $i,j \in \mathcal{B}$. 
%\revised{As in our expression for $\Pr\left[i\in \mathcal{T}\right]$, it can be seen} that $\Pr\left[i,j \in \mathcal{T}\right] = \min(1,\bv{a}_i^2\tau_{\bv{a}}^{i,j},\bv{b}_i^2\tau_{\bv{b}}^{i,j})\cdot \min(1,\bv{a}_j^2\tau_{\bv{a}}^{i,j},\bv{b}_j^2\tau_{\bv{b}}^{i,j})$. \revised{Furthermore, conditioned on $i,j \in \mathcal{T}$, $\tau_{\bv{a}}^{i,j} = \tau_{\bv{a}}$ and $\tau_{\bv{b}}^{i,j} = \tau_{\bv{b}}$.So,}
{As in our expression for $\Pr\left[i\in \mathcal{T}\right]$, it can be seen} that $\Pr[i,j \in \mathcal{T}\mid \tau^{i,j}_{\bv{a}}, \tau^{i,j}_{\bv{b}}] = \min(1,\bv{a}_i^2\tau_{\bv{a}}^{i,j},\bv{b}_i^2\tau_{\bv{b}}^{i,j})\cdot \min(1,\bv{a}_j^2\tau_{\bv{a}}^{i,j},\bv{b}_j^2\tau_{\bv{b}}^{i,j})$. {Furthermore, conditioned on $i,j \in \mathcal{T}$, $\tau_{\bv{a}}^{i,j} = \tau_{\bv{a}}$ and $\tau_{\bv{b}}^{i,j} = \tau_{\bv{b}}$.
So,}
    % \begin{align*}
    %     \E[\hat{w}_{i}\hat{w}_{j}]
    %     &= \frac{\bv{a}_i \bv{b}_i}{\min(1,\bv{a}_i^2\tau_{\bv{a}},\bv{b}_i^2\tau_{\bv{b}})} \frac{\bv{a}_j \bv{b}_j}{\min(1,\bv{a}_j^2\tau_{\bv{a}},\bv{b}_j^2\tau_{\bv{b}})} \Pr\left[i,j \in \mathcal{T}\right]\\
    %     &= \bv{a}_i \bv{b}_i \bv{a}_j \bv{b}_j 
    %     =  \E[\hat{w}_{i}]\E[\hat{w}_{j}],
    % \end{align*}

    \begin{align*}
        \E[\hat{w}_{i}\hat{w}_{j}]
        &= \E_{\tau^{i,j}_{\bv{a}}, \tau^{i,j}_{\bv{b}}}\Bigg[\frac{\bv{a}_i \bv{b}_i}{\min(1,\bv{a}_i^2\tau^{i,j}_{\bv{a}},\bv{b}_i^2\tau^{i,j}_{\bv{b}})} \frac{\bv{a}_j \bv{b}_j}{\min(1,\bv{a}_j^2\tau^{i,j}_{\bv{a}},\bv{b}_j^2\tau^{i,j}_{\bv{b}})} \cdot \ldots  \\
        &\ldots \Pr\left[i,j \in \mathcal{T}\mid \tau^{i,j}_{\bv{a}}, \tau^{i,j}_{\bv{b}}\right]\Bigg]
        = \bv{a}_i \bv{b}_i \bv{a}_j \bv{b}_j 
        =  \E[\hat{w}_{i}]\E[\hat{w}_{j}],
    \end{align*}
\revised{as desired.}
\revised{
Since $\E[\hat{w}_{i}\hat{w}_{j}] = \E[\hat{w}_{i}]\E[\hat{w}_{j}]$ for all $i,j$ we can apply linearity of variance to conclude that $\Var[W] = \sum_{i \in \mathcal{A}\cap\mathcal{B}} \Var[\hat{w}_{i}]$.}

So, it suffices to establish individual bounds on $\Var[\hat{w}_{i}]$ for $i \in \mathcal{A}\cap \mathcal{B}$. To do so, first observe that,
conditioned on $\tau_{\bv{a}}^i$ and  $\tau_{\bv{b}}^i$,
    % \begin{align*}
    %     \E\left[\hat{w}_{i}^2\mid \tau_{\bv{a}}^i, \tau_{\bv{b}}^i\right] &= \left(\frac{\bv{a}_i \bv{b}_i}{\min(1,\bv{a}_i^2\tau_{\bv{a}}^i,\bv{b}_i^2\tau_{\bv{b}}^i)}\right)^2 \cdot \Pr[i \in \mathcal{T}]\\ 
    %     &= \frac{\bv{a}_i^2 \bv{b}_i^2}{\min(1,\bv{a}_i^2\tau_{\bv{a}}^i,\bv{b}_i^2\tau_{\bv{b}}^i)} = \bv{a}_i^2 \bv{b}_i^2 \max\left(1,\frac{1}{\bv{a}_i^2\tau_{\bv{a}}^i},\frac{1}{\bv{b}_i^2\tau_{\bv{b}}^i}\right).
    % \end{align*}
    \begin{align*}
        \E\left[\hat{w}_{i}^2\mid \tau_{\bv{a}}^i, \tau_{\bv{b}}^i\right] &= \left(\frac{\bv{a}_i \bv{b}_i}{\min(1,\bv{a}_i^2\tau_{\bv{a}}^i,\bv{b}_i^2\tau_{\bv{b}}^i)}\right)^2 \cdot \Pr\left[i \in \mathcal{T}\mid \tau_{\bv{a}}^i,\tau_{\bv{b}}^i\right]\\ 
        &= \frac{\bv{a}_i^2 \bv{b}_i^2}{\min(1,\bv{a}_i^2\tau_{\bv{a}}^i,\bv{b}_i^2\tau_{\bv{b}}^i)} = \bv{a}_i^2 \bv{b}_i^2 \max\left(1,\frac{1}{\bv{a}_i^2\tau_{\bv{a}}^i},\frac{1}{\bv{b}_i^2\tau_{\bv{b}}^i}\right).
    \end{align*}
	\revised{
    We can thus write $\Var\left[\hat{w}_i\right] = \E\left[\hat{w}_i^2\right] - \E\left[\hat{w}_i\right]^2 = \E\left[\hat{w}_i^2\right] - \bv{a}_i^2\bv{b}_i^2 $ as:
        \begin{align*}
            \Var\left[\hat{w}_i\right] &= \bv{a}_i^2\bv{b}_i^2 \E\left[\max\hspace{-.2em}\left(1,\frac{1}{\bv{a}_i^2\tau_{\bv{a}}^i},\frac{1}{\bv{b}_i^2\tau_{\bv{b}}^i}\right)\right] - \bv{a}_i^2\bv{b}_i^2 \\ 
            &= \bv{a}_i^2\bv{b}_i^2 \E\left[\max\hspace{-.2em}\left(0,\frac{1}{\bv{a}_i^2\tau_{\bv{a}}^i}-1,\frac{1}{\bv{b}_i^2\tau_{\bv{b}}^i}-1\right)\right]\\
            % &\leq \bv{a}_i^2\bv{b}_i^2 \E\left[\max\hspace{-.2em}\left(0,\frac{1}{\bv{a}_i^2\tau_{\bv{a}}^i},\frac{1}{\bv{b}_i^2\tau_{\bv{b}}^i} \right)\right]\\
			&\leq \bv{a}_i^2\bv{b}_i^2 \E\left[\frac{1}{\bv{a}_i^2\tau_{\bv{a}}^i} + \frac{1}{\bv{b}_i^2\tau_{\bv{b}}^i} \right]
            % &\leq \bv{a}_i^2 \E_{\tau_{\bv{b}}^i}\left[ \frac{1}{t'}\right] + \bv{b}_i^2 \E_{\tau_{\bv{a}}^i}\left[\frac{1}{t} \right]\\
            % &\leq \bv{a}_i^2 \int_{0}^{\infty} \frac{1}{t'} \Pr[\tau_{\bv{b}}^i = t']dt' + \bv{b}_i^2 \int_{0}^{\infty} \frac{1}{t} \Pr[\tau_{\bv{a}}^i = t]dt \\
            =
            \bv{a}_i^2 \E\left[\frac{1}{\tau_{\bv{b}}^i}\right] + 
            \bv{b}_i^2 \E\left[\frac{1}{\tau_{\bv{a}}^i}\right].
        \end{align*}
	}
    So, we have reduced the problem to bounding the expected inverse of $\tau_{\bv{a}}^i$ and $\tau_{\bv{b}}^i$. Doing so is not straightforward: these are complex random variables that depend on all entries in $\bv{a}$ and $\bv{b}$, respectively. However, it was recently shown in \cite{DaliriFreireMusco:2023} (Claim 5) that $\E[{1}/{\tau_{\bv{a}}^i}] \leq \|\bv{a}\|_2^2/{(m-1)}$ and  $\E[{1}/{\tau_{\bv{b}}^i}] \leq {\|\bv{b}\|_2^2}/{(m-1)}$.
    Finally, we have:
    \begin{align*}
        \Var[W] = \sum_{i \in A\cap B} \Var[\hat{w}_{i}] &\leq \sum_{i \in A\cap B}\bv{a}_i^2 \E\left[{1}/{\tau_{\bv{b}}^i}\right] + \bv{b}_i^2 \E\left[{1}/{\tau_{\bv{a}}^i}\right]\\
        &\leq \sum_{i \in A\cap B} \bv{a}_i^2 \frac{\|\bv{b}\|^2}{m-1} + \bv{b}_i^2 \frac{\|\bv{a}\|^2}{m-1}\\
         &= \frac{1}{m-1} \left(\|\bv{a}_{\mathcal{I}}\|_2^2\|\bv{b}\|_2^2 +  \|\bv{a}\|_2^2\|\bv{b}_{\mathcal{I}}\|_2^2\right).
    \end{align*}
    Noting that for any $c, d$, $c+d \leq 2\max(c,d)$ completes the proof. 
\end{proof}

\section{Join-Correlation Estimation}
\label{sec:applications}
In addition to our theoretical results, we perform an empirical evaluation of Threshold and Priority Sampling for inner product sketching. One of our main motivating applications is \emph{join-correlation estimation}~\cite{SantosBessaChirigati:2021, EsmailoghliQuiane-RuizAbedjan:2021}. This problem has previously been addressed using (unweighted) consistent sampling methods, like the KMV sketch \cite{SantosBessaChirigati:2021, SantosBessaMusco:2022}. In this section, we show how it can be solved using \emph{any} inner product sketching algorithm in a black-box way, expanding the toolkit of methods that can be applied to the task.

\smallskip\noindent\textbf{Problem Statement.} The join-correlation problem consists of computing the Pearson's correlation coefficient between two data columns that originally reside in different data tables. Specifically, we are interested in the correlation between values that would appear in the columns \emph{after} performing an (inner) join on the tables, i.e., values for which the same key appears in both tables. We call this quantity the \emph{post-join correlation}, or simply the \emph{join-correlation}. As a concrete illustration, consider the example tables in Figure~\ref{fig:example-tables-join-correlation}(a). The goal of join-correlation estimation is to approximate the correlation $\rho_{\bv{x}, \bv{y}}$ between the vectors $\bv x$ and $\bv y$ from $\mathcal{T}_{A\bowtie B}$. 

The join-correlation problem arises in dataset search applications, where the goal is to  discover new data to augment  a query dataset, e.g., to improve predictive models \cite{LiuChaiLuo:2022, IonescuHaiFragkoulis:2022,ChepurkoMarcusZgraggen:2020}. 
In such applications, we typically want to estimate join-correlation for columns in a query table and those in a large collection of other data tables. Accordingly, the brute-force approach that explicitly joins tables and computes the correlation between attributes is infeasible. 

Prior work proposes to use sketching as an efficient alternative. The idea is to pre-process (i.e., sketch) the collection of tables in advance, so that join-correlation between columns in any two tables $\mathcal{T}_{A}$ and $\mathcal{T}_B$ can be evaluated \emph{without explicitly materializing the join ${A\bowtie B}$}. 
Specifically, Santos et al.~\cite{SantosBessaChirigati:2021} propose an extension of KMV sketches that uniformly samples entries from each table, and then uses the join between the sketches to estimate correlation. Unfortunately, just like inner product estimation, this approach can suffer when $\mathcal{T}_{A}$ and $\mathcal{T}_{B}$ contain entries with widely varying magnitude: larger entries often contribute more to the correlation, but are not selected with higher probability by the KMV sketch.

% \subsection{Text Similarity Estimation}
% In natural language processing, cosine similarity is a widely used similarity measure used to compare pairs of text documents. In this setting, a document is usually represented using term vectors that encode the frequency (or importance) of each term (word) in the document.
% Such vectors are usually sparse because a document only contains a small subset of all the terms in a language's vocabulary. %In other words, most words in a language do not appear in a given document, and their frequency of occurrence is zero. 
% Therefore, the term vector for a document is mostly composed of zero values, except for the few words that appear in the document and their corresponding frequencies.  This sparsity suggests that sketching approaches sensitive to index overlap like \thresholdsampling might perform well to summarize these vectors. The cosine similarity of two term vectors $\bv{a}$ and $\bv{b}$ is computed as:
% $\cos(\bv{a},\bv{b}) = 
% \langle \bv{a} , \bv{b} \rangle
% /
% ||\bv{a}||_2 ||\bv{b}||_2
% $.
% Therefore, when the vectors are normalized to have unit norm, i.e., $||\bv{a}||_2 =1 $ and $||\bv{b}||_2=1$, the cosine similarity becomes a simple inner product between the two vectors. In order to compute the cosine similarity in our experiments in Section~\ref{sec:experiments}, we simply normalize each vector to unit-norm before computing its sketch.

\begin{figure}[t]
\centering
\footnotesize
\begin{subfigure}[b]{\linewidth}
\parbox[t]{.25\linewidth}{
  \centering
  \begin{tabular}[t]{|c|c|}
    \multicolumn{2}{c}{$\mathcal{T}_A$}        \\
    \hline
    \textbf{$\bv{k_a}$} & \textbf{$\bv{a}$} \\
    \hline
    \rowcolor{gray!30} 3  & 2.5  \\
    6  & 2.3  \\
    \rowcolor{gray!30} 8 & 4   \\
    \rowcolor{gray!30} 11 & 0.5  \\
    \rowcolor{gray!30} 13 & 3    \\
    16 & -3.7 \\
    \hline
  \end{tabular}
}
% \hspace{1em}
\parbox[t]{.25\linewidth}{
    \centering
    \begin{tabular}[t]{|c|c|}
    \multicolumn{2}{c}{$\mathcal{T}_B$}        \\
    \hline
    \textbf{$\bv{k_b}$} & $\bv{b}$ \\
    \hline
    \rowcolor{gray!30} 3  & -3.1  \\
    7  & 0.4  \\
    \rowcolor{gray!30} 8 & -4.2 \\
    10 & 1.5  \\
    \rowcolor{gray!30} 11 & 1    \\
    \rowcolor{gray!30} 13 & -2.6    \\
    14 & -5.9 \\
    \hline
    \end{tabular}
}
% \hspace{1em}
\parbox[t]{.4\linewidth}{
    \centering
    \begin{tabular}[t]{|c|c|c|}
    \multicolumn{3}{c}{$\mathcal{T}_{A\bowtie B}$}        \\
    \hline
    \textbf{$\bv k_{a \bowtie b}$} & \textbf{$\bv x$} & {$\bv y$} \\
    \hline
    3	& 2.5	& -3.1 \\
    8	& 4	& -4.2 \\
    11	& 0.5	& 1 \\
    13	& 3    & -2.6 \\
    \hline
    \end{tabular}
    % \vspace{0.3em}
    
    % \texttt{SIZE}($V_{A \bowtie B}$) = 4 \\
    % \texttt{SUM}($V_{A}^2$) = 59.48 \\
    % \texttt{SUM}($V_{B}^2$) = 78.47 \\
    % \texttt{DOT\_PRODUCT}($A, B$) = 25.85 \\
} 
\vspace{-.25em}
\caption{The table $\mathcal{T}_{A \bowtie B}$ is the output of a join between tables $\mathcal{T}_A$ and $\mathcal{T}_B$. 
% The size of the table resulting from the join of two tables is equal to $4$. 
The goal of join-correlation estimation is to approximate the Pearson's correlation between the second two columns in $\mathcal{T}_{A \bowtie B}$.}
\label{fig:table-join-example}
\end{subfigure}

\vspace{-.5em}
\begin{subfigure}[b]{\linewidth}
\centering
\parbox{1.0\linewidth}{
\centering
\setlength{\tabcolsep}{0.35em}
\begin{tabular}{|c|c|c|c|c|c|c|c|c|c|c|c|c|c|c|c|c|}
\multicolumn{10}{c}{}        \\
    \hline
    index & \bf 1 & \bf 2 & \bf 3 & \bf 4 & \bf 5 & \bf 6 & \bf 7 & \bf 8 & \bf 9 & \bf 10 & \bf 11 & \bf 12 & \bf 13 & \bf 14 & \bf 15 & \bf 16 \\
    \hline \hline
    $\bv{a}$ & 0 & 0 & 2.5 & 0 & 0 & 2.3 & 0 & 4 & 0 & 0 & 0.5 & 0 & 3 & 0 & 0 & -3.7\\
    \hline
    $\bv{a}^2$ & 0 & 0 & 6.25 & 0 & 0 & 5.29 & 0 & 16 & 0 & 0 & .25 & 0 & 9 & 0 & 0 & 13.69\\
    \hline
    $\bv{1_a}$ & 0 & 0 & 1 & 0 & 0 & 1 & 0 & 1 & 0 & 0 & 1 & 0 & 1 & 0 & 0 & 1\\
    \hline
    \hline
    $\bv{b}$ & 0 & 0 & -3.1 & 0 & 0 & 0 & 0.4 & -4.2 & 0 & 1.5 & 1 & 0 & -2.6 & -5.9 & 0 & 0\\
    \hline
    $\bv{b}^2$ & 0 & 0 & 9.61 & 0 & 0 & 0 & .16 & 17.64 & 0 & 2.25 & 1 & 0 & 6.76 & 34.81 & 0 & 0\\
    \hline
    $\bv{1_b}$ & 0 & 0 & 1 & 0 & 0 & 0 & 1 & 1 & 0 & 1 & 1 & 0 & 1 & 1 & 0 & 0\\
    \hline
\end{tabular}
}
\vspace{-.25em}
\caption{
We define six sparse vectors $\bv{a}$, $\bv{a}^2$, $\bv{1_a}$, $\bv{b}$, $\bv{b}^2$, and $\bv{1_b}$ that encode the information in $\mathcal{T}_A$ with $\mathcal{T}_B$. In \cref{eq:final_jc_est}, we show how to express the join-correlation as a combination of inner products involving these vectors, which can be estimated with a sketching method.
}
\label{fig:table-join-vector}
\end{subfigure}

\vspace{-.3cm}
% \vspace{-0.5em}
\caption{
Join-Correlation via inner product sketching. 
}
\vspace{-.5em}
\label{fig:example-tables-join-correlation}
\end{figure}

\myparagraph{Join-Correlation via Inner Product Sketching}
%In this section, 
We show an alternative approach for attacking the join-correlation problem by reducing it to inner product estimation. The reduction allows us to take advantage of sketches like \revised{\wmh, Threshold Sampling, and Priority Sampling,} which naturally make use of weighted sampling. 

Referring again to Figure~\ref{fig:example-tables-join-correlation}(a), consider the vectors $\bv x$ and $\bv y$ from $\mathcal{T}_{A\bowtie B}$. Let $\overline{x}$ (resp. $\overline{y}$) denote the mean of $\bv x$ (resp. $\bv y$), $n$ denote the length of the vectors (number of rows in $\mathcal{T}_{A\bowtie B}$), $\Sigma_{\bv{x}}$ (resp. $\Sigma_{\bv{y}}$) denote the summation of all values in $\bv{x}$ (resp. $\bv{y}$), and $\Sigma_{\bv{x}^2}$ (resp. $\Sigma_{\bv{y}^2}$) denote the summation of all squared values of $\bv{x}$ (resp.~$\bv{y}$).
It can be verified that correlation coefficient between $\bv x$ and $\bv y$ can be rewritten as: \vspace{-.1cm}
\begin{align}\label{eq:correlation}
\rho_{\bv{x}, \bv{y}} 
= 
\frac{
    \langle \bv{x} - \overline{x}, \bv{y} - \overline{y} \rangle
}{
    \|\bv{x} - \overline{x}\|_2 \|\bv{y} - \overline{y}\|_2
}
= 
\frac{
    n \langle \bv{x}, \bv{y} \rangle - \Sigma_{\bv{x}} \Sigma_{\bv{y}}
}{
    \sqrt{
        n \Sigma_{\bv{x}^2} - \Sigma_{\bv{x}}^2
    }
    \sqrt{
        n \Sigma_{\bv{y}^2} - \Sigma_{\bv{y}}^2
    }
}.
\end{align}
Our observation is that all of the values in \cref{eq:correlation} can be computed using only inner product operations over vectors derived from tables $\mathcal{T}_{A}$ and $\mathcal{T}_{B}$ independently. The vectors are shown in Figure~\ref{fig:example-tables-join-correlation}(b): 
vectors $\bv{a}$ and $\bv{b}$ contain the values, with $\bv{a}_i$ (resp. $\bv{b}_i$) set to zero if key $i$ was not present in table $\mathcal{T}_{A}$ (resp. table $\mathcal{T}_{B}$). Vectors $\bv{1_a}$ and $\bv{1_b}$ are indicator vectors for the corresponding join keys in each table. Finally, $\bv{a}^2$ and $\bv{b}^2$ are  equal to $\bv{a}$ and $\bv{b}$ with an entrywise square applied. 
%$K_\bv{a}$ and $K_\bv{b}$, respectively. 
Using these vectors, we can compute all components of the correlation formula as inner products:
\begin{align*}
n &= \langle \bv{1_a}, \bv{1_b} \rangle, &
\Sigma_{\bv{x}} 
&= \langle \bv{a}, \bv{1_b} \rangle, &
\Sigma_{\bv{y}} 
&= \langle \bv{1_a}, \bv{b} \rangle, \\
\langle \bv{x}, \bv{y} \rangle 
&= \langle \bv{a}, \bv{b} \rangle,  &
\Sigma_{\bv{x}^2}  
&= \langle \bv{a}^2, \bv{1_b} \rangle, &
\Sigma_{\bv{y}^2} 
&= \langle \bv{1_a}, \bv{b}^2 \rangle.
\end{align*}

\noindent
In particular, we can rewrite $\rho_{\bv{x,y}}$ equivalently as:
\begin{equation}
\label{eq:final_jc_est}
\frac{
    \langle \bv{a}, \bv{b} \rangle \langle \bv{1_a}, \bv{1_b} \rangle - \langle \bv{a}, \bv{1_b} \rangle \langle \bv{1_a}, \bv{b} \rangle
}{
    \sqrt{
        \left(
            \langle \bv{1_a}, \bv{1_b} \rangle \langle \bv{a}^2, \bv{1_b} \rangle - \langle \bv{a}, \bv{1_b} \rangle^2
        \right)
        \left( 
            \langle \bv{1_a}, \bv{1_b} \rangle \langle \bv{b}^2, \bv{1_a} \rangle - \langle \bv{b}, \bv{1_a} \rangle^2
        \right)
    }
}.
\end{equation}
Given this formula, we can use any inner product sketching method to approximate join-correlation. In particular, given $\mathcal{T}_A$, we compute three separate sketches, one for each of $\bv{a}$, $\bv{a^2}$, $\bv{1_a}$. When combined with sketches for $\bv{b}$, $\bv{b^2}$, $\bv{1_b}$, we can estimate all of the inner products in \eqref{eq:final_jc_est} separately, and combine them to obtain an estimate for $\rho_{\bv{x}, \bv{y}}$. 

\revised{For data discovery, the vectors described above are often extremely sparse with limited overlap between non-zero entries. Therefore, they are amenable to the sampling-based sketches studied in this paper, and benefit from our improvements over \eqref{eq:jl_guar}. In particular, the length of $\bv{a}$, $\bv{a^2}$, and $\bv{1_a}$ equals the total universe of possible keys, while the number of non-zeros in these vectors equals the number of keys in $\mathcal{T}_{A}$. The overlap between the non-zeros in $\bv{a}$, $\bv{a^2}$, and $\bv{1_a}$,  and those in $\bv{b}$, $\bv{b^2}$, and $\bv{1_b}$ is equal to the number of keys in common between $\mathcal{T}_{A}$ and $\mathcal{T}_{B}$, which can be very small. As an example, consider a data augmentation task where were wish to join a query data table, $\mathcal{T}_{A}$, with keys that are addresses in a single neighborhood to a statewide database of addresses in $\mathcal{T}_{B}$.}

\smallskip
\noindent\textbf{Optimization for Sampling-Based Sketches.}
In \Cref{sec:experiments}, we use the approach above to estimate correlation using linear sketching methods like CountSketch and JL. Given sketch size budget  $m$, we allocate $m/3$ space to sketching each of the three vectors $\bv{a}$, $\bv{a^2}$, and $\bv{1_a}$. Our final join-correlation sketch is then the concatenation of the equally sized sketches $\mathcal{S}(\bv{a})$, $\mathcal{S}(\bv{a}^2)$, and $\mathcal{S}(\bv{1_a})$. We take roughly the same approach \revised{for Threshold and Priority Sampling.} However, in a sampling-based sketch, if we select index $i$ when sketching \emph{any} of the three vectors $\bv{1_a}$, $\bv{a}$, and $\bv{a^2}$, then we might as well use the index in estimating inner products involving \emph{all} three. In particular, by storing the single key/value pair $(i,\bv{a}_i)$, we can compute the information $(i,1)$, $(i,\bv{a}_i)$, and $(i,\bv{a}_i^2)$ needed to estimate all inner products. We take advantage of this fact to squeeze additional information out of our sketches.  Details of the resulting optimized approach are included in \subm{the extended version of this paper \cite{DaliriFreireMusco:2023b}}\arxiv{\Cref{app:correlation}}. 
% and code implementing it is included in the supplemental material. 

\vspace{-.2cm}
\section{Experiments}
\label{sec:experiments}

\noindent\textbf{Baselines.} We assess the performance of our methods by comparing them to representative baselines\revised{, all of which were implemented in Python.} We include both sampling and linear sketching methods for inner product estimation:
\begin{itemize}[leftmargin=12pt]
    \item \textbf{\textit{Johnson-Lindenstrauss Projection (JL):}} For this \emph{linear sketch}, we use a dense random matrix  $\bs{\Pi}$ with scaled $\pm 1$ entries, which is equivalent to the AMS sketch~\cite{AlonMatiasSzegedy:1999,Achlioptas:2003}.
    \item \revised{\textbf{\textit{CountSketch (CS):}} The classic \textit{linear sketch} introduced in~\cite{CharikarChenFarach-Colton:2002}, and also studied under the name Fast-AGMS sketch in \cite{CormodeGarofalakis:2005}. We use one repetition of the sketch.\footnote{\revised{While prior work suggests partitioning the sketch budget and taking the median of multiple independent estimators \cite{LarsenPaghTetek:2021}, we found that doing so slightly decreased accuracy in our experiments.}}
    }
    % Prior work on using CountSketch to estimate inner products suggests splitting the sketch budget into 3, and taking the median result from the three sketches \cite{LarsenPaghTetek:2021}. We found this approach was only beneficial for extremely large sketch sizes, so only use a single sketch in our experiments.
    \item \textbf{\textit{Weighted MinHash Sampling (\wmhsexperiments):}} The method described in \cite{BessaDFMMSZ:2023}, which is the first sketch with tighter theoretical bounds than linear sketching for inner product estimation.
    \item \textbf{\textit{MinHash Sampling (MH):}} Also described in \cite{BessaDFMMSZ:2023}, MH is similar to Weighted MinHash, but indices are sampled uniformly at random from $\bv{a}$, not with probability proportional to $\bv{a}_i^2$.  
    \item \textbf{\textit{Uniform Priority Sampling (\PSuniform):}} The same as our Priority Sampling method, but the rank of index $i$ in \Cref{alg:priority_sampling} is chosen without taking the squared magnitude $\bv{a}_i^2$ into account, so indices are sampled uniformly. This method is equivalent to the KMV-based inner product sketch implemented in \cite{BessaDFMMSZ:2023}.
    \item \textbf{\textit{Uniform Threshold Sampling (\TSuniform):}} The same as our Threshold Sampling method, but $\bv{a}_i^2$ is not taken into account when computing $\tau_i$, so indices are sampled uniformly. 
\end{itemize}
To distinguish from the uniform sampling versions, our proposed Threshold and Priority Sampling methods are called \textbf{\TSweighted} and \textbf{\PSweighted} in the remainder of the section. In addition to the baselines above, we implemented and performed initial experiments using the End-Biased sampling method from \cite{EstanNaughton:2006}, which is equivalent to Threshold Sampling (\Cref{alg:threshold_sampling}), but with probability proportional to ${|\bv{a}_i|}/{\|\bv{a}\|_1}$. More details on how to implement this method, as well as \TSuniform and \PSuniform are included in \subm{the extended version of this paper \cite{DaliriFreireMusco:2023b}}\arxiv{\Cref{app:alternative_sampling_prob}}. As shown in \Cref{exp:synth_data}, End-Biased sampling performed slightly worse than our version of Threshold Sampling, which also enjoys stronger theoretical guarantees. So, we excluded End-Biased sampling from the majority of our experiments for conciseness and plot clarity.
\revised{We also note that there are other versions of linear sketching designed to speed up computation time in comparison to the classic JL/AMS sketch \cite{Achlioptas:2003, RusuDobra:2008}. We focus on CountSketch/Fast-AGMS because it is one of the most widely studied of these methods, and runs in $O(n)$ time with a small constant factor. As such, it offers a challenging baseline for our sampling methods in terms of computational efficiency.}

\myparagraph{Storage Size}
For linear sketches, we store the output of the matrix multiplication $\bs{\Pi}\bv{a}$ as 64-bit doubles. For sampling-based sketches, both samples (64-bit doubles) and hash values (32-bit ints) need to be stored.
As a result, a sampling sketch with $m$ samples takes $1.5x$ as much space as a linear sketch with $m$ entries.
In our experiments, \emph{storage size} denotes the total number of bits in the sketch divided by 64, i.e., the total number of doubles that the sketch equates to.
% For all methods besides \thresholdsampling, the storage size is fixed.
Storage size is fixed for all methods except \thresholdsampling, for which we report the expected storage size. 
We note that there are variants of linear sketching that further compress $\bs{\Pi}\bv{a}$ by thresholding or rounding its entries, e.g., SimHash~\cite{Charikar:2002} and quantized JL methods \cite{Jacques:2015}. While an interesting topic for future study, we do not evaluate these methods because quantization can be used to reduce the sketch size of \emph{all methods}. For instance, for sampling-based sketches, we do not need to store full 64-bit doubles. Evaluating optimal quantization strategies is beyond the scope of this work. 

\myparagraph{Estimation Error}
To make it easier to compare across different datasets, when estimating inner products, we define the following error measure: the absolute difference between ground truth inner product $\langle \bv{a}, \bv{b}\rangle$ and the estimate, scaled by $1/\|\bv{a}\|_2\|\bv{b}\|_2$. Given that most methods tested (except the uniform sampling methods) achieve an error guarantee at least as good as \cref{eq:jl_guar}, this scaling roughly ensures that reported errors lie between $0$ and $1$.
% When reporting estimation error across multiple pairs of vectors, we compute the averaged scaled difference. 

%\subsection{Accuracy vs. Storage: Synthetic Data}
\vspace{-.2cm}
\subsection{Estimation Accuracy for Synthetic Data}
\label{exp:synth_data}
\noindent\textbf{Synthetic Data.} We ran experiments on synthetic data to validate the performance of 
our methods in a controlled setting.
To contrast the behavior of linear sketching and weighted sampling methods like \wmhsexperiments, \TSweighted, and \PSweighted, we generate vector pairs $\bv{a},\bv{b}$ with varying amounts of overlap, $\mathcal{I}$, between their non-zero entries ($1\%$ to $100\%$). 
This allows us to verify our theoretical results: when $|\mathcal{I}|$ is large, we expect linear sketching and sampling to perform similarly since the linear sketching error bound of $\epsilon \|\bv{a}\|_2\|\bv{b}\|_2$ is closer to our bound of $\epsilon \cdot \max(\|\bv{a}\|_2\|\bv{b}_{\mathcal{I}}\|_2, \|\bv{a}\|_2\|\bv{b}_{\mathcal{I}}\|_2)$. When $|\mathcal{I}|$ is small, we expect a bigger difference.
 
We generate 100 pairs of synthetic vectors, \revised{each with 100,000 entries, 20,000 of which are non-zero.}
The locations of non-zero entries are randomly selected with a specific overlap $\mathcal{I}$, and their values are uniformly drawn from the interval $[-1,1]$.
Then, \revised{$2\%$ of entries are chosen randomly as outliers.} We include outliers to differentiate the performance of weighted sampling methods from \revised{their uniform counterparts (MH, \TSuniform and \PSuniform). If all entries have similar magnitude, weighted and uniform sampling are essentially the same.} \revised{Outliers are chosen to be uniform random numbers between $0$ and $10$}, which are fairly moderate values. For datasets with even larger outliers, we expect an even more pronounced difference between weighted and unweighted sampling.

\begin{figure}[t]
	\centering
        \begin{subfigure}{.9\columnwidth} 
		\centering
		\includegraphics[width=1.0\linewidth]{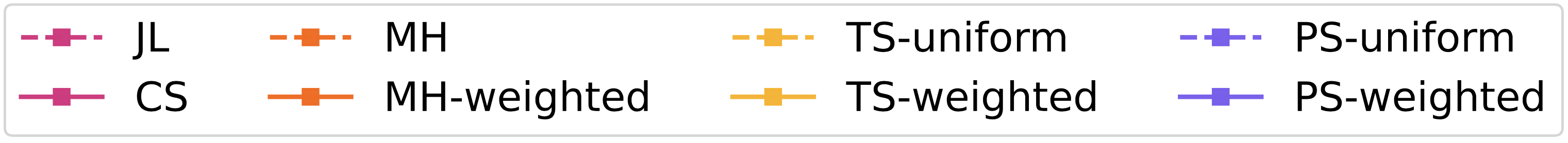}
	\end{subfigure}
 
	\begin{subfigure}{.49\columnwidth} 
		\centering
		\includegraphics[width=1.0\linewidth]{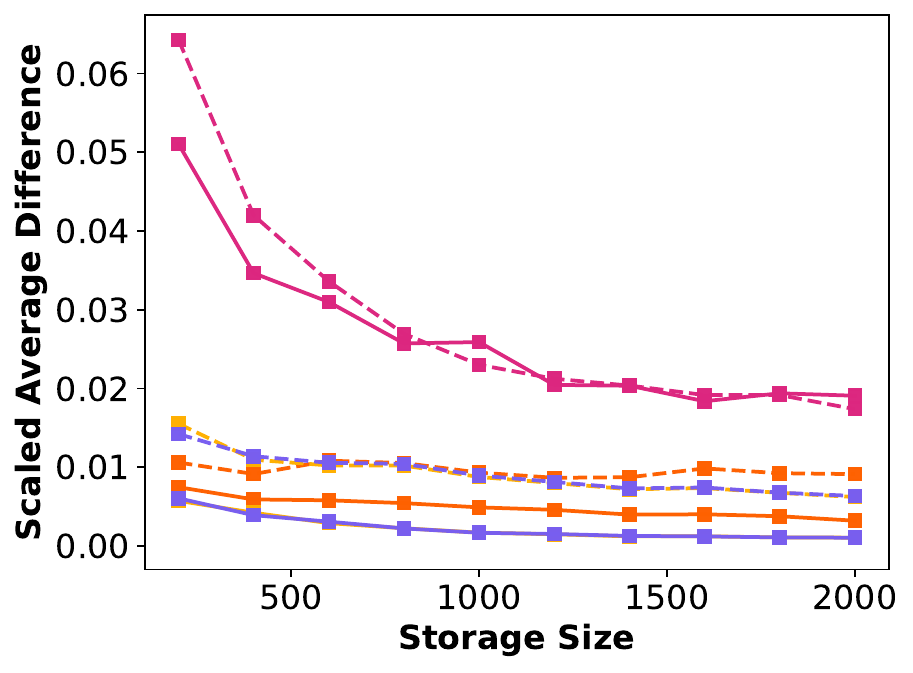}
		\vspace{-1.8em}
		\caption{1\% overlap}
	\end{subfigure}
	\begin{subfigure}{.49\columnwidth} 
		\centering
		\includegraphics[width=1.0\linewidth]{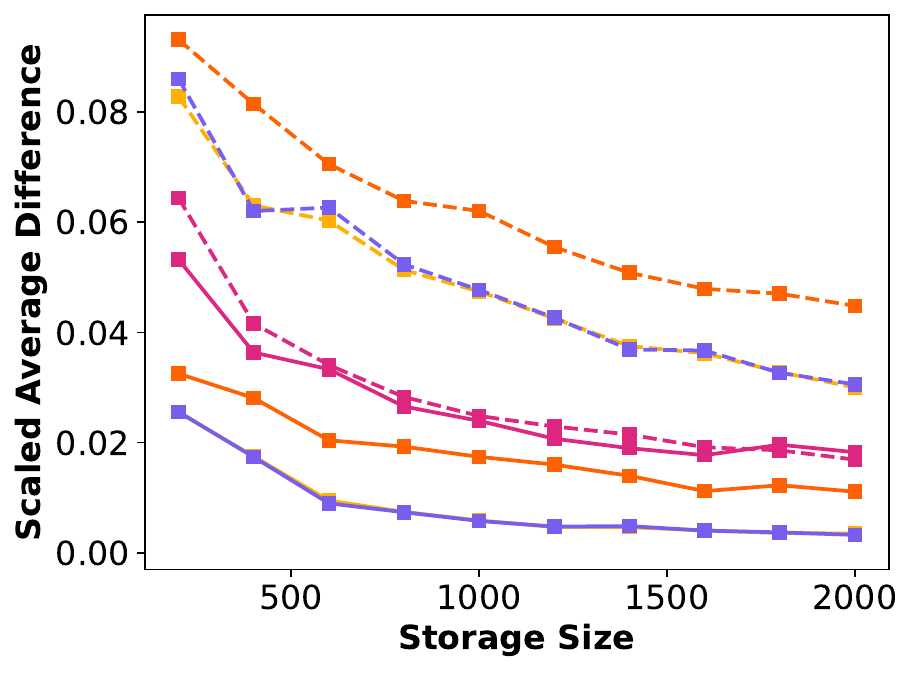}
		\vspace{-1.8em}
		\caption{10\% overlap}
	\end{subfigure}%
 
	\begin{subfigure}{.49\columnwidth} 
		\centering
		\includegraphics[width=1.0\linewidth]{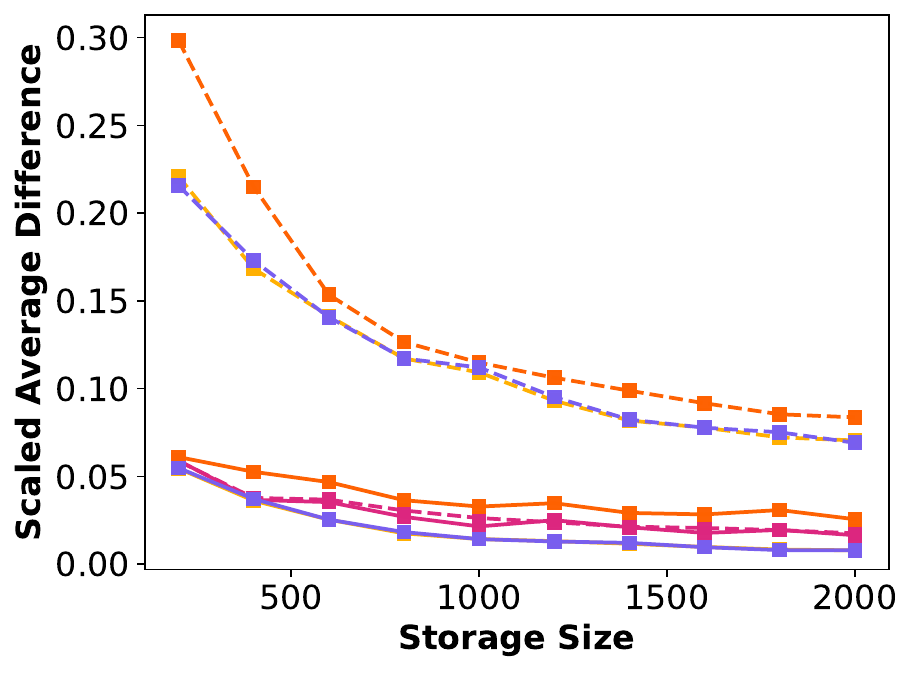}
				\vspace{-1.8em}
		\caption{50\% overlap}
	\end{subfigure}
	\begin{subfigure}{.49\columnwidth} 
		\centering
		\includegraphics[width=1.0\linewidth]{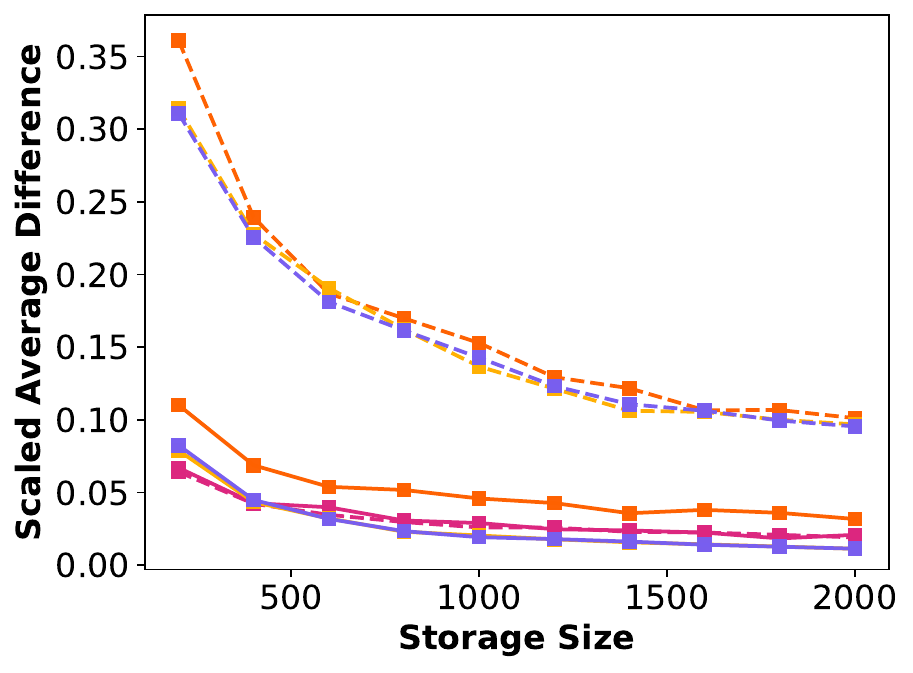}
				\vspace{-1.8em}
		\caption{100\% overlap}
	\end{subfigure}
\vspace{-.75em}
	\caption{\revised{
 Inner product estimation for real-valued synthetic data. The lines for \PSuniform and \TSuniform overlap, as do the lines for our \PSweighted and \TSweighted methods. As predicted by our theoretical results, \PSweighted and \TSweighted consistently outperform all other baselines.
 }
 \vspace{-.2cm}
 }
	\label{fig:InnerProductEst_Synthetic}
\vspace{-.4em}
\end{figure}

% \begin{figure}[t]
% 	\centering
%         \begin{subfigure}{.9\columnwidth} 
% 		\centering
% 		\includegraphics[width=1.0\linewidth]{figs/legend_ip.png}
% 	\end{subfigure}
 
% 	\begin{subfigure}{.49\columnwidth} 
% 		\centering
% 		\includegraphics[width=1.0\linewidth]{fig_large_vector/ip_diff-overlap_0.01+outlier_0.1+corr_0.7+mode_ip+t_1+iteration_100+20231204220147.pdf}
% 		\vspace{-1.8em}
% 		\caption{1\% overlap}
% 	\end{subfigure}
% 	\begin{subfigure}{.49\columnwidth} 
% 		\centering
% 		\includegraphics[width=1.0\linewidth]{fig_large_vector/ip_diff-overlap_0.1+outlier_0.1+corr_0.7+mode_ip+t_1+iteration_100+20231205022733.pdf}
% 		\vspace{-1.8em}
% 		\caption{10\% overlap}
% 	\end{subfigure}%
 
% 	\begin{subfigure}{.49\columnwidth} 
% 		\centering
% 		\includegraphics[width=1.0\linewidth]{fig_large_vector/ip_diff-overlap_0.5+outlier_0.1+corr_0.7+mode_ip+t_1+iteration_100+20231205064224.pdf}
% 				\vspace{-1.8em}
% 		\caption{50\% overlap}
% 	\end{subfigure}
% 	% \begin{subfigure}{.49\columnwidth} 
% 	% 	\centering
% 	% 	\includegraphics[width=1.0\linewidth]{}
% 	% 			\vspace{-1.8em}
% 	% 	\caption{100\% overlap}
% 	% \end{subfigure}
% \vspace{-.5em}
% 	\caption{Inner product estimation for real-valued synthetic data (20,000 non-zero elements). 
%  }
% 	% \label{fig:InnerProductEst_Synthetic}
% \vspace{-.4em}
% \end{figure}

\vspace{-.15cm}
\subsubsection{Inner Product Estimation}
\Cref{fig:InnerProductEst_Synthetic} shows the scaled average difference between the actual and estimated inner product for the different techniques. 
%As expected, error decreases with increasing storage size all methods.
The plot is consistent with our theoretical findings: \TSweighted and \PSweighted are more accurate than all other methods for all levels of overlap. They consistently outperform the prior state-of-the-art sampling sketch, \wmhsexperiments.
For very low overlap, even unweighted sampling methods  (MH, \TSuniform, and \PSuniform) outperform linear sketches (JL, CS), but this advantage decreases as overlap increases. 
Note that when overlap  is above $50\%$, the performance of linear sketching is comparable to \wmhsexperiments. However, our proposed methods, \TSweighted and \PSweighted, continue to outperform linear sketching, even in this regime.

\begin{figure}[t]
	\vspace{.59em}
	\centering
	\begin{subfigure}{.9\columnwidth} 
		\centering
		\includegraphics[width=1.0\linewidth]{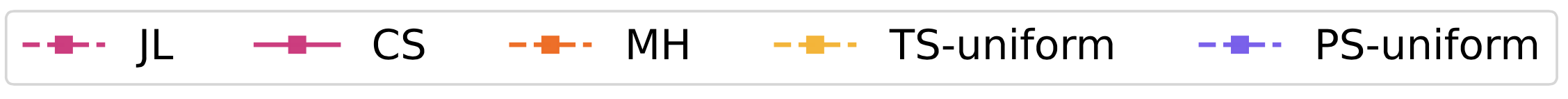}
	\end{subfigure}
	
	\begin{subfigure}{.49\columnwidth} 
		\centering
		\includegraphics[width=1.0\linewidth]{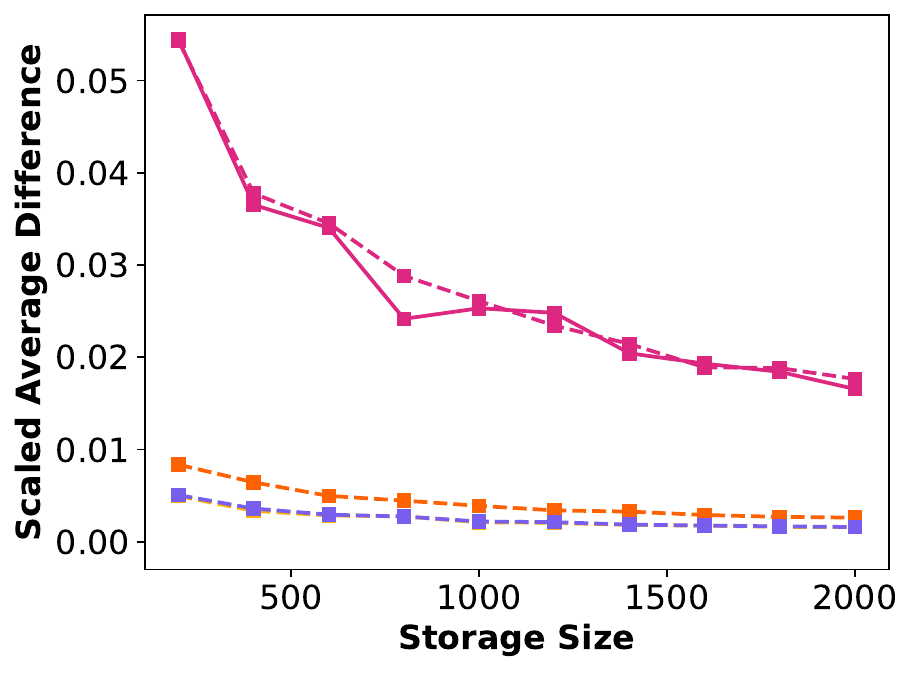}
		\vspace{-1.8em}
		\caption{1\% overlap}
	\end{subfigure}
	\begin{subfigure}{.49\columnwidth} 
		\centering
		\includegraphics[width=1.0\linewidth]{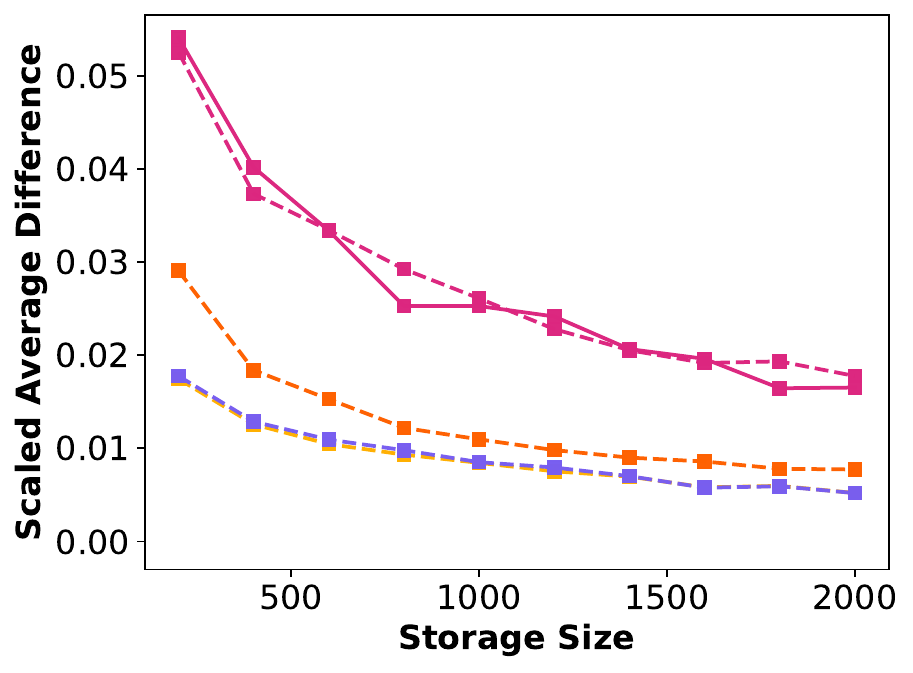}
		\vspace{-1.8em}
		\caption{10\% overlap}
	\end{subfigure}%
	
	\begin{subfigure}{.49\columnwidth} 
		\centering
		\includegraphics[width=1.0\linewidth]{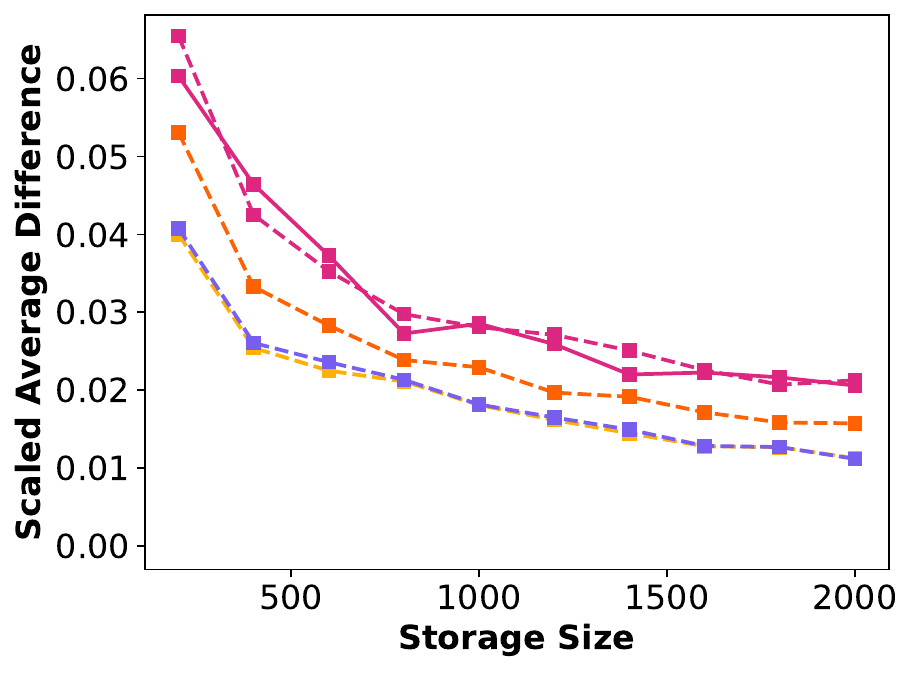}
				\vspace{-1.8em}
		\caption{50\% overlap}
	\end{subfigure}
	\begin{subfigure}{.49\columnwidth} 
		\centering
		\includegraphics[width=1.0\linewidth]{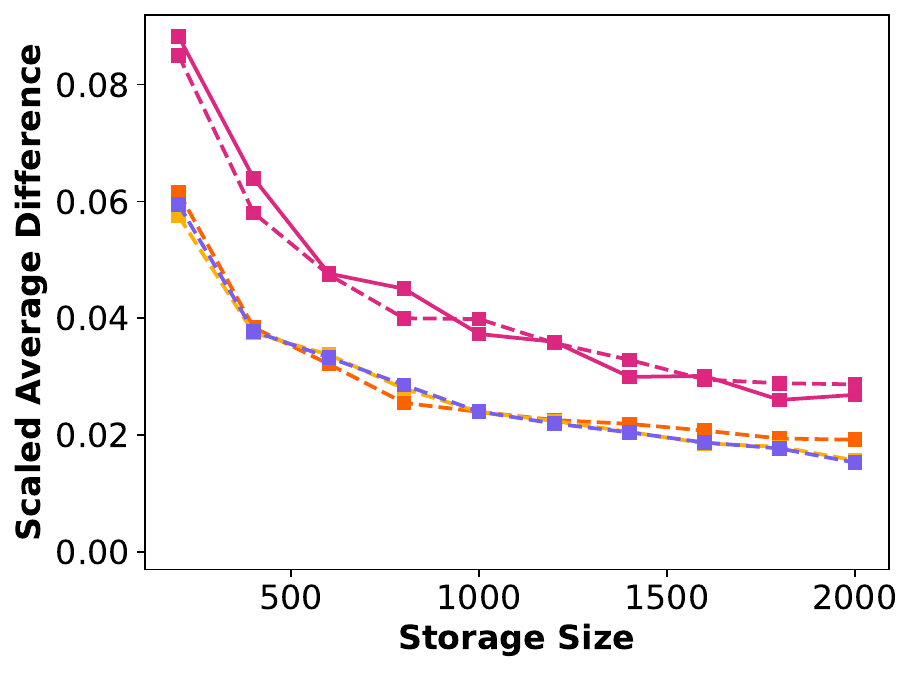}
				\vspace{-1.8em}
		\caption{100\% overlap}
	\end{subfigure}
\vspace{-.75em}
	\caption{\revised{
 Inner product estimation for synthetic binary data. Weighted sampling methods are excluded since they are equivalent to their unweighted counterparts for binary vectors. Our \PSuniform and \TSuniform methods outperform both linear sketches and MH for computing inner products.
 }
 }
	\label{fig:InnerProductEst_Synthetic_Binary}
\vspace{-1em}
\end{figure}

% \begin{figure}[t]
% 	\vspace{.59em}
% 	\centering
% 	\begin{subfigure}{.9\columnwidth} 
% 		\centering
% 		\includegraphics[width=1.0\linewidth]{figs/legend_binary.png}
% 	\end{subfigure}
	
% 	\begin{subfigure}{.49\columnwidth} 
% 		\centering
% 		\includegraphics[width=1.0\linewidth]{fig_large_vector/ip_diff-overlap_0.01+outlier_0.1+corr_0.7+mode_join_size+t_1+iteration_100+20231204220144.pdf}
% 		\vspace{-1.8em}
% 		\caption{1\% overlap}
% 	\end{subfigure}
% 	\begin{subfigure}{.49\columnwidth} 
% 		\centering
% 		\includegraphics[width=1.0\linewidth]{fig_large_vector/ip_diff-overlap_0.1+outlier_0.1+corr_0.7+mode_join_size+t_1+iteration_100+20231204232750.pdf}
% 		\vspace{-1.8em}
% 		\caption{10\% overlap}
% 	\end{subfigure}%
	
% 	\begin{subfigure}{.49\columnwidth} 
% 		\centering
% 		\includegraphics[width=1.0\linewidth]{fig_large_vector/ip_diff-overlap_0.5+outlier_0.1+corr_0.7+mode_join_size+t_1+iteration_100+20231205005315.pdf}
% 				\vspace{-1.8em}
% 		\caption{50\% overlap}
% 	\end{subfigure}
% 	\begin{subfigure}{.49\columnwidth} 
% 		\centering
% 		\includegraphics[width=1.0\linewidth]{fig_large_vector/ip_diff-overlap_1.0+outlier_0.1+corr_0.7+mode_join_size+t_1+iteration_100+20231205021853.pdf}
% 				\vspace{-1.8em}
% 		\caption{100\% overlap}
% 	\end{subfigure}
% \vspace{-.5em}
% 	\caption{Inner product estimation for synthetic binary data (20,000 non-zero elements). }
% 	% \label{fig:InnerProductEst_Synthetic_Binary}
% \vspace{-1em}
% \end{figure}

\vspace{.2em}
\subsubsection{Binary Inner Product Estimation} 

\revised{
We also evaluate inner product estimation for binary $\{0,1\}$ vectors, which can be applied to problems like join size estimation for tables with unique keys~\cite{CormodeGarofalakis:2016} and set intersection estimation. Set intersection has been studied e.g., for applications like document similarity estimation \cite{LiKonig:2011,PaghStockelWoodruff:2014,Broder:1997}.}
We use the same synthetic data as before, except that all non-zero entries are set to 1. 
Results are presented in \Cref{fig:InnerProductEst_Synthetic_Binary}. Note that weighted sampling methods (\wmh, \TSweighted, and \PSweighted) are not included because they are exactly equivalent to the unweighted methods for binary vectors.
All of the sampling methods clearly outperform linear sketching, and the gap is most significant when the overlap is small, as predicted by our theoretical results.

%\vspace{.2em}
\vspace{-.15cm}
\subsubsection{Join-Correlation Estimation} 
As discussed in \Cref{sec:applications}, post-join correlation estimation can be cast as an inner product estimation problem involving three vectors derived from a data column, which we denote $\bv{a}$, $\bv{a}^2$, and $\bv{1_{a}}$. We do not explicitly construct synthetic database columns but instead generate vectors $\bv{a}$ and $\bv{b}$ as before, and derive $\bv{a}^2$, $\bv{1_{a}}$, $\bv{b}^2$, and $\bv{1_{b}}$ based on them. \revised{We set the overlap between vector pairs to $10\%$ and} control the correlation between the vectors (which are generated randomly) using a standard regression-based method for adjusting correlation~\cite{Howell:2018}.
For the linear sketching methods, we split the storage size evenly among the sketches for all three vectors and estimate correlation as discussed in \Cref{sec:applications}. For the uniform sampling methods (MH, \TSuniform, and \PSuniform), we instead follow the approach from \cite{SantosBessaChirigati:2021}, computing a single sketch for each of $\bv{a}$ and $\bv{b}$ and then estimating the empirical correlation of the sampled entries. For \TSweighted and \PSweighted, we use our new method described in \Cref{sec:applications}.% and \Cref{app:correlation}.

As \Cref{fig:JoinCorrEst_Synthetic} shows, MH, \TSuniform, and \PSuniform perform well despite the lack of weighted sampling. This is consistent with observations in prior work on the \PSuniform (KMV) method~\cite{SantosBessaChirigati:2021}. 
Even without weighting, these sketches are able to benefit from the advantage of data sparsity, unlike linear sketches.
Nonetheless, our \TSweighted and \PSweighted outperform all other approaches in terms of accuracy vs. sketch size. We note that we use the optimized variants of these methods discussed in \Cref{sec:applications}.
%  \revised{and \subm{the extended version of this paper \cite{DaliriFreireMusco:2023b}}\arxiv{\Cref{app:correlation}}}.

\vspace{-.15cm}
\subsubsection{Comparison to End-Biased Sampling}
\begin{figure}[t]
	\centering
	\centering
	\begin{subfigure}{.99\columnwidth} 
		\centering
		\includegraphics[width=1.0\linewidth]{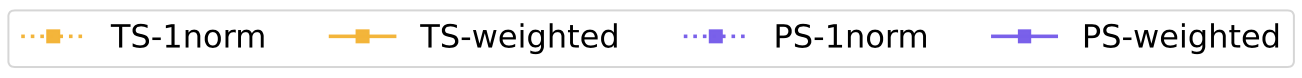}
	\end{subfigure}
	
	\begin{subfigure}{.49\columnwidth} 
		\centering
		\includegraphics[width=1.0\linewidth]{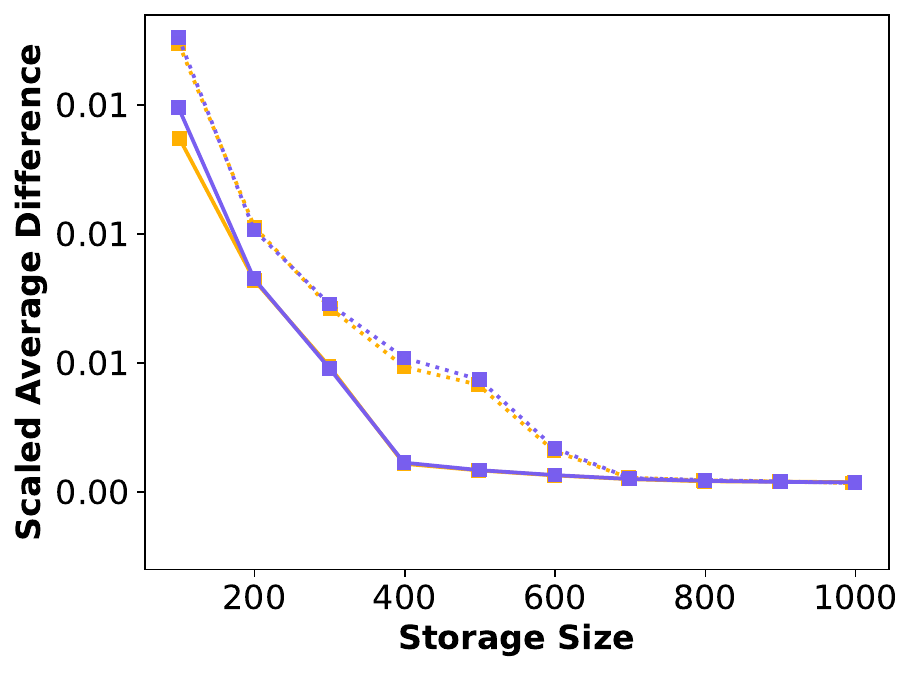}
		\vspace{-1.8em}
		\caption{1\% overlap}
	\end{subfigure}
	\begin{subfigure}{.49\columnwidth} 
		\centering
		\includegraphics[width=1.0\linewidth]{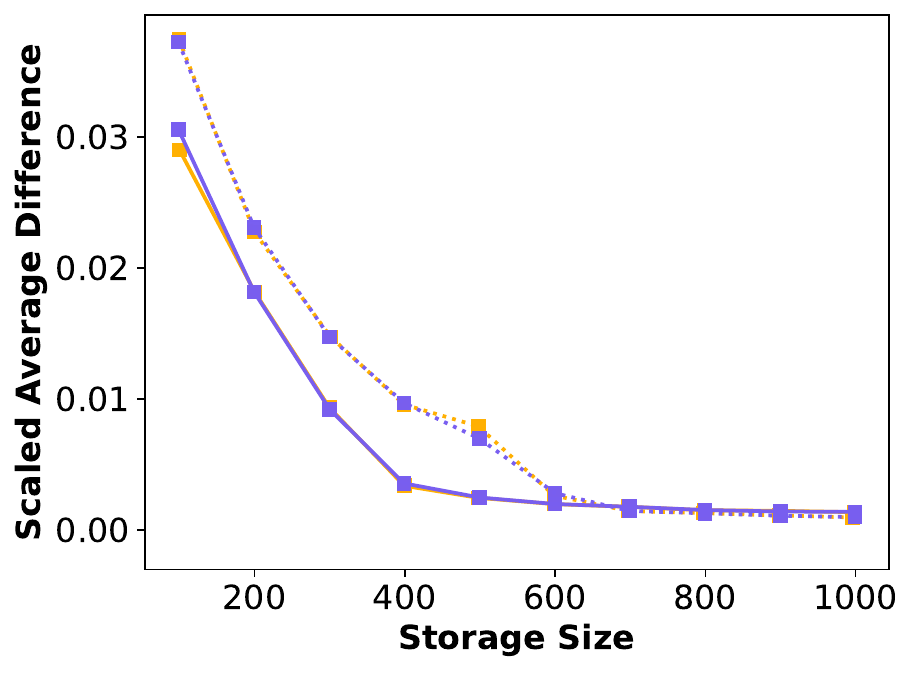}
		\vspace{-1.8em}
		\caption{10\% overlap}
	\end{subfigure}%
	
	\begin{subfigure}{.49\columnwidth} 
		\centering
		\includegraphics[width=1.0\linewidth]{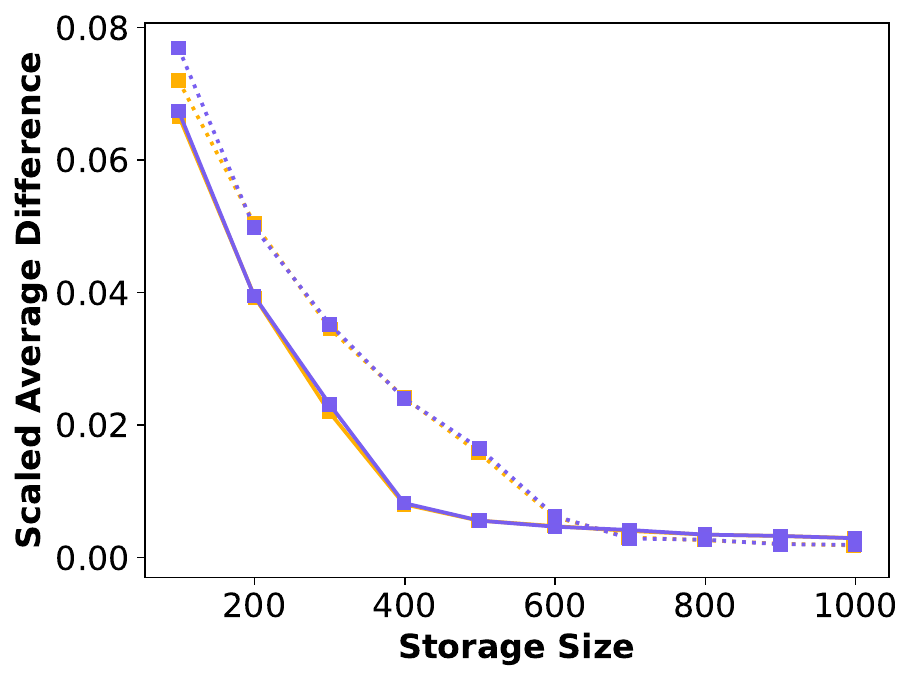}
		\vspace{-1.8em}
		\caption{50\% overlap}
	\end{subfigure}
	\begin{subfigure}{.49\columnwidth} 
		\centering
		\includegraphics[width=1.0\linewidth]{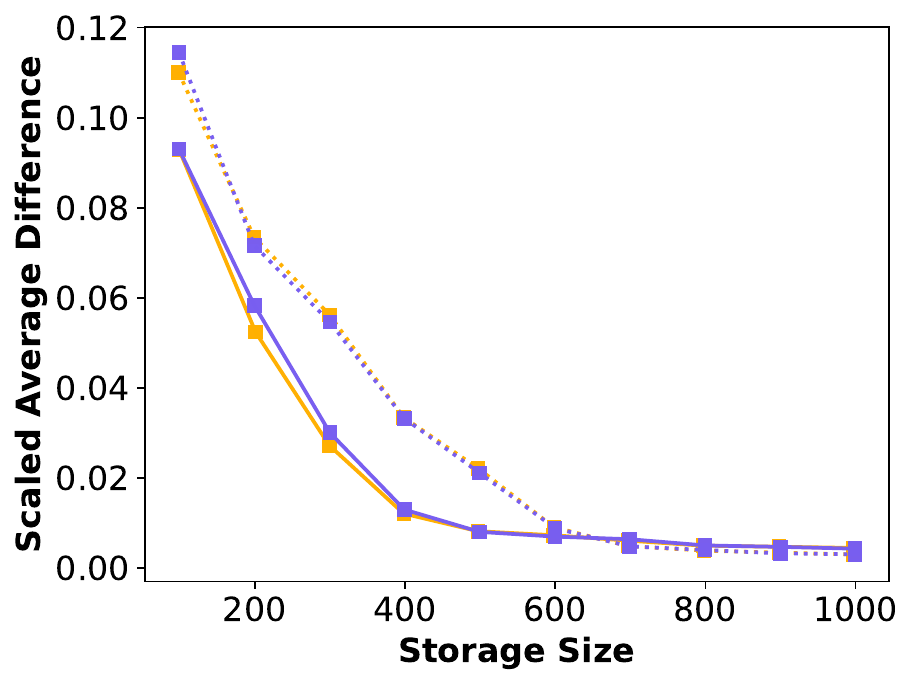}
		\vspace{-1.8em}
		\caption{100\% overlap}
	\end{subfigure}
	\vspace{-.75em}
	\caption{Comparison of End-Biased Sampling (TS-1norm) and its Priority Sampling counterpart (PS-1norm) against our  \TSweighted and \PSweighted methods.  
	}
	\vspace{-.75em}
	\label{fig:InnerProductEst_Synthetic_Comparel1l2}
\end{figure}

% \vspace{-0.25cm}
\revised{As mentioned, we also considered adding End-Biased Sampling \cite{EstanNaughton:2006} as a baseline. This method is equivalent to \thresholdsampling, but samples vector entries with probability proportional to their magnitude, normalized by the vector $\ell_1$ norm. We refer to this as $\ell_1$ sampling to highlight the difference between our methods, which sample based on \emph{squared magnitude} normalized by the $\ell_2$ norm. A variant of \prioritysampling can also be implemented using  $\ell_1$ sampling. We found that End-Biased Sampling performed similarly, but never significantly better than, \thresholdsampling. This is shown in \Cref{fig:InnerProductEst_Synthetic_Comparel1l2}, which uses the same experimental setting as \Cref{fig:InnerProductEst_Synthetic}.}

\begin{figure}[t]%[t]
	\centering
	\begin{subfigure}{.99\columnwidth} 
		\centering
		\includegraphics[width=1.0\linewidth]{figs/legend_ip.png}
	\end{subfigure}
	
	\begin{subfigure}{.49\columnwidth} 
		\centering
		\includegraphics[width=1.0\linewidth]{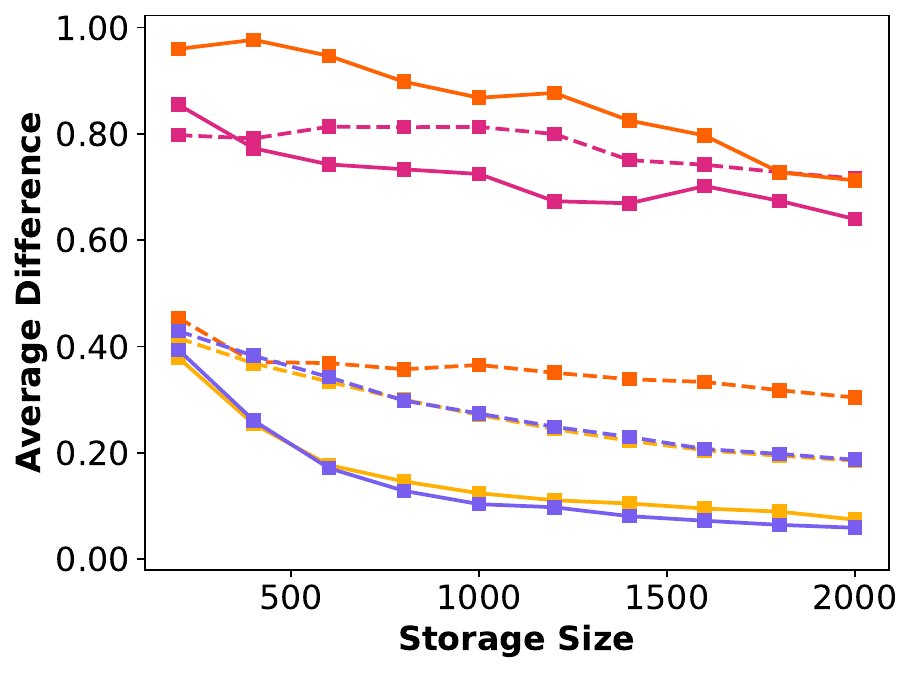}
		\vspace{-1.8em}
		\caption{Correlation -0.2}
	\end{subfigure}
	\begin{subfigure}{.49\columnwidth} 
		\centering
		\includegraphics[width=1.0\linewidth]{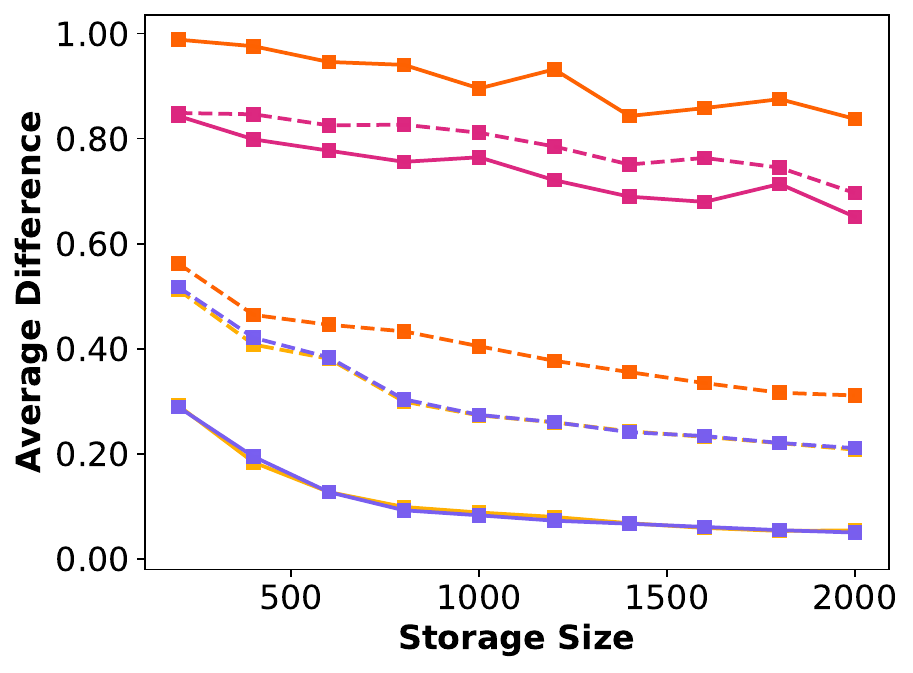}
		\vspace{-1.8em}
		\caption{Correlation 0.4}
	\end{subfigure}%
	
	\begin{subfigure}{.49\columnwidth} 
		\centering
		\includegraphics[width=1.0\linewidth]{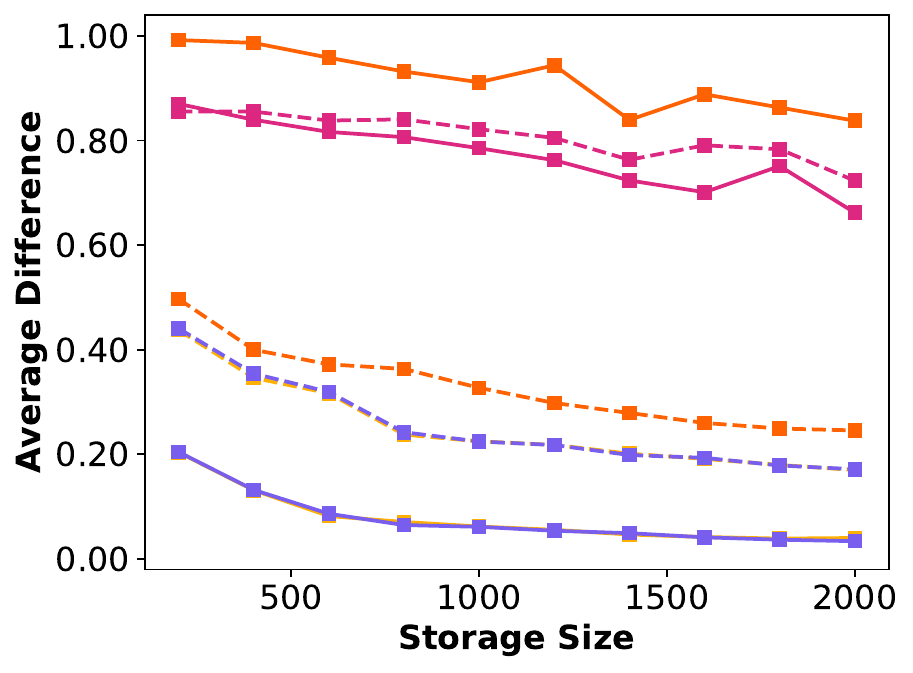}
				\vspace{-1.8em}
		\caption{Correlation -0.6}
	\end{subfigure}
	\begin{subfigure}{.49\columnwidth} 
		\centering
		\includegraphics[width=1.0\linewidth]{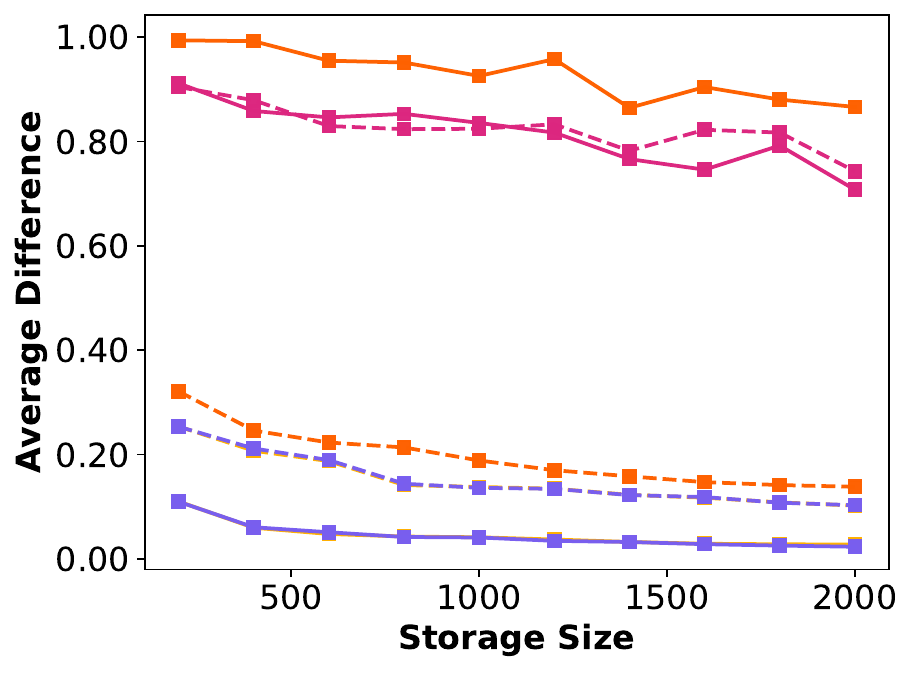}
				\vspace{-1.8em}
		\caption{Correlation 0.8}
	\end{subfigure}
\vspace{-.75em}
	\caption{\revised{
 Join-Correlation Estimation for synthetic data. The lines for \PSweighted and \TSweighted overlap, as do the lines for our \PSuniform and \TSuniform methods, which outperform all other baselines.
 }
 }
	\label{fig:JoinCorrEst_Synthetic}
		\vspace{-.6cm}
\end{figure}

\vspace{-.25cm}
% \subsection{Assessing Runtime Performance}
\subsection{Runtime Performance}
As discussed in~\Cref{sec:intro}, it is also important to consider the time required to compute inner product sketches.
% Threshold and \prioritysampling are designed to compute a sketch of size $m$ in time $O(N)$ and ${O}(N\log m)$, respectively, for a vector with $N$ non-zero entries, essentially matching the asymptotic complexity of the fastest linear sketching methods like CountSketch, and significantly improving on the $O(Nm)$ complexity of the \wmh method from \cite{BessaDFMMSZ:2023}. 
Threshold and \prioritysampling compute a sketch of size $m$ in time $O(N)$ and ${O}(N\log m)$, respectively, for a vector with $N$ non-zero entries, matching 
the complexity of the fastest methods like CountSketch, and improving on the $O(Nm)$ complexity of \wmh \cite{BessaDFMMSZ:2023}.
To see how this theoretical improvement translates to practice, we assess the \runtime efficiency of these methods using high-dimensional synthetic vectors with 250,000 entries, 50,000 of which are non-zero.
As above, non-zero entries are random values in $[-1,1]$, except 10\%  are chosen as outliers. However, for all methods considered, the precise values of entries should have little to no impact on \runtime.

% hide    

\begin{figure}
\centering
    \begin{subfigure}{0.99\columnwidth} 
		\centering
		\includegraphics[width=0.9\linewidth]{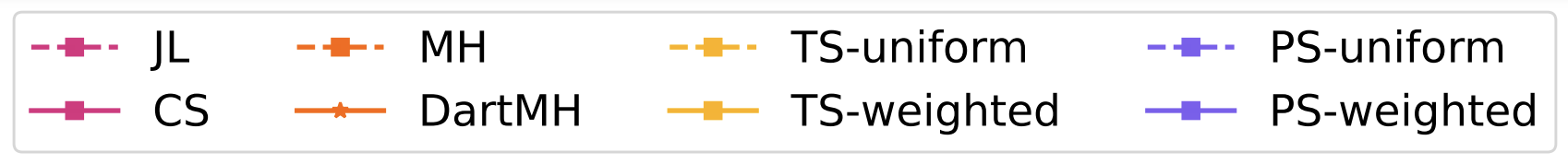}
    \end{subfigure}
 
    \begin{subfigure}{.95\columnwidth}
        \centering
        \includegraphics[width=.66\linewidth]{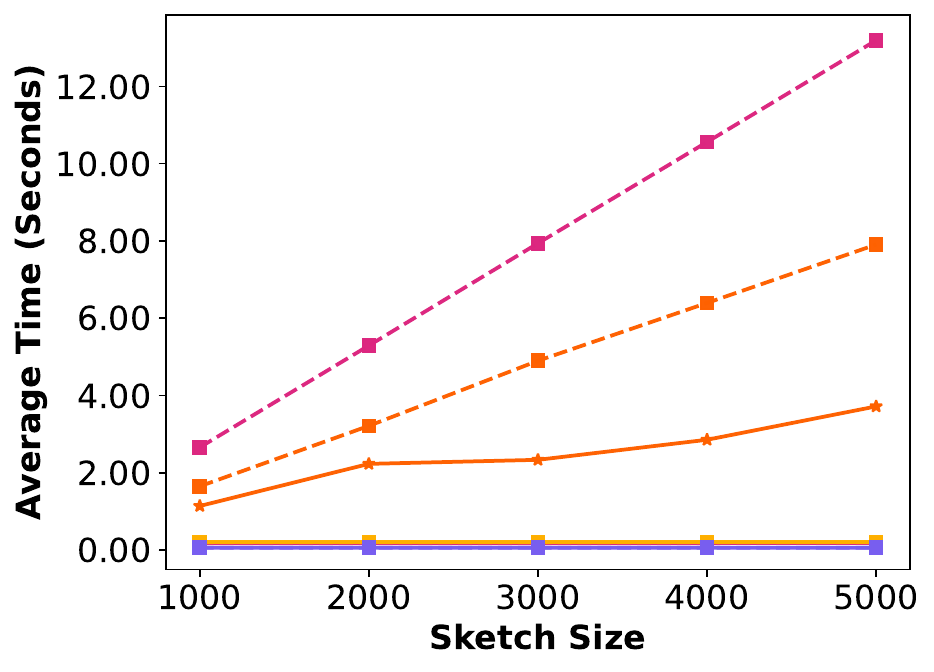}
        \vspace{-.3em}
    \end{subfigure}
    \vspace{-.75em}
    \caption{Sketch construction time. We omit \wmhsexperiments since its slow time would make it difficult to see the other lines. We see a clear linear dependence on the sketch size for JL and MH, and a milder dependence for DartMH. The \runtime of CountSketch,  \thresholdsampling, and \prioritysampling does not noticeably scale with the sketch size.}
    \label{fig:sketch_construction_time}
	\vspace{-.5cm}
\end{figure}

\hide{
\begin{table}
    \small
    \centering
    \begin{tabular}{c r r r r r}
    \toprule
    \textbf{Sketch Size} & \textbf{1000} & \textbf{2000} & \textbf{3000} & \textbf{4000} & \textbf{5000} \\
    \midrule
    \textbf{JL} & 2.650s & 5.295s & 7.939s & 10.55s & 13.19s \\
    \textbf{CS} & 0.180s & 0.184s & 0.181s & 0.180s & 0.191s  \\
    \textbf{MH} & 1.646s & 3.213s & 4.902s & 6.394s & 7.906s \\
    \textbf{\wmhs} & 42.92s & 85.28s & 127.8s & 170.2s & 212.9s\\
    \textbf{DartMH} & 1.303s & 2.232s & 2.308s & 2.858s & 3.662s \\
    \textbf{\TSuniform} & 0.201s & 0.202s & 0.202s & 0.201s & 0.202s \\
    \textbf{\TSweighted} & 0.209s & 0.211s & 0.211s & 0.209s & 0.208s \\
    \textbf{\PSuniform} & 0.062s & 0.061s & 0.062s & 0.061s & 0.061s \\
    \textbf{\PSweighted} & \textbf{\textit{0.062s}} & \textbf{\textit{0.061s}} & \textbf{\textit{0.062s}} & \textbf{\textit{0.061s}} & \textbf{\textit{0.061s}}\\
    \bottomrule
    \end{tabular}
    \caption{Sketch time (seconds) of different methods for different choices of sketch size. Although the basic \thresholdsampling method exhibits a slightly faster asymptotic runtime compared to the \prioritysampling method, our practical implementation of \thresholdsampling is actually slower. This discrepancy arises because, in order to achieve robust accuracy guarantees, we utilize \Cref{alg:threshold_sampling_adapt_computation} to dynamically select an optimal threshold value $\tau$.
    } 
    \label{table:sktech_time}
    \vspace{-.7cm}
\end{table}
}

In addition to our standard baselines, to evaluate runtime, we considered more efficient implementations of the \wmh algorithm from \cite{BessaDFMMSZ:2023}. That paper uses a sampling method studied in \cite{ManasseMcSherryTalwar:2010} and \cite{Ioffe:2010} that 1) requires $O(Nm)$ hash evaluations, and 2) requires an expensive discretization step. Several papers attempt to eliminate these limitations \cite{Shrivastava:2016, Ertl:2018}. 
We implement a recent, improved method called DartMinHash (DartMH) from \cite{Christiani:2020}, which runs in $O(N+m \log m)$. Details on the method are discussed in \subm{the extended version \cite{DaliriFreireMusco:2023b}}\arxiv{\Cref{subsec:fast_wmhs}}.

The times required by different methods to create sketches of varying sizes are shown in \cref{fig:sketch_construction_time}.
As expected, both our weighted and unweighted Threshold and \prioritysampling methods are significantly faster than the $O(Nm)$ time methods like \wmh, unweighted MinHash~(MH) and Johnson-Lindenstrauss (JL). With an average runtime of $.06$ seconds across all sketch sizes, \prioritysampling is competitive with  the less accurate CountSketch, whose average runtime is $.05$ seconds. \thresholdsampling was slightly slower, with an average time of $.21$ seconds. While this method has better asymptotic complexity than \prioritysampling (since there is no need to sort ranks), its slower empirical performance is due to the algorithm used to adaptively adjust the expected sketch size to exactly equal $m$ (discussed in \Cref{sec:thresholdsampling}). \revised{However, we emphasize that our results are primarily meant to illustrate coarse differences in runtime. Evaluating small differences between CountSketch, \prioritysampling, and \thresholdsampling would require more careful implementation in a low-level language, an effort we leave to future work. In any case, all algorithms offer extremely good performance, with no dependence on the size of the sketch.} 

The \wmh method from \cite{BessaDFMMSZ:2023} is not competitive with any of the other methods, requiring $43$ seconds to produce a sketch of size $1000$, and  $213$ seconds to produce a sketch of size $5000$. As such, it was omitted from \cref{fig:sketch_construction_time}. DartMH succeeds in speeding up the method, but even this optimized algorithm is between 20x and 60x more expensive than our \prioritysampling method. 

Finally, for completeness, we evaluated the estimation time for all sketches. As expected there are no significant differences, since the estimation procedure for both sampling and linear sketches amounts to a simple sum over the entries in the sketch. For sketches of size $5000$, estimation times ranged between $0.014$ms and $0.052$ms. 

\hide{
\begin{table}
    \small
    \centering
    \begin{tabular}{ c r r r r r }
    \toprule
    \textbf{Sketch Size} & \textbf{1000} & \textbf{2000} & \textbf{3000} & \textbf{4000} & \textbf{5000} \\
    \midrule
    \textbf{JL} & 0.014ms & 0.016ms & 0.017ms & 0.018ms & 0.021ms   \\
    \textbf{CS} & 0.046ms & 0.047ms & 0.049ms & 0.049ms & 0.052ms  \\
    \textbf{MH} & 0.013ms & 0.016ms & 0.018ms & 0.020ms & 0.023ms  \\
    \textbf{\wmhs} & 0.019ms & 0.024ms & 0.030ms & 0.036ms & 0.041ms  \\
    \textbf{DartMH} & 0.017ms & 0.019ms & 0.023ms & 0.026ms & 0.030ms  \\
    \textbf{\TSuniform} & 0.013ms & 0.018ms & 0.024ms & 0.029ms & 0.034ms \\
    \textbf{\TSweighted} & 0.017ms & 0.026ms & 0.035ms & 0.042ms & 0.050ms \\
    \textbf{\PSuniform} & 0.015ms & 0.021ms & 0.029ms & 0.035ms & 0.042ms \\
    \textbf{\PSweighted} & 0.017ms & 0.025ms & 0.035ms & 0.043ms & 0.051ms \\
    \bottomrule
    \end{tabular}
    \caption{Estimation time (millisecond) of different methods. 
    }
    \label{table:estimation_time}
    \vspace{-.6cm}
\end{table}
}

\vspace{-.25cm}
% \subsection{Assessing Inner Product and Correlation Estimation for Real-World Data}
% \subsection{Assessing Inner Product and Correlation Estimation}
% \subsection{Assessing Estimation for Real-World Data}
% \subsection{Estimation Accuracy in Real-World Scenarios}
\subsection{Estimation Accuracy for Real-World Data }
% and Applications
In addition to synthetic data, we carry out experiments on real-world datasets for practical applications. We use World Bank Group Finances data \cite{WBF_2022} to assess sketching methods for inner product and join-correlation estimation. We also evaluate the performance of \thresholdandpriority for text similarity estimation on the 20 Newsgroups dataset~\cite{20newsgroups}, and for join-size estimation on the \revised{World Bank, Twitter~\cite{KwakLeeParkMoon2010}, and TPC-H datasets~\cite{TPCHSkew}}.

\vspace{-.15cm}
\subsubsection{World Bank Finances Data}
This collection consists of 56 tables~\cite{WBF_2022}, from which we randomly sampled 3,000 column pairs using the following approach (adapted from  ~\cite{SantosBessaChirigati:2021}).
A column pair is represented as ($\langle K_A, V_A \rangle$, $\langle K_B, V_B \rangle$), where $K_A$ and $K_B$ are join keys with temporal data, and $V_A$ and $V_B$ are columns with numerical values from tables $A$ and $B$.
Since there can be repeated keys in $K_A$ and $K_B$, we pre-aggregate the values in $V_A$ and $V_B$ associated with repeated keys into a single value by summing them. This ensures that each key is associated with a single vector index

\begin{table*}
    \centering
    \begin{small}
    \begin{tabular}{ c c c | c c c | c c c}
    \toprule
     & \textbf{Inner Product} & & & \textbf{Join-Correlation} & & & \textbf{Join Size} & \\
    \toprule
    \textbf{Method} & \textbf{Average Error} & $R^2$ \textbf{Score} &\textbf{Method} & \textbf{Average Error} & $R^2$ \textbf{Score}  &\textbf{Method} & \textbf{Average Error} & $R^2$ \textbf{Score} \\
    \midrule
    \textbf{\underline{\TSweighted}} & 0.009 & 0.998 & \textbf{\underline{\PSweighted}} & 0.066 & 0.863 & \textbf{\underline{\TSweighted}} & 0.018 & 0.919 \\
    \textbf{\underline{\PSweighted}} & 0.010 & 0.998 & \textbf{\underline{\TSweighted}} & 0.080 & 0.772 & \textbf{\underline{\PSweighted}} & 0.023 & 0.839 \\
    \textbf{CS} & 0.027 & 0.992 & \textbf{\PSuniform} & 0.104 & 0.697 & \textbf{TS-uniform} & 0.025 & 0.842 \\
    \textbf{\wmh} & 0.032 & 0.985 & \textbf{\TSuniform} & 0.107 & 0.704 & \textbf{PS-uniform} & 0.027 & 0.826 \\
    \textbf{JL} & 0.037 & 0.986 & \textbf{MH} & 0.124 & 0.534 & \textbf{MH} & 0.033 & -0.033 \\
    \textbf{TS-uniform} & 0.096 & 0.217 & \textbf{JL} & 0.203 & 0.347 & \textbf{\wmh} & 0.033 & 0.784 \\
    \textbf{PS-uniform} & 0.115 & 0.233 & \textbf{CS} & 0.210 & 0.417 & \textbf{CS} & 0.044 & 0.729 \\
    \textbf{MH} & 0.119 & -0.065 & \textbf{\wmh} & 0.298 & -0.212 & \textbf{JL} & 0.047 & 0.688 \\
    \bottomrule
    \end{tabular}
    \end{small}
    \caption{\revised{
    Inner product, correlation, and join size estimations for the World Bank data, ranked by average error. Our new \TSweighted and \PSweighted methods (underlined) have both the least average error and the best $R^2$ score for all three problems, although differences are more pronounced for inner product and correlation estimation.
    }
    }
    \label{table:InnerProduct_Corr_JoinSize_est_world_bank}
    \vspace{-.7cm}
\end{table*}

\myparagraph{Inner Product Estimation} 
We first evaluate  \thresholdandpriority on the basic task of computing inner products between the data columns. We normalize all columns to have unit Euclidean norm, which ensures the inner products have a consistent scale (and are upper bounded by $1$). Then we construct sketches of size 400 for all methods, which are used  to estimate inner products.
\revised{\Cref{table:InnerProduct_Corr_JoinSize_est_world_bank} shows the inner product estimation results ranked by the average error over all pairs of columns (a single trial each).} We also include the $R^2$ score, which measures the goodness of fit of the estimated inner products to the actual inner products. 
The best methods are our \TSweighted and \PSweighted, followed by \wmh and JL, which have average error roughly 3x larger.
These results underscore the effectiveness of the weighted sampling methods.
%
%\TSweighted and \PSweighted have the lowest average errors and the highest $R^2$ scores. \wmh and JL also perform well but they are not as accurate as \TSweighted and \PSweighted.

\begin{figure}[t]
	\centering
	\begin{subfigure}{.49\columnwidth} 
		\centering
		\includegraphics[width=1.0\linewidth]{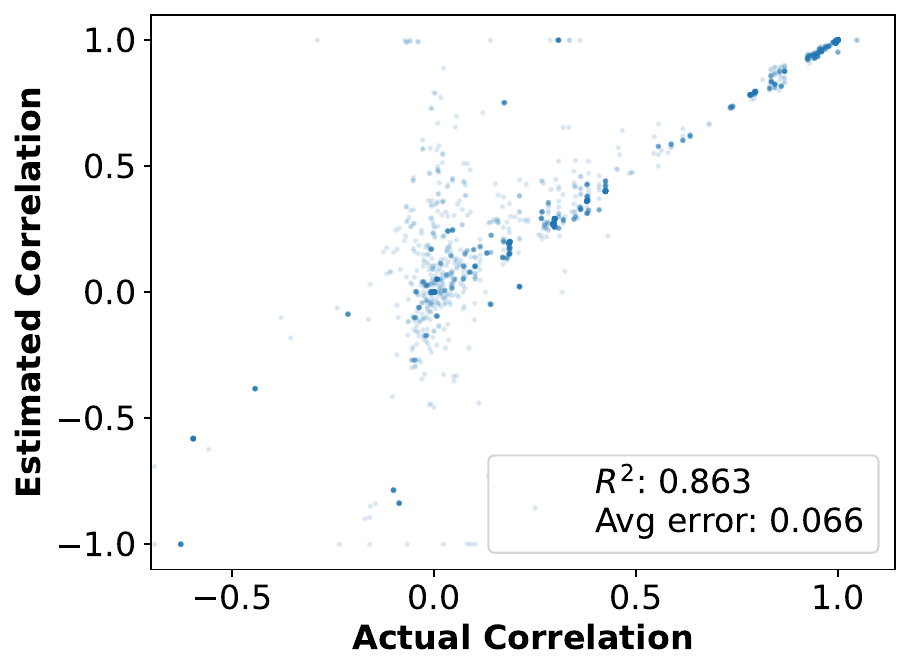}
		\vspace{-1em}
		\caption{\PSweighted}
	\end{subfigure}
	\begin{subfigure}{.49\columnwidth} 
		\centering
		\includegraphics[width=1.0\linewidth]{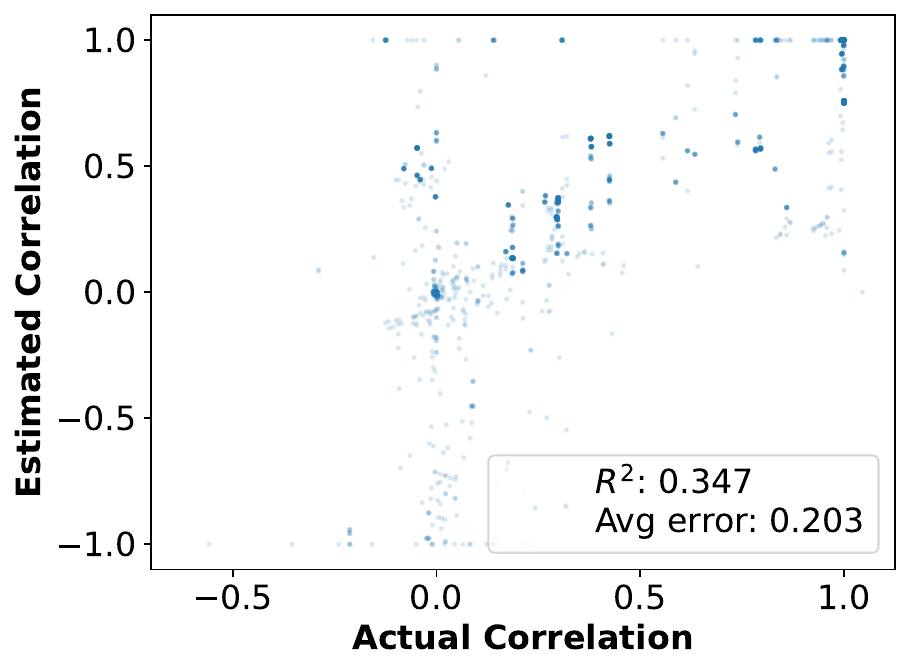}
		\vspace{-1em}
		\caption{Johnson-Lindenstrauss (JL)}
	\end{subfigure}
\vspace{-.5em}
	\caption{\revised{Join-correlation Estimation on World Bank data. The best sampling-based sketch (our \PSweighted method) captures correlations significantly more accurately than the best linear sketching method we tested (JL).}}
	\label{fig:JoinCorrEst_world_bank}
	\vspace{-1.5em}
\end{figure}

\myparagraph{Join-Correlation Estimation} 
We also evaluate the accuracy of \thresholdandpriority for join-correlation estimation using  
the  estimators described in \Cref{sec:applications}.
We consider the same pairs of columns used when estimating inner products above, and again use sketches of size $400$.
\revised{\Cref{table:InnerProduct_Corr_JoinSize_est_world_bank} shows the average error and $R^2$ score for all methods. \PSweighted and \TSweighted have the lowest average errors and the highest $R^2$ scores. They outperform the current state-of-the-art sketching method for join-correlation estimation, which is the 
 KMV-based sketch from \cite{SantosBessaChirigati:2021}. In the table we refer to this method as \PSuniform since it is identical to \prioritysampling with uniform weights.}
\Cref{fig:JoinCorrEst_world_bank} shows scatter plots of correlation estimates for our \PSweighted method (the best sampling-based method) and 
\revised{JL (the best linear sketching method).}
We note that there are a large number of points around zero; this is expected, since many of the datasets are not correlated. 

\revised{
\myparagraph{Join Size Estimation}
Finally, we evaluate our methods on the task of join size estimation using the same World Bank data, but without aggregating keys. We use the standard reduction between join size estimation and inner product computation with vectors containing key frequencies \cite{cormode2011sketch}.
% This task arises in data discovery problems, such as data augmentation for enhancing machine learning models. Prior to executing a join operation, join size estimation assists users in making informed decisions regarding which tables to join.
% This proactive step helps in minimizing computational resource usage and can result in improved model performance.
Results are presented in \Cref{table:InnerProduct_Corr_JoinSize_est_world_bank}.
% 
% We perform a join size experiment for World Bank data with same pairs of columns and sketch size $400$.
% 
% 
Since key frequencies vary, our weighted sampling methods, \TSweighted and \PSweighted, produce more accurate results. Linear sketching methods like CountSketch and JL perform worst.
}

\begin{figure}[t]
	\centering
	\begin{subfigure}{\columnwidth} 
		\centering
		\includegraphics[width=.99\linewidth]{figs/legend_ip.png} %legend is narrow if 0.75, 0.99 looks normal
	\end{subfigure}
	\begin{subfigure}{0.49\columnwidth} 
		\centering
		\includegraphics[width=1.0\linewidth]{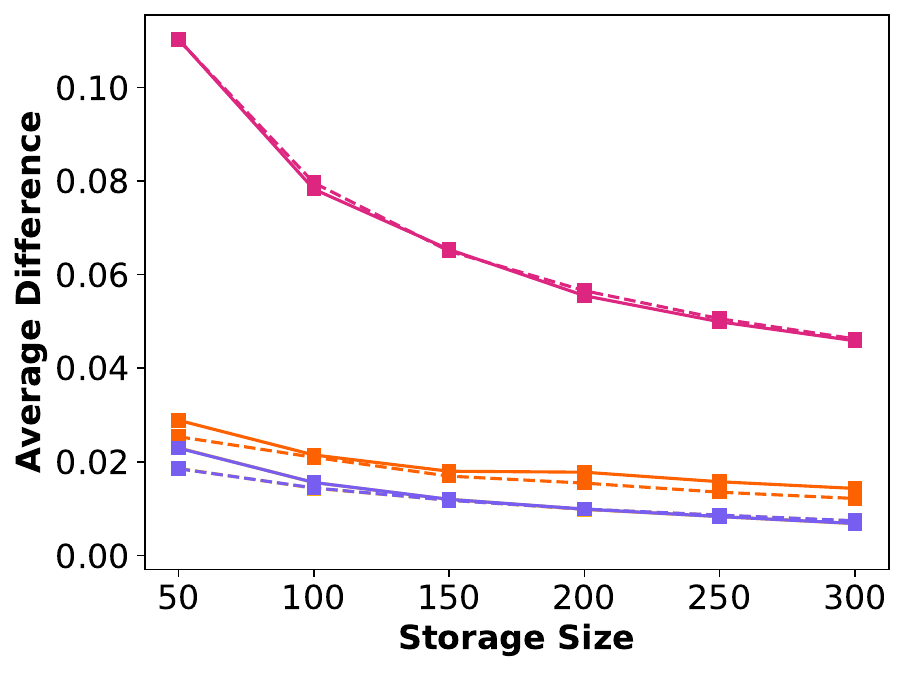}
		\vspace{-1em}
		\caption{All Documents}
        \label{fig:DocSim_TFIDF_all}
	\end{subfigure}
        \begin{subfigure}{0.49\columnwidth} 
		\centering
		\includegraphics[width=1.0\linewidth]{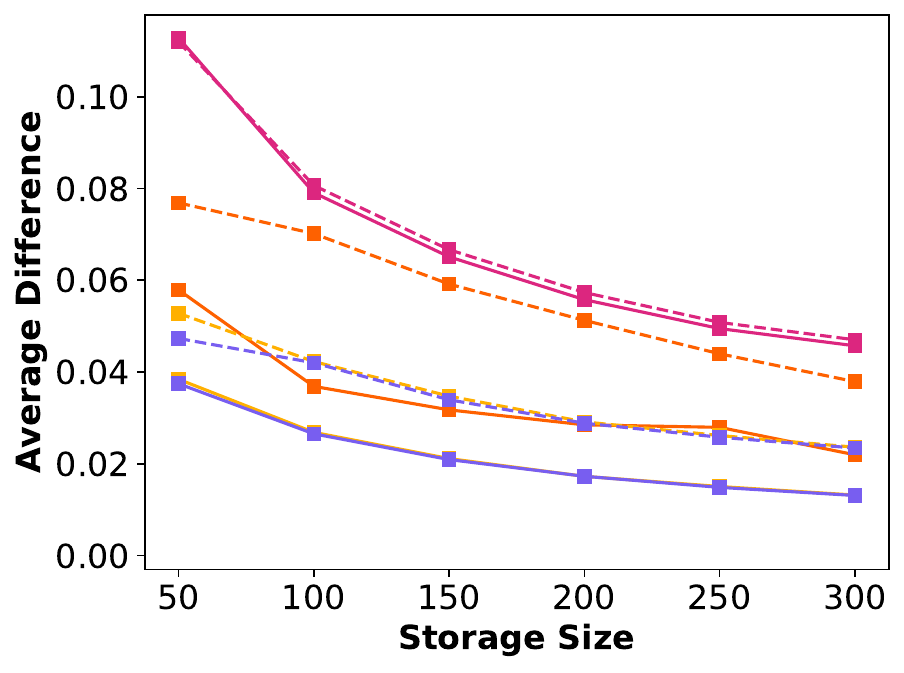}
		\vspace{-1em}
		\caption{Documents > 500 words}
        \label{fig:DocSim_TFIDF_500}
	\end{subfigure}
	\vspace{-.5em}
	\caption{\revised{
Average estimation error for text similarity estimation using the 20 Newsgroups dataset. Note that the lines for \PSweighted, \PSuniform, \TSweighted, and \TSuniform overlap in (a), as do the
the lines for \PSweighted and \TSweighted in (b).  \PSweighted and \TSweighted  outperform all baselines for documents more than 500 words.
}}
	\label{fig:DocSim_TFIDF}
		\vspace{-1.5em}
\end{figure}

\vspace{-.15cm}
\subsubsection{20 Newsgroups Dataset}
\revised{
We also assess the effectiveness of \thresholdandpriority for estimating \emph{document similarity} using the 20 Newsgroups Dataset~\cite{20newsgroups}. We generate a feature vector for each document that includes both unigrams (single words) and bigrams (pairs of words). We use standard TF-IDF weights to scale the entries of the vector \cite{SaltonWongYang:1975} and then measure similarity using the cosine similarity metric, which is equivalent to the inner product when the vectors are normalized to have unit norm.
}

We sample 200,000 document pairs from the dataset and plot average error. As \Cref{fig:DocSim_TFIDF_all} shows, the  linear sketching methods (JL and CountSketch) \revised{perform worst. \thresholdandpriority\xspace}obtain the best accuracy for all sketch sizes, although the difference between the unweighted and weighted methods is negligible. \revised{As shown in \Cref{fig:DocSim_TFIDF_500}, this difference becomes more pronounced when only considering documents with more than 500 words. For longer documents, our \TSweighted and \PSweighted perform notably better than their uniform-sampling counterparts.}
The larger performance gap could be due to more variability in  TF-IDF weights in longer documents (which benefits the weighted sampling methods) or simply to the fact that estimating cosine similarity is more challenging for longer documents, so differences in the methods become more pronounced as estimation error increases. 

\vspace{-.15cm}
\subsubsection{TPC-H Benchmark and Twitter Data}
\begin{figure}[t]
	\centering
	\begin{subfigure}{0.99\columnwidth} 
		\centering
		\includegraphics[width=.99\linewidth]{figs/legend_ip.png}
	\end{subfigure}
  \begin{subfigure}{.49\columnwidth} 
		\centering
		\includegraphics[width=1\linewidth]{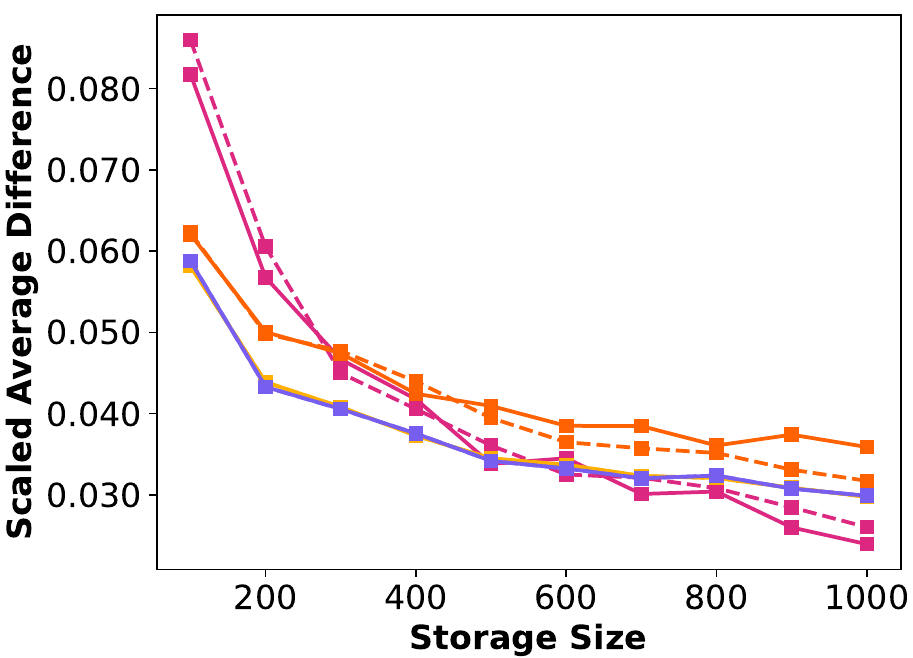}
		\vspace{-.2cm}
  		\caption{TPC-H Benchmark}
    	\label{fig:JoinSize_TPCH}
	\end{subfigure}
    \begin{subfigure}{.49\columnwidth} 
		\centering
		\includegraphics[width=1\linewidth]{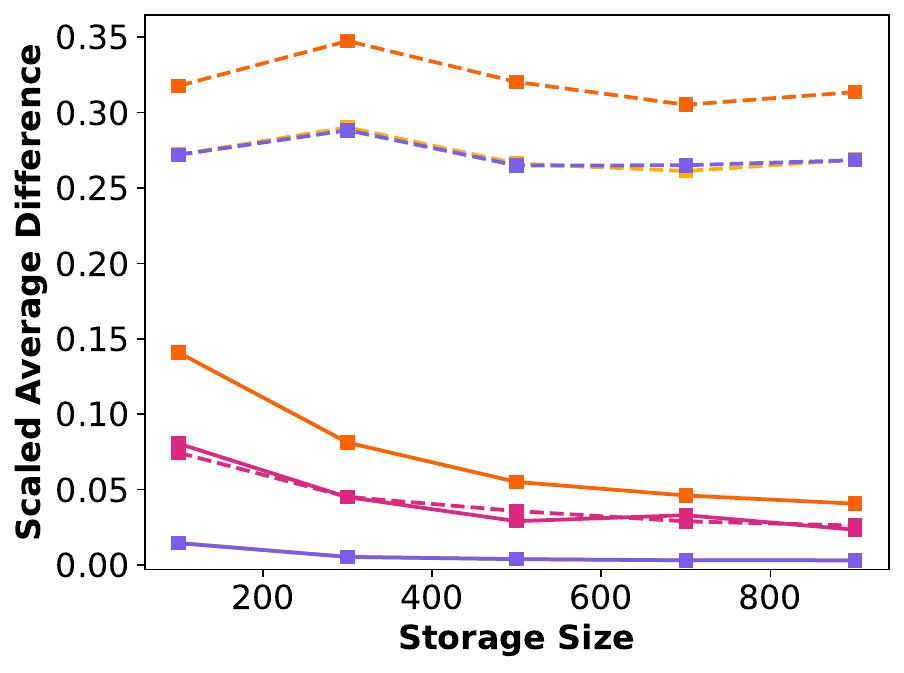}
		\caption{Twitter Self-Join}
  	\label{fig:JoinSize_Twitter}
	\end{subfigure}
 \vspace{-.3cm}
	\caption{\revised{Join size estimation for the
    % Twitter dataset and the TPC-H benchmark. 
    Twitter and TPC-H datasets. 
	The lines for \PSweighted and \TSweighted overlap, as do the lines for
    % our \PSuniform and \TSuniform methods.
    \PSuniform and \TSuniform. Our \PSweighted and \TSweighted methods are most reliable, performing well in both experiments, the second of which involves two tables with highly non-uniform key distributions.
    }}
	% \vspace{-.7cm}
\end{figure}
\revised{
Finally, we evaluate \thresholdandpriority on two join size estimation tasks. The first is the standard TPC-H benchmark~\cite{TPCHSkew}. TPC-H data was generated with a scale factor of 1 and skew parameter $z=2$. The join was performed between \textit{LINEITEM} and \textit{PARTSUPP} tables on the key \textit{SUPPKEY}.
Average relative error for $200$ trials are presented in \Cref{fig:JoinSize_TPCH}. For moderate sketch sizes (up to $\sim 600$) our sampling based methods outperform  linear sketching, and they always outperform MH and WMH. However, there is little difference between the weighted and unweighted sampling versions of our methods. We believe this is due to the fact that, even with skew,  the TPC-H benchmark only has a non-uniform key distribution in the LINEITEM table. The key distribution of the larger \textit{PARTSUPP} table remains uniform. 
\revised{Difference between the methods is much more pronounced in our second experiment on estimating join sizes using the Twitter dataset from ~\cite{KwakLeeParkMoon2010}. This data consists of a list of tuples (user, follower), representing the follower-followee relationship.} We sampled 14,000,000 (user, follower) tuples from the dataset, which include approximately 420,000 users.
Following the example in \cite{ChenYi2017},  we perform a self-join of the table to identify all the 2-hop "follows" relationships.
Results are shown in \Cref{fig:JoinSize_Twitter}. For all sketch sizes, \TSweighted and \PSweighted have the smallest errors, followed by the linear sketching methods, and then by \wmh. The unweighted sampling methods 
(MH, \TSuniform, \PSuniform) perform poorly, since in this dataset there is a lot of variability in key frequencies; \revised{weighted} sampling is needed to accurately estimate join size.

}

% --------------------
\vspace{-.2cm}
\section{Additional Related Work}
\label{sec:related_work}
As discussed in \cref{sec:intro}, we are only aware of two previous papers that directly address the inner product estimation problem using sampling-based sketches: the \wmh work of \cite{BessaDFMMSZ:2023} and the End-Biased Sampling work of \cite{EstanNaughton:2006}. Some follow-up work on End-Biased Sampling, such as Correlated Sampling~\cite{vengerov2015join} and Two-level Sampling \cite{ChenYi2017}, can also be used to estimate inner products. However, the goal of these works is to handle the more general problem of approximating data operations (such as SUM, COUNT, MEAN) with SQL predicates (WHERE clauses). In our setting, the methods from \cite{vengerov2015join} and \cite{ChenYi2017} degenerate to uniform sampling methods (i.e., KMV or \revised{Threshold Sampling}
%Priority Sampling
with uniform weights), as they do not take into account the vector entries (i.e., $\bv{a}_i$ and $\bv{b}_i$) when selecting samples.

We also note that inner product estimation can be seen as a special case of the predicate aggregation problem studied in \cite{CohenKaplanSen:2009}. While that work gives unbiased estimators based on Threshold and Priority Sampling, inner product estimation is not considered specifically, so there is no guidance on how probabilities should be chosen or variance analyzed. Follow-up work in \cite{Cohen:2015} can be used to analyze variance given a choice of probabilities. However, in our setting, the work leads to loose bounds that depend on $\max_i {|\bv{a}_i\bv{b}_i|}/{\min(\bv{a}_i^2, \bv{b}_i^2)}$. This value can be arbitrarily large in comparison to $\|\bv{a}\|_2\|\bv{b}\|_2$, so unlike our analysis, this prior work cannot be used to beat the linear sketching guarantee of \eqref{eq:jl_guar} for inner product estimation.

Beyond the problem of inner product estimation, our work is more broadly related to the large body of work on coordinated random sampling methods, which use shared randomness (e.g., a shared hash function or random permutation) to collect samples of two vectors $\bv{a}$ and $\bv{b}$. \revised{Threshold and Priority Sampling are both examples of coordinated sampling, as is MinHash and the $k$-minimum values (KMV) sketch.} However, there are other methods, including the coordinated random sampling method \cite{CohenKaplan:2007}, conditional random sampling \cite{LiChurchHastie:2006}, and coordinated variants of PPSWOR sampling \cite{Cohen:2023}.  
% Our work is also related to studies in survey sampling in statistics, which have been using coordinated sampling for several decades~(see e.g., ~\cite{wu2020general} for a recent book on the topic).  Under this lens, \thresholdsampling is an algorithm that performs sampling with probabilities proportional to size (PPS). Specifically, it implements sampling without replacement (PPSWOR) using Poisson sampling with inclusion probabilities proportional to the normalized squared vector entries (i.e.,  $\bv{a}_i^2/\|\bv{a}\|_2^2$). The idea of using weighted sampling for data sketching has also been described in several other papers~\cite{CohenKaplan:2007, CohenKaplanSen:2009, CohenKaplan:2013}, but these did not consider the specific problem of inner product estimation.

% ------------------
\vspace{-.25cm}
\section{Conclusion}
\label{sec:conclusion}

We propose two simple and efficient sampling-based sketches for
inner product estimation. We prove  theoretical accuracy guarantees for both methods that are stronger than the guarantees of popular linear sketching methods, and that match the best-known guarantees of the state-of-the-art hashing-based \wmh sketch~\cite{BessaDFMMSZ:2023}.
At the same time, our methods run in near-linear time, so are much faster than \wmh. They also perform better in our empirical evaluation. In particular, our fixed-size Priority Sampling method provides a new state-of-the-art for inner product estimation and related applications, including join-correlation estimation.

%\begin{acks}
\vspace{-.1cm}
\myparagraph{Acknowledments}
We thank Otmar Ertl, Jonathan Weare, and Xiaoou Cheng for helpful conversations.
This research was supported by NSF award CCF-2046235 and the DARPA D3M program. Opinions, findings, conclusions, or recommendations expressed in this material
are those of the authors and do not %necessarily 
reflect the views of NSF or DARPA. %other funding organizations.
%\end{acks}

% -------------------------
% Appendix
% -------------------------
\balance
\bibliographystyle{apalike}
\bibliography{paper}

\begin{thebibliography}{}

\bibitem[Achlioptas, 2003]{Achlioptas:2003}
Achlioptas, D. (2003).
\newblock Database-friendly random projections: {J}ohnson-{L}indenstrauss with
  binary coins.
\newblock {\em J. Comput. Syst. Sci.}, 66(4).
\newblock \pPODS{2001}.

\bibitem[Alon et~al., 2005]{AlonDuffieldLund:2005}
Alon, N., Duffield, N., Lund, C., and Thorup, M. (2005).
\newblock Estimating arbitrary subset sums with few probes.
\newblock In {\em \PODS{2005}}.

\bibitem[Alon et~al., 1999a]{AlonGibbonsMatias:1999}
Alon, N., Gibbons, P.~B., Matias, Y., and Szegedy, M. (1999a).
\newblock Tracking join and self-join sizes in limited storage.
\newblock In {\em \PODS{1999}}.

\bibitem[Alon et~al., 1999b]{AlonMatiasSzegedy:1999}
Alon, N., Matias, Y., and Szegedy, M. (1999b).
\newblock The space complexity of approximating the frequency moments.
\newblock {\em Journal of Computer and System Sciences}, 58(1).

\bibitem[Arriaga and Vempala, 2006]{ArriagaVempala:2006}
Arriaga, R.~I. and Vempala, S. (2006).
\newblock An algorithmic theory of learning: Robust concepts and random
  projection.
\newblock {\em Machine Learning}, 63(2).

\bibitem[Bar-Yossef et~al., 2002]{Bar-YossefJayramKumar:2002}
Bar-Yossef, Z., Jayram, T.~S., Kumar, R., Sivakumar, D., and Trevisan, L.
  (2002).
\newblock Counting distinct elements in a data stream.
\newblock In {\em \RANDOM{2002}}.

\bibitem[Bessa et~al., 2023]{BessaDFMMSZ:2023}
Bessa, A., Daliri, M., Freire, J., Musco, C., Musco, C., Santos, A., and Zhang,
  H. (2023).
\newblock Weighted minwise hashing beats linear sketching for inner product
  estimation.
\newblock In {\em \PODS{2023}}.

\bibitem[Beyer et~al., 2007]{BeyerHaasReinwald:2007}
Beyer, K., Haas, P.~J., Reinwald, B., Sismanis, Y., and Gemulla, R. (2007).
\newblock On synopses for distinct-value estimation under multiset operations.
\newblock In {\em \SIGMOD{2007}}.

\bibitem[Broder, 1997]{Broder:1997}
Broder, A. (1997).
\newblock On the resemblance and containment of documents.
\newblock In {\em Proceedings. Compression and Complexity of SEQUENCES 1997}.

\bibitem[{Castro Fernandez} et~al., 2019]{Castro-FernandezMinNava:2019}
{Castro Fernandez}, R., {Min}, J., {Nava}, D., and {Madden}, S. (2019).
\newblock Lazo: A cardinality-based method for coupled estimation of jaccard
  similarity and containment.
\newblock In {\em \ICDE{2019}}.

\bibitem[Charikar, 2002]{Charikar:2002}
Charikar, M. (2002).
\newblock Similarity estimation techniques from rounding algorithms.
\newblock In {\em \STOC{2002}}.

\bibitem[Charikar et~al., 2002]{CharikarChenFarach-Colton:2002}
Charikar, M., Chen, K., and Farach-Colton, M. (2002).
\newblock Finding frequent items in data streams.
\newblock In {\em \ICALP{2002}}.

\bibitem[Chen and Yi, 2017]{ChenYi2017}
Chen, Y. and Yi, K. (2017).
\newblock Two-level sampling for join size estimation.
\newblock In {\em Proceedings of the 2017 ACM International Conference on
  Management of Data}.

\bibitem[Chepurko et~al., 2020]{ChepurkoMarcusZgraggen:2020}
Chepurko, N., Marcus, R., Zgraggen, E., Fernandez, R.~C., Kraska, T., and
  Karger, D. (2020).
\newblock Arda: automatic relational data augmentation for machine learning.
\newblock {\em Proc. VLDB Endow.}, 13(9).

\bibitem[Chi and Zhu, 2017]{ChiZhu:2017}
Chi, L. and Zhu, X. (2017).
\newblock Hashing techniques: A survey and taxonomy.
\newblock {\em ACM Comput. Surv.}, 50(1).

\bibitem[Christiani, 2020]{Christiani:2020}
Christiani, T. (2020).
\newblock Dartminhash: Fast sketching for weighted sets.
\newblock {\em \arXiv{2005.11547}}.

\bibitem[Cohen, 2015]{Cohen:2015}
Cohen, E. (2015).
\newblock Multi-objective weighted sampling.
\newblock In {\em 2015 Third IEEE Workshop on Hot Topics in Web Systems and
  Technologies (HotWeb)}.

\bibitem[Cohen, 2023]{Cohen:2023}
Cohen, E. (2023).
\newblock Sampling big ideas in query optimization.
\newblock In {\em \PODS{2023}}.

\bibitem[Cohen and Kaplan, 2007]{CohenKaplan:2007}
Cohen, E. and Kaplan, H. (2007).
\newblock Summarizing data using bottom-k sketches.
\newblock In {\em \PODC{2007}}.

\bibitem[Cohen and Kaplan, 2013]{CohenKaplan:2013}
Cohen, E. and Kaplan, H. (2013).
\newblock What you can do with coordinated samples.
\newblock In {\em \APPROX{2013}}.

\bibitem[Cohen et~al., 2009]{CohenKaplanSen:2009}
Cohen, E., Kaplan, H., and Sen, S. (2009).
\newblock Coordinated weighted sampling for estimating aggregates over multiple
  weight assignments.
\newblock {\em Proc. VLDB Endow.}, 2(1).

\bibitem[Cohen et~al., 2020]{CohenPaghWoodruff:2020}
Cohen, E., Pagh, R., and Woodruff, D. (2020).
\newblock Wor and p\textquotesingle s: Sketches for \textbackslash
  ell\_p-sampling without replacement.
\newblock In Larochelle, H., Ranzato, M., Hadsell, R., Balcan, M., and Lin, H.,
  editors, {\em Advances in Neural Information Processing Systems}, volume~33,
  pages 21092--21104. Curran Associates, Inc.

\bibitem[Cormode, 2011]{cormode2011sketch}
Cormode, G. (2011).
\newblock Sketch techniques for approximate query processing.
\newblock {\em Foundations and Trends in Databases. NOW publishers}.

\bibitem[Cormode and Garofalakis, 2005]{CormodeGarofalakis:2005}
Cormode, G. and Garofalakis, M. (2005).
\newblock Sketching streams through the net: Distributed approximate query
  tracking.
\newblock {\em Proc. VLDB Endow.}

\bibitem[Cormode and Garofalakis, 2016]{CormodeGarofalakis:2016}
Cormode, G. and Garofalakis, M. (2016).
\newblock {\em Join Sizes, Frequency Moments, and Applications}.
\newblock Springer Berlin Heidelberg.

\bibitem[Cormode et~al., 2011]{CormodeGarofalakisHaas:2011}
Cormode, G., Garofalakis, M., Haas, P., and Jermaine, C. (2011).
\newblock {\em Synopses for Massive Data: Samples, Histograms, Wavelets,
  Sketches}.
\newblock Foundations and Trends in Databases. NOW publishers.

\bibitem[Daliri et~al., 2023]{DaliriFreireMusco:2023}
Daliri, M., Freire, J., Musco, C., Santos, A., and Zhang, H. (2023).
\newblock Simple analysis of priority sampling.
\newblock {\em SIAM Symposium on Simplicity in Algorithms (SOSA24)}.

\bibitem[Dasgupta and Gupta, 2003]{DasguptaGupta:2003}
Dasgupta, S. and Gupta, A. (2003).
\newblock An elementary proof of a theorem of johnson and lindenstrauss.
\newblock {\em Random Structures \& Algorithms}, 22(1).

\bibitem[Duffield et~al., 2004]{DuffieldLundThorup:2004}
Duffield, N., Lund, C., and Thorup, M. (2004).
\newblock Flow sampling under hard resource constraints.
\newblock In {\em Proceedings of the Joint International Conference on
  Measurement and Modeling of Computer Systems (SIGMETRICS 2004)}.

\bibitem[Duffield et~al., 2005]{DuffieldLundThorup:2005}
Duffield, N., Lund, C., and Thorup, M. (2005).
\newblock Learn more, sample less: control of volume and variance in network
  measurement.
\newblock {\em IEEE Transactions on Information Theory}, 51(5).

\bibitem[Ertl, 2018]{Ertl:2018}
Ertl, O. (2018).
\newblock Bagminhash - minwise hashing algorithm for weighted sets.
\newblock In {\em \KDD{2018}}.

\bibitem[Esmailoghli et~al., 2021]{EsmailoghliQuiane-RuizAbedjan:2021}
Esmailoghli, M., Quian{\'e}-Ruiz, J.-A., and Abedjan, Z. (2021).
\newblock Cocoa: Correlation coefficient-aware data augmentation.
\newblock In {\em 24th International Conference on Extending Database
  Technology (EDBT)}.

\bibitem[Estan and Naughton, 2006]{EstanNaughton:2006}
Estan, C. and Naughton, J. (2006).
\newblock End-biased samples for join cardinality estimation.
\newblock In {\em \ICDE{2006}}.

\bibitem[Flajolet, 1990]{Flajolet:1990}
Flajolet, P. (1990).
\newblock On adaptive sampling.
\newblock {\em Computing}, 43(4).

\bibitem[Gollapudi and Panigrahy, 2006]{GollapudiPanigrahy:2006}
Gollapudi, S. and Panigrahy, R. (2006).
\newblock Exploiting asymmetry in hierarchical topic extraction.
\newblock In {\em \CIKM{2006}}.

\bibitem[Howell, 2018]{Howell:2018}
Howell, D. (2018).
\newblock Generating correlated data.
\newblock {\em Outline of the Statistical Pages Folder}.

\bibitem[Ioffe, 2010]{Ioffe:2010}
Ioffe, S. (2010).
\newblock Improved consistent sampling, weighted minhash and l1 sketching.
\newblock In {\em \ICDM{2010}}.

\bibitem[Ionescu et~al., 2022]{IonescuHaiFragkoulis:2022}
Ionescu, A., Hai, R., Fragkoulis, M., and Katsifodimos, A. (2022).
\newblock Join path-based data augmentation for decision trees.
\newblock In {\em 2022 IEEE 38th International Conference on Data Engineering
  Workshops (ICDEW)}. IEEE.

\bibitem[Jacques, 2015]{Jacques:2015}
Jacques, L. (2015).
\newblock A quantized johnson--lindenstrauss lemma: The finding of buffon's
  needle.
\newblock {\em IEEE Transactions on Information Theory}, 61(9).

\bibitem[Jayaram and Woodruff, 2018]{JayaramWoodruff:2018}
Jayaram, R. and Woodruff, D.~P. (2018).
\newblock Perfect $l_p$ sampling in a data stream.
\newblock In {\em \FOCS{2018}}, pages 544--555.

\bibitem[Kwak et~al., 2010]{KwakLeeParkMoon2010}
Kwak, H., Lee, C., Park, H., and Moon, S. (2010).
\newblock {W}hat is {T}witter, a social network or a news media?
\newblock In {\em \WWW{2010}}, New York, NY, USA. ACM.

\bibitem[Larsen et~al., 2021]{LarsenPaghTetek:2021}
Larsen, K.~G., Pagh, R., and T{\v{e}}tek, J. (2021).
\newblock Countsketches, feature hashing and the median of three.
\newblock In {\em \ICML{2021}}. PMLR.

\bibitem[Li et~al., 2006]{LiChurchHastie:2006}
Li, P., Church, K., and Hastie, T. (2006).
\newblock Conditional random sampling: A sketch-based sampling technique for
  sparse data.
\newblock In {\em \NIPS{2006}}, volume~19.

\bibitem[Li and K\"{o}nig, 2011]{LiKonig:2011}
Li, P. and K\"{o}nig, A.~C. (2011).
\newblock Theory and applications of b-bit minwise hashing.
\newblock {\em Commun. ACM}, 54(8).

\bibitem[Liu et~al., 2022]{LiuChaiLuo:2022}
Liu, J., Chai, C., Luo, Y., Lou, Y., Feng, J., and Tang, N. (2022).
\newblock Feature augmentation with reinforcement learning.
\newblock In {\em \ICDE{2022}}.

\bibitem[Manasse et~al., 2010]{ManasseMcSherryTalwar:2010}
Manasse, M., McSherry, F., and Talwar, K. (2010).
\newblock Consistent weighted sampling.
\newblock Technical Report MSR-TR-2010-73, Microsoft Research.

\bibitem[Mitchell, 1997]{20newsgroups}
Mitchell, T. (1997).
\newblock 20 newsgroups dataset.
\newblock
  \url{https://scikit-learn.org/stable/modules/generated/sklearn.datasets.fetch_20newsgroups.html}.

\bibitem[Ohlsson, 1998]{Ohlsson:1998}
Ohlsson, E. (1998).
\newblock Sequential poisson sampling.
\newblock {\em Journal of Official Statistics}, 14(2).

\bibitem[Pagh et~al., 2014]{PaghStockelWoodruff:2014}
Pagh, R., St\"{o}ckel, M., and Woodruff, D.~P. (2014).
\newblock Is min-wise hashing optimal for summarizing set intersection?
\newblock In {\em \PODS{2014}}.

\bibitem[Rusu and Dobra, 2008]{RusuDobra:2008}
Rusu, F. and Dobra, A. (2008).
\newblock Sketches for size of join estimation.
\newblock {\em ACM Transactions on Database Systems (TODS)}, 33(3).

\bibitem[Salton et~al., 1975]{SaltonWongYang:1975}
Salton, G., Wong, A., and Yang, C.-S. (1975).
\newblock A vector space model for automatic indexing.
\newblock {\em Communications of the ACM}, 18(11).

\bibitem[Santos et~al., 2021]{SantosBessaChirigati:2021}
Santos, A., Bessa, A., Chirigati, F., Musco, C., and Freire, J. (2021).
\newblock Correlation sketches for approximate join-correlation queries.
\newblock In {\em \SIGMOD{2021}}.

\bibitem[Santos et~al., 2022]{SantosBessaMusco:2022}
Santos, A., Bessa, A., Musco, C., and Freire, J. (2022).
\newblock A sketch-based index for correlated dataset search.
\newblock In {\em \ICDE{2022}}.

\bibitem[Shrivastava, 2016]{Shrivastava:2016}
Shrivastava, A. (2016).
\newblock Simple and efficient weighted minwise hashing.
\newblock In {\em \NIPS{2016}}.

\bibitem[Szegedy, 2006]{Szegedy:2006}
Szegedy, M. (2006).
\newblock The {DLT} priority sampling is essentially optimal.
\newblock In {\em \STOC{2006}}.

\bibitem[Szegedy and Thorup, 2007]{SzegedyThorup:2007}
Szegedy, M. and Thorup, M. (2007).
\newblock On the variance of subset sum estimation.
\newblock In {\em \ESA{2007}}. Springer Berlin Heidelberg.

\bibitem[Vengerov et~al., 2015]{vengerov2015join}
Vengerov, D., Menck, A.~C., Zait, M., and Chakkappen, S.~P. (2015).
\newblock Join size estimation subject to filter conditions.
\newblock {\em Proc. VLDB Endow.}, 8(12).

\bibitem[{World Bank}, 2023]{WBF_2022}
{World Bank} (2023).
\newblock World bank group finances.
\newblock https://finances.worldbank.org/.

\bibitem[{Yang} et~al., 2019]{YangZhangZhang:2019}
{Yang}, Y., {Zhang}, Y., {Zhang}, W., and {Huang}, Z. (2019).
\newblock Gb-kmv: An augmented kmv sketch for approximate containment
  similarity search.
\newblock In {\em \ICDE{2019}}.

\bibitem[Yu, 2022]{TPCHSkew}
Yu, F. (2022).
\newblock Tpch skew.
\newblock \url{https://github.com/YSU-Data-Lab/TPC-H-Skew}.

\bibitem[Zhu et~al., 2016]{ZhuNargesianPu:2016}
Zhu, E., Nargesian, F., Pu, K.~Q., and Miller, R.~J. (2016).
\newblock {LSH} ensemble: Internet-scale domain search.
\newblock {\em Proc. VLDB Endow.}, 9(12).

\end{thebibliography}

\newpage
\appendix
\section{Implementation Details}

\subsection{Adaptive Threshold Selection and High Probability Bounds for Threshold Sampling}
\label{app:adaptive_and_high_prob}
In this section, we briefly discuss two points related to the sketch size of our Threshold Sampling method (\cref{alg:threshold_sampling}). 
First, as shown in \Cref{thm:main}, when \cref{alg:threshold_sampling} is run with target sketch size $m$ on a vector $\bv{a}$, it returns a sketch $\mathcal{S}(\bv{a})$ that contains $\leq m$ key/value pairs from $\bv{a}$ in expectation. 
Ideally, we would like the expected size of the sketch to be \emph{exactly} proportional to $m$. 

It is not difficult to adjust the sketching method to achieve this property. 
To do so, instead of setting each threshold $\tau_i$ equal to $m\cdot {\bv{a}_i^2}/{\|\bv{a}\|_2^2}$ on Line 3 of \cref{alg:threshold_sampling}, we can simply increase the threshold slightly to $m'\cdot {\bv{a}_i^2}/{\|\bv{a}\|_2^2}$ for some value of $m' \geq m$. Doing so ensures a higher probability of sampling every entry of $\bv{a}$, so will not increase the variance of the estimator given in \cref{alg:inner_product_estimator}.

To ensure an expected sketch size of exactly $m$, we can check from the proof of \cref{alg:threshold_sampling} that it suffices to set $m'$ so that:
\begin{align}
\label{eq:adaptive_thresh_requirement}
\sum_{i=1}^n \min\left(1, m'\cdot {\bv{a}_i^2}/{\|\bv{a}\|_2^2}\right) = m.
\end{align}
We can find $m'$ that satisfies this equation using a straightforward iterative method. In particular, observe that the expression $\sum_{i=1}^n \min\left(1, m'\cdot {\bv{a}_i^2}/{\|\bv{a}\|_2^2}\right)$
is an increasing function in $m'$, and that it evaluates to $\leq m$ for $m' = m$. So, we can initialize $m' = m$ and iteratively increase its value until the equation is satisfied.
A procedure for doing so is shown in \Cref{alg:threshold_sampling_adapt_computation}. The procedure takes advantage of the fact that any $m'$ satisfying \eqref{eq:adaptive_thresh_requirement} must also satisfy:
\begin{align}
\label{eq:mprime_recur}
m' = \frac{m-|\mathcal{C}|}{\sum_{i\notin \mathcal{C}} {\bv{a}_i^2}/{\|\bv{a}\|_2^2}},
\end{align}
where $\mathcal{C}$ is the set of indices for which $m'\cdot \bv{a}_i^2/\|\bv{a}\|_2^2 \geq 1$. Since we do not know $\mathcal{C}$ in advance, we simply set it based on our current guess for $m'$ (Line 3 of \Cref{alg:threshold_sampling_adapt_computation}). Then we set $m'$ as in \eqref{eq:mprime_recur}, and only terminate if \eqref{eq:adaptive_thresh_requirement} is satisfied.
It can be checked that the loop in \Cref{alg:threshold_sampling_adapt_computation} terminates in at most $m$ steps, after which \eqref{eq:adaptive_thresh_requirement} will be exactly satisfied. In particular, $\mathcal{C}$ must increase in size for the while condition to evaluate to true, and the size of $\mathcal{C}$ is at most $m$, so it increases at most $m$ times. Furthermore, we can presort the values in the vector $\bv{a}$ by their squared magnitude, which allows us to update $\mathcal{C}$ in no more than $O(m)$ time total across all iterations of the loop. We conclude that \Cref{alg:threshold_sampling_adapt_computation} runs in $O(N\log N)$ time (nearly linear) for a vector $\bv{a}$ with $N$ non-zero entries. 

\begin{algorithm}[b]
	\caption{Adaptive Threshold Selection}
	\label{alg:threshold_sampling_adapt_computation}
	\begin{algorithmic}[]
		\Require Length $n$ vector $\bv{a}$, target sketch size $m \leq N$, where $N$ is the number of non-zero values in $\bv{a}$.
		\Ensure Value $m' \geq m$ such that $\sum_{i=1}^n \min\left(1, m'\cdot {\bv{a}_i^2}/{\|\bv{a}\|_2^2}\right) = m$.
		\algrule
		\State Initialize set $\mathcal{C} = \emptyset$ and $m' = m$.
		\While{$\sum_{i=1}^n \min\left(1, m'\cdot {\bv{a}_i^2}/{\|\bv{a}\|_2^2}\right) \leq m$}
		\State $\mathcal{C} \leftarrow \{i\in \{1, \ldots, n\}: m'\cdot {\bv{a}_i^2}/{\|\bv{a}\|_2^2} \geq 1\}$.
		\State $m' \leftarrow  \frac{m-|S|}{\sum_{i\notin S} {\bv{a}_i^2}/{\|\bv{a}\|_2^2}}$.
		\EndWhile
		\State \Return $m'$
	\end{algorithmic}
\end{algorithm}
\vspace{1em}

Once $m'$ is selected, the only additional modification we need to the Threshold Sampling method (\cref{alg:threshold_sampling}) is that $\tau_{\bv{a}}$ in Line 6 should be set to $m'/\|\bv{a}\|_2^2$ instead of  $m/\|\bv{a}\|_2^2$. \Cref{alg:inner_product_estimator} can be used unmodified to estimate inner products and the proof of \cref{thm:main} goes through unchanged, except that our variance will scale with $\frac{1}{m'}$ instead of $\frac{1}{m}$ (which is better since $m' \geq m$).

\myparagraph{High Probability Sketch Size Bound} As shown above, we can easily adjust Threshold Sampling to have expected sketch size exactly equal to $m$. Moreover, we note that both for the original method and the method with an adaptively chosen threshold, it is possible to make a stronger claim: the sketch will have a size smaller than $m + O(\sqrt{m})$ with high probability. Formally,
\begin{lemma}\label{thm:high_prob_size}
Let $\mathcal{S}(\bv a)=\{K_{\bv{a}}, V_{\bv{a}}, \tau_{\bv{a}}\}$ be a sketch of $\bv{a}\in  \R^n$ returned by \Cref{alg:threshold_sampling}, possibly with thresholds set to $\tau_i = m'\cdot \bv{a}_i/\|\bv{a}\|_2^2$, where $m'$ is chosen using \Cref{alg:threshold_sampling_adapt_computation}. We have that
$\E\left[|K_\bv{a}|\right] \leq m$ and, for any $\delta \in (0,1)$,  with probability $1-\delta$, $|K_\bv{a}| \leq m + \sqrt{m/\delta}$ 
\end{lemma}
\begin{proof}
    Observe that $|K_\bv{a}| = \sum_{i=1}^n \mathbbm{1}\left[i\in K_{\bv{a}}\right]$ and  $\E[|K_\bv{a}|]\leq m$, with equality if $m'$ is selected using \Cref{alg:threshold_sampling_adapt_computation}. 
    Since
$\mathbbm{1}\left[i\in K_{\bv{a}}\right]$ and $\mathbbm{1}\left[j\in K_{\bv{a}}\right]$ are independent for $i\neq j$, we have $\Var\left[\left|K_{\bv{a}}\right|\right] =  \sum_{i=1}^n \Var\left[\mathbbm{1}\left[i\in K_{\bv{a}}\right]\right]$. It follows that, since the variance of a $0$-$1$ random variable is upper bounded by its expectation,
\begin{align*}
    \Var\left[\left|K_{\bv{a}}\right|\right]\leq \sum_{i=1}^n \E\left[\mathbbm{1}\left[i\in K_{\bv{a}}\right]\right] = \E[|K_\bv{a}|]\leq m.
\end{align*}
With the expectation and variance bounded, we can apply Chebyshev's inequality to conclude that for any $\delta \in (0,1)$,
\begin{align}
\label{eq:high_prob_bound}
    \Pr\left[\left|K_{\bv{a}}\right| \leq m + \sqrt{m/\delta}\right] \leq \delta.
\end{align}
For example, setting $\delta = 1/100$, we have that with probability $99/100$, $\left|K_{\bv{a}}\right| \leq m + 10\sqrt{m}$. For large $m$, we expect the $10\sqrt{m}$ term to be lower order, and certainly, the sketch size is less than $O(m)$.
\end{proof}

\subsection{Alternative Sampling Probabilities for Threshold and Priority Sampling}
\label{app:alternative_sampling_prob}
Our instantiations of Threshold Sampling and \prioritysampling for the inner product sketching problem (\cref{alg:threshold_sampling,alg:priority_sampling}) sample entries from a vector $\bv{a}$ with probabilities proportional to their \emph{squared magnitudes}, i.e., proportional to $\bv{a}_i^2/\|\bv{a}\|_2^2$ for the $i^\text{th}$ entry. This choice of probability is key in proving the strong worst-case error bounds of \cref{thm:main,thm:main_priority}. However, it is also possible to implement variants of Threshold and Priority Sampling using other probabilities. For example, entries can be sampled uniformly, or with probability proportion to the non-squared magnitude, $|\bv{a}_i|/\|\bv{a}\|_1$, where $\|\bv{a}\|_1 = \sum_{i=1}^n |\bv{a}_i|$ denotes the $\ell_1$-norm of $\bv{a}$. 
We discuss how to modify our methods to accommodate alternative sampling probabilities below. The described variants of Threshold and Priority sampling are tested in our experiments (\cref{sec:experiments}). Importantly, Threshold Sampling with probabilities proportional to $|\bv{a}_i|/\|\bv{a}\|_1$ is exactly equivalent to the `End-Biased Sampling' method proposed in \cite{EstanNaughton:2006}, and Priority Sampling with uniform probabilities is equivalent to the augmented KVM sketch considered in \cite{BessaDFMMSZ:2023}.

\myparagraph{\thresholdsampling with General Probabilities} 
Suppose we want to implement Threshold sampling with any list of sampling probabilities $p_1(\bv{a}), \ldots, p_n(\bv{a})$ that sum to $1$ and are a function of the input vector $\bv{a}$. For example, we might choose $p_i(\bv{a}) = {|\bv{a}_i|}/{\|\bv{a}\|_1}$ or, if $\bv{a}$ has $N$ non-zero entries, $p_i(\bv{a}) = 1/{N}$ for any $i$ where $\bv{a}_i \neq 0$ and $p_i(\bv{a}) = 0$ otherwise (uniform sampling). 
We can modify \Cref{alg:threshold_sampling} for such probabilities by simply setting $\tau_i = m\cdot p_i(\bv{a})$ on Line 3. After this adjustment, to obtain an unbiased estimate $W$ for the inner product -- i.e., with $\E[W] = \langle \bv{a}, \bv{b} \rangle$ -- our estimation method, \Cref{alg:inner_product_estimator}, needs to return: 
  \begin{align*}
  W = \sum_{i \in \mathcal{T}} \frac{\bv{a}_i\bv{b}_i}{\min(1, m \cdot p_i(\bv{a}), m \cdot p_i(\bv{b}))}.
  \end{align*}
  Doing so requires knowledge of $p_i(\bv{a})$ and $p_i(\bv{b})$, which need to be computed based on the sketch. For example, for uniform sampling, $p_i(\bv{a}) = 1/N_\bv{a}$ and $p_i(\bv{b}) = 1/N_\bv{b}$, where $N_\bv{a}$ and $N_{\bv{b}}$ are the number of non-zero entries in $\bv{a}$ and $\bv{b}$, respectively. $N_\bv{a}$ and $N_{\bv{b}}$ can be included in the sketches $\mathcal{S}(\bv{a})$ and $\mathcal{S}(\bv{b})$ for use at estimation time. For non-squared magnitude sampling (i.e., End-Biased Sampling), we need to compute $p_i(\bv{a}) = {|\bv{a}_i|}/{\|\bv{a}\|_1}$ and $p_i(\bv{b}) = {|\bv{b}_i|}/{\|\bv{b}\|_1}$. This can be done at estimation time for any $\bv{a}_i \in \mathcal{S}(\bv{a})$ or $\bv{b}_i \in \mathcal{S}(\bv{b})$ as long as $\|\bv{a}\|_1$ and $\|\bv{b}\|_1$ are included in our sketches.

\myparagraph{\prioritysampling with General Probabilities} 
We can modify Priority Sampling in a similar way. Again, suppose that want to sample $\bv{a}$ with probabilities $p_1(\bv{a}), \ldots, p_n(\bv{a})$ that sum to $1$. We adjust \Cref{alg:priority_sampling} at Line 2 by setting $R_i = h(i)/p_i(\bv{a})$. Then, to obtain an unbiased estimate $W$, \Cref{alg:inner_product_estimator}, needs to return: 
  \begin{align*}
  W = \sum_{i \in \mathcal{T}} \frac{\bv{a}_i\bv{b}_i}{\min(1, p_i(\bv{a})\cdot \tau_\bv{a}, p_i(\bv{b})\cdot \tau_\bv{b})},
  \end{align*}
  where $\tau_{\bv{a}}$ and $\tau_{\bv{b}}$ are the values set on Line 3 of \cref{alg:priority_sampling}. As for \thresholdsampling, we may need to include auxiliary information (typically a single number) in $\mathcal{S}(\bv{a})$ and $\mathcal{S}(\bv{b})$ to compute $p_i(\bv{a})$ and $p_i(\bv{b})$ at estimation time for all $i\in \mathcal{T}$.

\subsection{Faster \wmh using \dartmh}
\label{subsec:fast_wmhs}
One of the baselines we compare our Threshold and Priority Sampling methods against is the recently introduced \wmh method from \cite{BessaDFMMSZ:2023}. 
The original implementation of this method produces a sketch by computing the smallest hash value amongst all non-zero indices in an expansion of the vector $\bv{a}$ being sketched. As discussed, this procedure is slow, requiring $O(Nm\log n)$ time to produce a sketch of size $m$ when $\bv{a}$ has $N$ non-zero entries. 
As discussed in \cref{sec:experiments}, \wmh can be accelerated using the recent \dartmh method from \cite{Christiani:2020}, which is designed to collect $m$ weighted MinHash samples in just $O(N + m\log m)$ time. 
However, integrating \dartmh into \wmh for inner product sketching is non-trivial, since the hash values computed when sampling from $\bv{a}$ and $\bv{b}$ serve dual purposes. First, they facilitate coordinated sampling. Second, they are used to estimate the \emph{weighted union size} between $\bv{a}$ and $\bv{b}$, which is defined as $U = \sum_{i=1}^n \max(\bv{a}_i^2,\bv{b}_i^2)$. $U$ plays a crucial role since it is used as a normalization factor in the inner product estimator. Since \dartmh collects weighted MinHash samples using an entirely different method (which does not explicitly hash all non-zero indices of $\bv{a}$ and $\bv{b}$) we cannot estimate $U$ using the same method as in \cite{BessaDFMMSZ:2023}, which is based on a classic hashing-based distinct elements estimator \cite{Bar-YossefJayramKumar:2002, BeyerHaasReinwald:2007}. 

Fortunately, $U$ can be estimated from the sketch produced by \dartmh in an alternative way. In particular, when sketching $\bv{a}$ and $\bv{b}$, the method returns of set of $m$ ``ranks'' for a sketch of size $m$. If we let $\mathcal{W}$ denote the smaller of the $i^\text{th}$ rank in the sketch for $\bv{a}$ and the sketch for $\bv{b}$, then we have the relation:
\begin{align*}
	\frac{1}{U} (1-e^{-m \log(m)}) \leq \E\left[\mathcal{W}\right] \leq \frac{1}{U}  
\end{align*}
Since the $e^{-m \log(m)}$ term is negligible  for reasonable choices of sketch size $m$, we can use $1/\mathcal{W}$ as an estimate of $U$.

\subsection{Optimized Methods for Join-Correlation Estimation}
\label{app:correlation}

In \Cref{sec:applications}, we introduce a technique for reducing the problem of join-correlation estimation to inner product estimation. The approach requires constructing inner product sketches for three vectors, $\bv{a}$, $\bv{a^2}$, and $\bv{1_a}$. So, if we have a sketch size budget of $m$, we must divide this budget among all three vectors. The easiest way is to do so evenly, so each vector is compressed to a sketch of size $m/3$. The resulting reduction in effective sketch size hurts the performance of methods like JL and CountSketch for join-correlation estimation. However, the issue can be partially mitigated for sampling-based sketches like Threshold and Priority Sampling. 

In particular, in a sampling-based sketch, if we select index $i$ when sketching \emph{any} of the three vectors $\bv{1_a}$, $\bv{a}$, and $\bv{a^2}$, then we might as well use the index in estimating inner products involving \emph{all} three. To do so, instead of computing independent sketches $\mathcal{S}(\bv{a})$, $\mathcal{S}(\bv{a}^2)$, and $\mathcal{S}(1_\bv{a})$, we construct a single \emph{global sketch}, which we denote by $\mathcal{S}_G(\bv{a})$. This sketch is used to derive sketches for all three vectors. As in our standard sampling-based sketches, $\mathcal{S}_G(\bv{a})$ consists of a set of indices ${K}_\bv{a}$ and values ${V}_\bv{a}$ from $\bv{a}$.
To obtain a sketch $\mathcal{S}(\bv{a})$ for estimating inner products with $\bv{a}$, we simply leave $\mathcal{S}_G(\bv{a})$ unchanged. To obtain a sketch $\mathcal{S}(\bv{1_a})$ for $\bv{1_a}$ we set all of the values in ${V}_\bv{a}$ equal to $1$, and to obtain a sketch $\mathcal{S}(\bv{a}^2)$ for $\bv{a}^2$, we square all the values in ${V}_\bv{a}$. 
$K_{\bv{a}}$ remains unchanged in all three sketches, meaning that we effectively reuse samples. 

Note that we would typically use \emph{different sampling probabilities} when sketching vectors $\bv{a}$, $\bv{a}^2$, and $\bv{1_a}$ since the relative squared magnitude of entries in these vectors differs. Accordingly, to collect samples for $\mathcal{S}_G(\bv{a})$, we sample index $i$ according to the \emph{maximum} probability it would have received when constructing any of the three sketches independently. Below, we provide details on how to implement this approach for both Threshold and Priority Sampling.

\begin{algorithm}[t]
    \caption{\thresholdsampling for Join-Correlation}
    \label{alg:corr_sketch}
    \begin{algorithmic}[1]
        \Require Length $n$ vector $\bv{a}$, random seed $s$, sketch size param. $m'$.
        \Ensure Sketch $\mathcal{S}_G(\bv{a}) = \{K_{\bv{a}}, V_{\bv{a}}, \tau_\bv{1_a}, \tau_{\bv{a}}, \tau_\bv{a^2}\}$, where $K_{\bv{a}}$ is a subset of indices from $\{1, \ldots, n\}$ and $V_{\bv{a}}$ contains $\bv{a}_i$ for all $i\in K_{\bv{a}}$.
        \algrule
        \State Use random seed $s$ to select a uniformly random hash function $h: \{1,..., n\}\rightarrow [0,1]$. Initialize $K_{\bv{a}}$ and $V_{\bv{a}}$ to be empty lists.
        \For{$i$ such that $\bv{a}[i]\neq 0$}
            \State Set $\tau_i(\bv{a}) = m'\cdot \frac{\bv{a}_i^2}{\|\bv{a}\|_2^2}$, $\tau_i(\bv{1_a}) = m'\cdot \frac{1}{\|\bv{1_a}\|_2^2}$, $\tau_i(\bv{a^2}) = m'\cdot \frac{\bv{a}_i^4}{\|\bv{a^2}\|_2^2}$.
            \State Set threshold $T_i(\bv{a}) = \max\left(\tau_i(\bv{1_a}), \tau_i(\bv{a}), \tau_i(\bv{a}^2)\right)$.
                \If{$h(i) \leq T_i(\bv{a})$}
              \State Append $i$ to $K_{\bv{a}}$, append $\bv{a}_i$ to $V_{\bv{a}}$.
              \EndIf
        \EndFor
        \State \Return $\mathcal{S}(\bv{a}) = \{K_{\bv{a}}, V_{\bv{a}}, \tau_\bv{1_a}, \tau_{\bv{a}}, \tau_\bv{a^2}\}$, where $\tau_{\bv{a}} = m'/\|\bv{a}\|_2^2$, $\tau_{\bv{1_a}} = m'/\|\bv{1_a}\|_2^2$ and $\tau_{\bv{a^2}} = m'/\|\bv{a^2}\|_2^2$.
    \end{algorithmic}
\end{algorithm}

\myparagraph{Join-Correlation Estimation with \thresholdsampling}
Our optimized Threshold Sampling sketch for join-correlation is shown in \cref{alg:corr_sketch}. To sample the $i^\text{th}$ index according to its maximum importance in each of $\bv{a}$, $\bv{a}^2$, and $\bv{1_a}$, on Line 3 we select a ``global threshold'' for the index, $T_i(\bv{a})$, equal to:
\begin{align*}
    T_i(\bv{a}) &= \max\left(\tau_i(\bv{a}), \tau_i(\bv{a}^2), \tau_i(\bv{1_a})\right), 
\end{align*}
where $\tau_i(\bv{v})$ denote $\tau_i(\bv{v}) = m'\cdot \bv{v}_i^2/\|\bv{v}\|_2^2$ for a vector $\bv{v}$. The index is then sampled exactly as in our standard Threshold Sampling method for inner products -- we check if a uniformly random hash value $h(i)$ falls below $T_i(\bv{a})$, and add $(i, \bv{a}_i)$ to the sketch if it does. 

With the above strategy in place, we can estimate any of the six inner products needed to compute the join-correlation:
$\langle \bv{1_a}, \bv{1_b} \rangle$, $\langle \bv{a}, \bv{1_b} \rangle$, $\langle \bv{1_a},\bv{b} \rangle$, $\langle \bv{a}, \bv{b} \rangle$, $\langle \bv{a^2}, \bv{1_b} \rangle$ and, $\langle \bv{1_a}, \bv{b^2} \rangle$ (see \cref{sec:applications}).
To ensure that our estimates are unbiased, 
we must normalize by the inverse of the probability that entry $i$ is included both in $\mathcal{S}_G(\bv{a})$ and $\mathcal{S}_G(\bv{b})$.
This probability is equal to $\min(1, T_i(\bv{a}), T_i(\bv{b}))$
which can be computed at estimation time based on the content of $\mathcal{S}_G(\bv{a})$ and $\mathcal{S}_G(\bv{b})$, as long as we include in the sketches the three additional numbers $m'/\|\bv{a}\|_2^2$, $m'/\|\bv{a}^2\|_2^2$, and $m'/\|\bv{1_a}\|_2^2$. For example, to approximate $\langle \bv{a}, \bv{b}\rangle$, we compute $\mathcal{T} = K_{\bv{a}} \cap K_{\bv{b}}$ and return the estimate:
\begin{align*}
      W_{\langle \bv{a}, \bv{b}\rangle} = \sum_{i \in \mathcal{T}} \frac{\bv{a}_i\cdot \bv{b}_i}{\min\left(1, T_i(\bv{a}), T_i(\bv{b}\right))}.
\end{align*}
Or to approximate $\langle \bv{a}, \bv{1_b}\rangle$, we return:
\begin{align*}
      W_{\langle \bv{a}, \bv{1_b}\rangle} = \sum_{i \in \mathcal{T}} \frac{\bv{a}_i\cdot 1}{\min\left(1, T_i(\bv{a}), T_i(\bv{b}\right))}.
\end{align*}
Most generally, for functions $f(\bv{a})$ and $g(\bv{b})$ we estimate: 
\begin{align*}
      W_{\langle f(\bv{a}), g(\bv{b})\rangle} = \sum_{i \in \mathcal{T}} \frac{f(\bv{a}_i)\cdot g(\bv{b}_i)}{\min\left(1, T_i(\bv{a}), T_i(\bv{b}\right))}.
\end{align*}

The size of the \thresholdsampling sketch is random, and not deterministically bounded. However, as discussed in \Cref{app:adaptive_and_high_prob}, we ideally want to construct a sketch with expected size \emph{exactly} equal to some specified constant $m$. This can be accomplished using an approach essentially identical to \Cref{alg:threshold_sampling_adapt_computation} in \Cref{app:adaptive_and_high_prob}. In particular, for \Cref{alg:corr_sketch}, we have that:
\begin{align*}
    \E\left[\left|K_{\bv{a}}\right|\right] = \sum_{i=1}^n \E\left[\mathbbm{1}\left[i\in K_{\bv{a}}\right] \right] &= \sum_{i=1}^{n} \min(1,T_i(\bv{a})),
\end{align*}
where $T_i(\bv{a}) = \max\left(m'\bv{a}_i^2/\|\bv{a}\|_2^2, m'/\|\bv{1_a}\|_2^2,  m'\bv{a}_i^4/\|\bv{a}^2\|_2^2\right)$. $\E\left[\left|K_{\bv{a}}\right|\right]$ is an increasing function in $m'$, and we can check that $\E\left[\left|K_{\bv{a}}\right|\right] \leq m$ when $m' = m/3$. So, before sampling, we start at $m' = m/3$ and iteratively increase $m'$ as in \Cref{alg:threshold_sampling_adapt_computation} until $\E\left[\left|K_{\bv{a}}\right|\right] = m$. Then  we run \Cref{alg:corr_sketch} with parameter $m'$.

\begin{algorithm}[t]
    \caption{\prioritysampling for Join-Correlation}
    \label{alg:corr_sketch_priority sampling}
    \begin{algorithmic}[1]
        \Require Length $n$ vector $\bv{a}$, random seed $s$, sketch size param. $m'$.
        \Ensure Sketch $\mathcal{S}_G(\bv{a}) = \{K_{\bv{a}}, V_{\bv{a}}, \tau_\bv{1_a}, \tau_{\bv{a}}, \tau_\bv{a^2}\}$, where $K_{\bv{a}}$ is a subset of indices from $\{1, \ldots, n\}$ and $V_{\bv{a}}$ contains $\bv{a}_i$ for all $i\in K_{\bv{a}}$.
        \algrule
        \State Use random seed $s$ to select a uniformly random hash function $h: \{1,..., n\}\rightarrow [0,1]$. Initialize $K_{\bv{a}}$ and $V_{\bv{a}}$ to be empty lists.
        \State Compute rank $R_\bv{a}(i) = \frac{h(i)}{\bv{a}_i^2}$, $R_\bv{1_a}(i) = h(i)$, and $R_\bv{a^2}(i) = \frac{h(i)}{\bv{a}^4_i}$ for all $i$ such that $\bv{a}_i\neq 0$.
        \State Assign the $(m'+1)^{\text{st}}$ smallest values from $R_\bv{a}(i)$, $R_\bv{1_a}(i)$, and $R_\bv{a^2}(i)$ to $\tau_{\bv{a}}$, $\tau_{\bv{1_a}}$, and $\tau_{\bv{a^2}}$ respectively. If $\bv{a}$ contains fewer than $m+1$ non-zero values, set of all these variables to $\infty$.
        \For{$i$ such that $\bv{a}_i \neq 0$}
                \If{$h(i) < \max(\tau_{\bv{a}} \bv{a}^2_i, \tau_{\bv{1_a}}, \tau_{\bv{a^2}} \bv{a}^4_i)$}
              \State Append $i$ to $K_{\bv{a}}$, append $\bv{a}_i$ to $V_{\bv{a}}$.
              \EndIf
        \EndFor
        \State \Return $\mathcal{S}(\bv{a}) = \{K_{\bv{a}}, V_{\bv{a}}, \tau_\bv{1_a}, \tau_{\bv{a}}, \tau_\bv{a^2}\}$.
    \end{algorithmic}
\end{algorithm}

\myparagraph{Join-Correlation Estimation with \prioritysampling}
We can modify Priority Sampling in a manner similar to \thresholdsampling. To compute a single global sketch $\mathcal{S}_G(\bv{a})$ with entries sampled by their maximum important in $\bv{a}$, $\bv{a}^2$, and $\bv{1_a}$, we compute three different ranks, $R_\bv{a}(i)$, $R_\bv{a^2}(i)$,  and  $R_\bv{1_a}(i)$ for each index $i$. As shown in \Cref{alg:corr_sketch_priority sampling}, we then sample all indices $i$ that rank within the $m'$ smallest of $R_\bv{a}(i)$, $R_\bv{a^2}(i)$,  or $R_\bv{1_a}(i)$. As before, to obtain a unbiased estimate from the sketches, we need to know the probability that entry $i$ gets included both in $\mathcal{S}_G(\bv{a})$ and $\mathcal{S}_G(\bv{b})$. For $i\in K_{\bv{a}}\cap K_{\bv{b}}$ , this probability is equal to 
\begin{align*}
    p_i = \min\left(1, \max\left(\tau_{\bv{1_a}}, \bv{a}_i^2 \tau_{\bv{a}}, \bv{a}_i^4 \tau_{\bv{a^2}}\right), \max\left(\tau_{\bv{1_b}}, \bv{b}_i^2 \tau_{\bv{b}}, \bv{b}_i^4 \tau_{\bv{b^2}}\right)\right), 
\end{align*}
where $\tau_{\bv{1_a}}, \tau_{\bv{a}}, \tau_{\bv{a^2}}$ and $\tau_{\bv{1_b}}, \tau_{\bv{b}}, \tau_{\bv{b^2}}$ are the values computed on Line 3 of \Cref{alg:corr_sketch_priority sampling} and included in $\mathcal{S}_G(\bv{a})$ and $\mathcal{S}_G(\bv{b})$, respectively.

Accordingly, to obtain an unbiased estimate for the inner product $\langle f(\bv{a}), g(\bv{b})\rangle$ for any scalar-valued functions applied entrywise to $\bv{a}$ and $\bv{b}$, we compute $\mathcal{T} = K_{\bv{a}} \cap K_{\bv{b}}$ and return:
\begin{align*}
      W_{\langle f(\bv{a}), g(\bv{b})\rangle} = \sum_{i \in \mathcal{T}} \frac{f(\bv{a}_i)\cdot g(\bv{b}_i)}{p_i}.
\end{align*}
Note that, since it returns the $m'$ smallest indices according to three different rank functions, the Priority Sampling procedure in \Cref{alg:corr_sketch_priority sampling} could return up to $3\cdot m'$ indices, 
in which case we should set $m' = m/3$ to obtain a global sketch  $\mathcal{S}_G(\bv{a})$ of size $m$. However, often there will be overlap between the indices that minimize these rank functions: for example, if $\bv{a}_i$ has a large magnitude, then we expect both $R_{\bv{a}}(i)$ and $R_{\bv{a}^2}(i)$ to be small. As a result, if we set  $m' = m/3$ the method often returns a sketch with fewer than $m$ index/value pairs, which is wasteful. To obtain a sketch whose size exactly matches our budget $m$, again we need to initialize $m'$ to $m/3$ and then increase the value of $m'$ until \Cref{alg:corr_sketch_priority sampling} selects exactly $m$ indices. This can be done via a standard binary search: we need only consider integer values of $m'$ between $m/3$ and $m$ (since setting $m'=m$ results in a sketch of size \emph{at least} $m$).

\balance
\end{document}